\documentclass[lettersize,journal]{IEEEtran}
\usepackage{amsmath,amsfonts,amssymb,amsthm,thmtools}
\usepackage{algorithmic}
\usepackage{algorithm}
\usepackage{array}
\usepackage[caption=false,font=normalsize,labelfont=sf,textfont=sf]{subfig}
\usepackage{textcomp}
\usepackage{stfloats}
\usepackage{url}
\usepackage{verbatim}
\usepackage{graphicx}
\usepackage{cite}
\usepackage{color}
\usepackage{xcolor}
\definecolor{ccr}{RGB}{10,121,183} 
\definecolor{kkk}{RGB}{0,0,0} 
\usepackage[colorlinks,linkcolor=kkk,anchorcolor=ccr,citecolor=ccr,urlcolor=ccr]{hyperref}
\usepackage{titlesec}
\allowdisplaybreaks[4]
\titlespacing{\subsection}{0pt}{1ex plus 0ex minus .2ex}{1ex plus .2ex}
\begin{document}
\title{Sensor Network Localization via Riemannian Conjugate Gradient and Rank Reduction: An Extended Version
}

\author{Yicheng Li, 
	Xinghua Sun,
	~\IEEEmembership{Member,~IEEE}
	\thanks{Yicheng Li and Xinghua Sun are with the School of Electronics and Communication Engineering, Sun Yat-sen University, Shenzhen 518107, China (e-mail:\url{liych75@mail2.sysu.edu.cn};\url{sunxinghua@mail.sysu.edu.cn}).}
	\thanks{This is an extended version of the manuscript ``Sensor Network Localization via Riemannian Conjugate Gradient and Rank Reduction", which has been accepted by IEEE Transactions on Signal Processing. In this extended version, we provide (i) extra numerical evidence for the SNL problem; and (ii) extra references and analyses toward the EDMC problem.}
}
\markboth{}
{Li \MakeLowercase{\textit{et al.}}: Sensor Network Localization via Riemannian Conjugate Gradient and Rank Reduction: An Extended Version}

\maketitle
\begin{abstract}
This paper addresses the Sensor Network Localization (SNL) problem using received signal strength. The SNL is formulated as an Euclidean Distance Matrix Completion (EDMC) problem under the unit ball sample model. Using the Burer-Monteiro factorization type cost function, the EDMC is solved by Riemannian conjugate gradient with Hager-Zhang line search method on a quotient manifold. A ``rank reduction" pre-process is proposed for proper initialization and to achieve global convergence with high probability. Simulations on a synthetic scene show that our approach attains better localization accuracy and is computationally efficient compared to several baseline methods. Characterization of a small local basin of attraction around the global optima of the s-stress function under Bernoulli sampling rule and incoherence matrix completion framework is conducted for the first time. Theoretical result conjectures that the Euclidean distance problem with a structure-less sample mask can be effectively handled using spectral initialization followed by vanilla first-order methods. This preliminary analysis, along with the aforementioned numerical accomplishments, provides insights into revealing the landscape of the s-stress function and may stimulate the design of simpler algorithms to tackle the non-convex formulation of general EDMC problems.  
\end{abstract}


\section{Introduction}
\IEEEPARstart{T}{his} paper investigates multi-hop Sensor Network Localization (SNL) algorithms based on pairwise distance measurements using Received Signal Strength Indicator (RSSI). In such systems, distances between adjacent nodes are estimated through a path loss model and forwarded to the data center, where an incomplete distance map is generated\cite{MaoSNLSurvey}\cite{HanSNLSurvey}. Only a small number of pairwise distances are available at the data center due to the limited range of radio coverage of each sensor. Additionally, obstacles and scatters can cause significant errors in some of these measurements as illustrated in Fig. \ref{fig_IllustSNL_EDM}. Therefore, an appropriate algorithm for the SNL problem should be capable of recovering the absolute position of all nodes when provided with a sampled set of noisy distance measurements and the positions of anchors (nodes with known positions). Moreover, the algorithm should be robust under the presence of outliers in these measurements.
\subsection{Mathematical Setup}
\label{subsec_math_setup}
Before delving into a detailed introduction of prior arts, we pause to discuss two different sampling schemes considered in this paper when generating available distance measurements,
and give a highly simplified mathematical setup of the SNL problem\footnote{Currently we have not considered anchors. And algorithms proposed in this work do not explicitly rely on the cliques structure formed by anchors to recover the relative positions of nodes.}. Assume there are in total $n$ nodes distributed in the $d$ dimensional Euclidean space, where $d=2,3\ll n$. Denote the ground truth distance between $(i,j)$ nodes as $d_{ij}$ and let $\mathbf{Y}:=[\mathbf{p}_1,\dots,\mathbf{p}_n]^T\in\mathbb{R}^{n\times d}$ be all nodes' positions (assume $\mathbf{Y}$ has full column rank). Let the set of partially observed distances be $\{d_{ij}^2|(i,j)\in\Omega\}$, $\Omega\subset\mathbb{I}=[n]^2$, where $\mathbb{I}$ is the complete set\footnote{Since the distance measurements are symmetric, we should set $\mathbb{I}:=\{(i,j):1\leq i\leq j\leq n\}$. Some analyses are based on this model\cite{TasissaEDMCProof}\cite{ZhangSVD_MDS}\cite{DingEDMCrestrictCVX} and this does not pose a fundamental difference.}.
Two sample models are defined as follows:
\begin{itemize}
	\item Unit ball rule: For a given radius $r>0$, $\boldsymbol{\alpha}=(i,j)\in\Omega$ iff $d_{ij}<r$. $r$ is controlled by transmission power, fading environment, and receiver's sensitivity. 
	\item Bernoulli rule: This is inherited from standard analysis in Low-rank Matrix Completion (LRMC)\cite{LRMC_Can1}\cite{DavenportLRMRoverview}. Throughout this paper, $\delta_{\boldsymbol{\alpha}},\,\boldsymbol{\alpha}=(i,j)$ is a sequence of i.i.d. $0/1$ Bernoulli random variables with $\mathbb{P}(\delta_{\boldsymbol{\alpha}}=1)=p$ and set $\Omega=\{(i,j)|\delta_{\boldsymbol{\alpha}}=1\}$. 
\end{itemize}
Consider a matrix $\mathbf{D}\in\mathbb{R}^{n\times n}$ where the $(i,j)$-th entry is $d_{ij}^2$, this is called the Euclidean Distance Matrix (EDM), denoted by $\mathbf{D}\in\mathbb{EDM}^n$. Simple linear algebra reveals that
\begin{equation*}
	\setlength\belowdisplayskip{3pt}
	\setlength\abovedisplayskip{3pt}
	d_{ij}^2=\mathbf{D}_{ij}=\mathrm{tr}(\mathbf{YY}^T\mathbf{z}_{\boldsymbol{\alpha}}\mathbf{z}_{\boldsymbol{\alpha}}^T)=\mathrm{tr}(\mathbf{YY}^T\boldsymbol{\omega}_{\boldsymbol{\alpha}}),
\end{equation*}
where $\mathrm{tr}(\mathbf{A}^T\mathbf{B})=\langle\mathbf{A},\mathbf{B}\rangle$ denotes the inner product of two matrices, $\mathbf{z}_{\boldsymbol{\alpha}}=\mathbf{e}_i-\mathbf{e}_j$ and $\mathbf{e}_i$ is the $i$-th canonical Euclidean basis. Thus $\boldsymbol{\omega}_{\boldsymbol{\alpha}}$ is a rank $1$ symmetric positive semidefinite (PSD) matrix. And we let $\mathbf{G}=\mathbf{YY}^T$ denote the Gram matrix. A simplified SNL problem belongs to the famous set of non-convex Quadratically Constrained Quadratic Programs\cite{LuoNonCVX_QCQP}
\begin{equation}
	\setlength\belowdisplayskip{3pt}
	\setlength\abovedisplayskip{3pt}
	\label{eq_SNL_as_nonCVX_QCQP}
	\begin{aligned}
		\mathrm{find}&\,\mathbf{\mathbf{G}}	\\
		\mathrm{s.t.}&\, \mathbf{D}_{ij}=\mathrm{tr}(\mathbf{G}\boldsymbol{\omega}_{\boldsymbol{\alpha}}),\,\boldsymbol{\alpha}\in\Omega,\\
		&\,\mathbf{G}\succcurlyeq \mathbf{0},\,\mathrm{rank}(\mathbf{G})=d,\,\mathbf{G1}=\mathbf{0}.
	\end{aligned}
\end{equation}
We note that one can only recover $\mathbf{G}$ up to translations. This ambiguity is removed by considering only self-centered $\mathbf{G}$, i.e., $\mathbf{G1}=\mathbf{0}$, where $\mathbf{1}$ is the vector of all ones. 
By dropping the rank constraint, one obtains the Semidefinite Relaxation (SDR) of \eqref{eq_SNL_as_nonCVX_QCQP}\footnote{This SDR is not tight and only used for illustration. Practical solutions contain another list of works thoroughly studied by So et al. that utilize the known positions of anchors. Please see\cite{SoAnthonySNL2006}\cite{BiswasSDRSNLacm}\cite{LuoNonCVX_QCQP} for the actual formulation, conditions of exact recovery, and numerical performance.}. Readers familiar with LRMC may find great similarity between \eqref{eq_SNL_as_nonCVX_QCQP} and either the original matrix completion problem\cite{CandTaoLRMC}\cite{GrossLRMCProof}, or the rank-one quadratic sampling framework\cite{YxChentitQuadraticSampling}\cite{LiYMaNonCVXrankOne} aiming at recovering the Gram matrix $\mathbf{G}$. The major difference is twofold. First, this set of sample basis ($\boldsymbol{\omega}_{\boldsymbol{\alpha}}$ or $\mathbf{z}_{\boldsymbol{\alpha}}$) is neither a coordinate nor a sub-Gaussian random vector\cite[Ch. 3.4.2]{VershyninHigDimbook}. Second, the real-world SNL problem is based on the unit ball sample model, which eliminates the chance of observing large distance measurements, and causes the sample mean to be a biased estimate. These two disparities imply that theories established in the low-rank recovery context cannot be readily applied to the SNL problem. But one can still ask: will minimizing the convex surrogate of rank function, e.g., the nuclear norm (or $\mathrm{tr}(\mathbf{G})$ since $\mathbf{G}\succcurlyeq \mathbf{0}$), still work? If so, then how large is the minimum sample complexity? Moreover, can one find a similar theoretical guarantee of the non-convex Burer-Monteiro factorization\cite{BurerMonteiroFR} approach as in \cite{ZhengLaffertyNonCVXFR}\cite{SL15NonCVXFR}\cite{MaImplicitRegularNonCVX_GD}\cite{RongGeSpuriousLocalMinima} when applied to the SNL problem? The answer to the first two questions is affirmative and well-known from both practical and theoretical perspectives while the last, to the best of the authors' knowledge, remains mystery.
\begin{figure}[!t]
	\centering
	\setlength{\abovecaptionskip}{0cm}
	\setlength{\belowcaptionskip}{-3cm}
	\includegraphics[width=8.5cm]{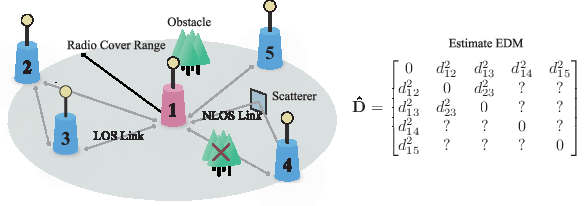}
	\caption{An illustration of pairwise distance estimations in a sensor network (Left) and the corresponding distance map in an EDM style (Right). Due to the fading of electromagnetic waves and limited sensor power supply, the radio coverage range of each sensor is a measurable constant (denoted as $r$ throughout this paper). Nodes within distance $r$ share an available link, which is represented by two-sided gray arrow, thus the inter-node distance measurement between them can be obtained. Obstacles and scatters can cause outliers in the estimate EDM elements, says in $d_{14}$.}	
	\label{fig_IllustSNL_EDM}
\end{figure}
\subsection{Related Work}
The SNL problem has been extensively studied in various contexts, and framing it as the Euclidean Distance Matrix Completion (EDMC) problem with a specific sampling scheme might be the most popular choice\cite{LibertiEDMSurvey}. Since completing a hollow, symmetric matrix (as shown in Fig. \ref{fig_IllustSNL_EDM}, right) given a subset of its entities into an EDM well captures part of the ill-posed nature of SNL: the distance between any user node to at least $d+1$ anchors might be missing, preventing a direct use of triangulation. Moreover, if the EDM is complete and noiseless, then classical Multi-dimensional Scaling (cMDS)\cite{MDSSurvey} can recover the global position of the network given anchors' positions. 
A fundamental problem along this line is to determine under which condition the solution is unique. Eren et al.\cite{NetworkLoc1} addressed this using graph rigidity theory. They concluded that generic global rigidity\cite{GortlerGlobalRigidity} is the key indicator for a globally localizable network, which can persevere its shape under continuous and discontinuous transformations in its embedding dimension. Alfakih et al.\cite{AlfakihuniquenessEDMC} showed sufficient and necessary conditions to ensure that the EDMC problem admits a unique solution. These prerequisites play an analogous role as the information theory lower bound in LRMC\cite[Thm. 1.7]{CandTaoLRMC} (please also see\cite{SingerTheBound}). As for algorithm side, EDMC techniques discussed in this paper can be roughly divided into three groups: (i) the spectral methods SVD-MDS\cite{DrineasSVD_MDS} and MAP-MDS\cite{ShangYDistMDS}; (ii) the trace heuristic\cite{BiswasTR_sdp}\cite{JavanmardMontanariTrEDMC}\cite{TasissaEDMCProof}; (iii) the non-convex approach based on Riemannian optimization\cite{MishraRieEDMC}\cite{ParhizkarOptSpace}\cite{NguyenLRM-CG} or modified Newton method\cite{FangEDMCNewton}. We refer the interested readers to\cite{DokmanicEDMTheory}\cite{LibertiEDMSurvey} for broader classes of algorithms and different implementation scenarios of the EDMC problem, including but not limited to channel charting\cite{AgostiniChannelCEDMC}\cite{ChanChartStuder}, cell genome reconstruction\cite{PaulsenCell3DGenomeReconstruction}\cite{WangY3DGenomeFLAMINGO}, Inverse Kinematics problem\cite{MaricRieOptIKRobot}, solving partial differential equations on manifolds\cite{LaiPDEmanifolds}, Simultaneous Localization And Mapping\cite{TabaghiKineticEDM}.

The MAP-MDS mainly consists of two steps, first, estimate the missing distances using the shortest path between nodes to obtain an approximate EDM $\hat{\mathbf{D}}$, and second, apply cMDS to $\hat{\mathbf{D}}$. Its performance under unit ball sample rule was thoroughly analyzed in\cite{KarbasiOhMAPMDSTri}, where nodes are assumed to be distributed uniformly in $[0,1]^d$ hypercube. The SVD-MDS uses exactly the same Singular Value Decomposition (SVD) reconsecration on the incomplete EDM as in LRMC\cite[Sec. 3.8]{ChenSpectralMethods} since any EDM has rank no larger than $d+2$\cite{DokmanicEDMTheory}, then followed by cMDS.\cite{ZhangSVD_MDS} provides an in-depth theoretical study of this process using Bernoulli sample model. The MAP-MDS has a severe drawback. If the nodes are distributed via irregular manners, e.g., on a Swiss roll in $\mathbb{R}^3$, then the shortest path algorithm can only approximate the geodesic distance on these manifolds but not the Euclidean one when radio coverage range is limited\cite{AriasCastro3Pertubound}. This was partially addressed in later work by Shang and Ruml\cite{ImprovedMDSShang} through a rather complicated scheme, and only numerical success was obtained. 

The trace heuristic might be the most celebrated success of the EDMC problem, and it can be divided into two ``conflicting" ideas inherited from Maximum Variance Unfolding (MVU)\cite{WeinbergerMVU} in manifold learning and Nuclear Norm Minimization (NNM)\cite{LRMC_Can1} in LRMC, respectively. Both algorithms can be viewed as the regularized version of SDR of \eqref{eq_SNL_as_nonCVX_QCQP}. \cite{BiswasTR_sdp} introduced the idea of maximizing the total variance of the point set, i.e., $\mathrm{tr}(\mathbf{G})$,
to the SNL problem. 
While spectacular work by Javanmard and Montanari\cite{JavanmardMontanariTrEDMC} paralleled the idea of NNM to minimize the trace. By using a random geometry model, i.e., all nodes distribute uniformly in the unit hypercube, they were able to give explicit sample size lower bound to the SNL problem. Moreover, \cite[Cor. 3.1]{JavanmardMontanariTrEDMC} states that under noiseless conditions, if the radio outage range $r$ is larger than a typical threshold, then the network constructed by the unit ball rule is globally rigid with high probability, and can be exactly recovered by trace minimization. 
The MVU was first analyzed in\cite{DingEDMCrestrictCVX} under uniform sampling rule by resorting to a distinct definition of restricted strong convexity developed in weighted matrix completion\cite{weightedMCNegahban}. Dropping the incoherence condition\cite{ChenIncoOptimalMC} enables them to treat the EDMC as ``sampling w.r.t. the EDM" directly, relying minor on the Gram matrix while losing the exact recovery insurance even in noiseless situation. Recent work by Tasissa and Lai\cite{TasissaEDMCProof} generalized the inexact dual certificate and golfing scheme originally developed in\cite{GrossLRMCProof} for the LRMC to EDMC settings\footnote{Actually their techniques also work under general non-orthonormal basis, which is not too far from an orthonormal basis. Please see\cite{TasissaLRMC_genebasis} for details.}. Under the set of non-orthonormal basis $\boldsymbol{\omega}_{\boldsymbol{\alpha}}$
with a ``sampled uniformly at random with replacement" scheme, they showed that the sample complexity required by trace minimization in EDMC for exact recovery is of nearly the same order as the incoherence optimal LRMC\cite{ChenIncoOptimalMC}\footnote{We note that\cite{TasissaEDMCProof} needs the joint incoherence assumption w.r.t. $\boldsymbol{\omega}_{\boldsymbol{\alpha}}$. Their analysis follows from\cite{GrossLRMCProof} but not from the $l_{2,\infty}$ method in\cite{ChenIncoOptimalMC}. However, since $d=\mathcal{O}(1)$ in EDMC, their bound can be compared with\cite{ChenIncoOptimalMC}.}. 

While convex algorithms show tolerance towards different sampling schemes\footnote{We note that even though convexification, along with SDP solvers, exhibits superior performance guarantees in finding an approximation of the noisy and incomplete EDM\cite{AlfakihSNLSDPSurvey}, their computational cost often scales at least cubically with the problem size. Developing theoretically guaranteed non-convex algorithms for matrix recovery problems has became a trend in many fields\cite{ChenYXChiLRMCSurvey}.}, the non-convex ones seem to be seriously affected by the unit ball rule. We will detail this in Section \ref{subsecIII-C_Basin}. Characterizing the EDMC problem by means of non-convex lens can be dated back to the 1980s by computational chemists\cite{KearsleySstress}\cite{CrippenEneemb}, as they aimed at restoring the structure of a molecule via inter-atomic distances obtained by Nuclear Magnetic Resonance\cite{HavelCompChemEDMC}. A famous criterion used in this list of works is called the s-stress function\cite{TakaneSStress}, given by
\begin{equation}
	\setlength\belowdisplayskip{3pt}
	\setlength\abovedisplayskip{3pt}
	\label{eq_FNMOptFunc_LRMC}
	\begin{gathered}
		\min_{\mathbf{D}\in \mathbb{EDM}^n} \frac{1}{2} \Vert \mathbf{W}\odot \mathcal{P}_{\Omega}(\mathbf{D}-\mathbf{D}_e) \Vert_F^2, \\
		\left[ \mathcal{P}_{\Omega}(\mathbf{A}) \right]_{ij}=
		\begin{cases}
			\mathbf{A}_{ij}, & {\mathrm{if}} (i,j)\in \Omega \\
			0, & {\mathrm{otherwise.}}
		\end{cases}.
	\end{gathered}
\end{equation}
\eqref{eq_FNMOptFunc_LRMC} has an empirically designed Frobenius Norm Minimization (FNM) structure, where $\mathbb{EDM}^n$ is the set of $n\times n$ EDM, $\Omega$ is the sample set as defined in Section \ref{subsec_math_setup}, $\mathbf{D}_e$ stands for distance estimations which may be corrupted by noise and outliers and $\mathbf{W}$ is a weight matrix to model noisy measurements. 
Restriction $\mathbf{D}\in \mathbb{EDM}^n$ conveys more information than saying $\mathbf{D}$ is a hollow, symmetric matrix has rank $d+2$\cite{DokmanicEDMTheory}. This indicates that when the sample complexity is limited, it is not preferable to use original LRMC techniques to solve \eqref{eq_FNMOptFunc_LRMC} directly, but rather to take EDM properties into account\footnote{We will present numerical evidence towards this statement in Appendix \ref{subsec_EDMC_LRMC_CVXmethod}.}. This is accomplished by an operator $g$ mapping the Gram matrix $\mathbf{G}$ to an EDM (please see Section \ref{subsecII-A_EDM} for detail). Since $\mathbf{G}$ is a PSD matrix, the Burer-Monteiro factorization applies. Non-convex algorithms based on this idea solve the following unconstrained problem
\begin{equation}
	\label{eq_s_stress_gmap_EDMC}
	\setlength\belowdisplayskip{3pt}
	\setlength\abovedisplayskip{3pt}
	\min_{\mathbf{Y}\in\mathbb{R}^{n\times d}} \frac{1}{2}\Vert\mathbf{W}\odot \mathcal{P}_{\Omega}(g(\mathbf{YY}^T)-\mathbf{D}_e)\Vert_F^2.
\end{equation}
\eqref{eq_s_stress_gmap_EDMC} and its variants are sometimes known as gradient refinement after MAP-MDS\cite{ImprovedMDSShang} or Biswas-Ye SDR\cite[Sec. 5]{BiswasSDRSNLacm}, implying that vanilla first-order method seems to enjoy good convergence if initialized carefully. However, the SDP (Semidefinite Programming) solver is time-consuming and the MAP-MDS is not accurate enough to trigger convergence with high probability given an irregularly-shaped network. Having said this, there have been constant efforts on finding good initial points\cite{FangEDMCNewton} and utilizing powerful tools from optimization machinery to quest for satisfactory numerical performance\cite{MishraRieEDMC}\cite{NguyenLRM-CG}. We will detail these three works in Section \ref{subsecIII-D_rankReduction} and \ref{subsecII-A_EDM}. To the best of the authors' knowledge, the only theoretical guaranteed success along the non-convex line is owned by Parhizkar et al.\cite{ParhizkarOptSpace} utilizing OptSpace\cite{KeshavanOptSpace} yet it requires nodes to be distributed on a circular ultrasound probe, which introduces unique constraints on both the EDM and the sampling scheme. Their unpublished work\cite[Ch. 3]{ParhizkarPhDthsis} inspires this work, i.e., the idea of analyzing \eqref{eq_s_stress_gmap_EDMC} using non-convex LRMC framework.

After completing this manuscript, we noticed two interesting works that have appeared on arXiv recently\cite{TangTohYinReiDimreduce}\cite{LeiYinBlessHighOrderSNL}. The ideas behind these works bear a strong resemblance to our approach, namely, the utilization of over-parameterization to solve the SNL problem\footnote{Historically, this idea emerged quite early. Please refer to Section \ref{subsecIII-D_rankReduction} for further discussions.}. Both of these works employ the Biswas-Ye's cost function (or its variants)\cite[Sec. III-B]{BiswasTR_sdp}, which involves separating the distance measurements into user-to-user and user-to-anchor components. Specifically, let $\mathbf{a}_i$ denote the position of $i$-th anchor. \cite{TangTohYinReiDimreduce} suggests that the following problem can be handled by a cubic regularized dimension reduced Riemannian Newton method\footnote{The word ``dimension reduced" here refers to evaluating the Riemannian Hessian only within a specific subspace of the tangent space at the current iteration point. We refer implementation details to their paper.}.
\begin{align}
	\label{eq_TangToh_BiswasYe_Riecons_Cost}
	\setlength\belowdisplayskip{0pt}
	\setlength\abovedisplayskip{0pt}
	\min\, &\{\frac{1}{2}\sum_{(i,j)\in\Omega_1}(\Vert\mathbf{p}_i-\mathbf{p}_j\Vert_2^2-d_{ij}^2)^2-\underbrace{\frac{\lambda}{2n}\sum_{ij}^{[n]^2}\Vert\mathbf{p}_i-\mathbf{p}_j\Vert_2^2}_{\mathrm{mvu}}\}\nonumber\\
	\mathrm{s.t.}\,&\forall (i,k)\in\Omega_2,\Vert\mathbf{p}_i-\mathbf{a}_k\Vert_2=d_{ik},
\end{align}
where $\Omega_1$ contains only user-to-user measurements and $\Omega_2$ contains only user-to-anchor measurements. The $\mathrm{mvu}$ regularization term here follows the same strategy as in\cite{BiswasTR_sdp} to maximize the trace of the point set. The Riemannian aspect arises from the constraint. They proved that the feasible set of \eqref{eq_TangToh_BiswasYe_Riecons_Cost} forms a Riemannian manifold in $\mathbb{R}^d$. They have also demonstrated that the localization accuracy can be further improved by utilizing an over-parameterized stage on \eqref{eq_TangToh_BiswasYe_Riecons_Cost} as a warm start. By employing both the $\mathrm{mvu}$ term and the cubic regularization method, they are able to lift the rank to be sufficiently large to meet the rank upper bound of Biswas-Ye SDR. Subsequently, they selected the first $d$ columns of the higher-rank solution as the warm start\footnote{We note that similar strategy can also be applied to our approach. However, since Algorithm \ref{alg_RankReduct_Initial} does not contain any regularizer, its performance (computational time and accuracy) will be slightly harmed.}. Meanwhile, \cite{LeiYinBlessHighOrderSNL} studies the general landscape of 
\begin{equation}
	\label{eq_LeiYin_BiswasYe_Genland}
	\setlength\belowdisplayskip{3pt}
	\setlength\abovedisplayskip{3pt}
	\sum_{(i,j)\in\Omega_1}|\Vert\mathbf{p}_i-\mathbf{p}_j\Vert_2^b-d_{ij}^b|^c+\sum_{(k,j)\in\Omega_2}|\Vert\mathbf{a}_k-\mathbf{p}_j\Vert_2^b-d_{kj}^b|^c,
\end{equation} 
where $b, c\in\mathbb{Z}^+$. They proved that \eqref{eq_LeiYin_BiswasYe_Genland} is non-convex with high probability in the unit-disk SNL case, even the underlying graph is complete. This result confirms the numerical evidence that designing effective first-order local search algorithms for the SNL problem is challenging. And it is possible for \eqref{eq_LeiYin_BiswasYe_Genland} to show spurious local minimizer under some specific choices of $b,c$\cite[Sec. 3.1]{LeiYinBlessHighOrderSNL}. However, we note that this result is not as informative as the global landscape analysis\cite{ZhuGlobalGeo}\cite{RongGeSpuriousLocalMinima}\cite{LiSymmSaddleLandscape}\cite{ChenYXPRRandominit}\cite{SunQuWrightPRLandScape} or the local attractive basin characterization\cite{ZhengLaffertyNonCVXFR}\cite{SL15NonCVXFR}\cite{MaImplicitRegularNonCVX_GD}\cite{PhaseretrievalWirtinger}\cite{TuProcrustesFlow}\cite{LiYMaNonCVXrankOne} developed in recent context of low-rank recovery.

While the intriguing aspect of these two works is that, they both empirically showed that the accurately parameterized case of the Biswas-Ye's cost function under the unit ball sample model is not highly non-convex. At least, there exists a small region around the ground truth within which gradient descent is expected to converge. 
Additionally, techniques that can ``reshape" the landscape (over-parameterization, low-rank inductive regularization) or escape from saddles (trust-region that utilized the negative curvature condition explicitly\cite[Alg. 3]{BoumalGlobalratesRTRv2}, cubic regularization\cite[Sec. 9.3.1]{ChiYNonCVX_Facor_overview}, stochastic gradient descent\cite[Sec. 5]{ZhangSGDEDMC}) may exhibit good empirical convergence on the SNL problem. To address this long-studied, NP-hard problem using non-convex methods with strong theoretical guarantees, it is crucial to understand the local/global landscape of \eqref{eq_s_stress_gmap_EDMC} under unit ball sampling model\footnote{It is also interesting to investigate the landscape of Biswas-Ye's cost function, e.g., setting $b=c=2$ in \eqref{eq_LeiYin_BiswasYe_Genland}. This seems to be more complex.}.
\begin{table*}[t]
	\renewcommand{\arraystretch}{1.2}
	\vspace{-4pt}
	\caption{List of Commonly Used Abbreviations}
	\label{table_Abbreviations}
	\centering\vspace{-4pt}
	\scalebox{0.95}{\begin{tabular}{c c | c c | c c}
			\hline
			\bfseries Abbreviation & \bfseries Full Name & \bfseries Abbreviation & \bfseries Full Name & \bfseries Abbreviation & \bfseries Full Name \\
			\hline 
			SNL & Sensor Network Localization & LRMC & Low-rank Matrix Completion & EDMC & Euclidean Distance Matrix Completion\\
			MDS & Multi Dimensional Scaling & RCG & Riemannian Conjugate Gradient & VGD & Vanilla Gradient Descent\\
			HZLS & Hager-Zhang Line Search & MVU & Maximal Variance Unfolding & RTR & Riemannian Trust-Region\\
			\hline
	\end{tabular}}\vspace{-10pt}
\end{table*}
\subsection{Major Contributions}
\label{subsec_Major_Contr}
Burer-Monteiro factorization and a Riemannian Conjugate Gradient (RCG) coupled with Hager-Zhang Line Search (HZLS), under quotient geometry $\mathbb{R}^{n\times d}_{*}/\mathrm{O}(d)$ of the set of PSD matrices of fixed rank $d$, i.e., $\mathcal{S}_{+}^{d,n}$, is applied to solve the SNL for numerical superiority. Our major contributions are threefold:
\begin{itemize}
	\item We show how to generalize HZLS to $\mathbb{R}^{n\times d}_{*}/\mathrm{O}(d)$ under two different Riemannian metrics, i.e., the canonical Euclidean inner product and the one first proposed by Mishra et al.\cite{MishraRiegeometryMetric}.
	\item (Main) By adopting the non-convex LRMC framework originally developed in\cite{ZhengLaffertyNonCVXFR}\cite{SL15NonCVXFR}, and combining it with recent results by Tasissa\cite{TasissaEDMCProof}, we characterize an attractive basin of the s-stress cost function \eqref{eq_s_stress_gmap_EDMC}. We show that the regularity condition\cite[Sec. 7.D]{PhaseretrievalWirtinger} will hold in this region as soon as the sample complexity reaches a typical threshold under Bernoulli rule. This result conjectures that EDMC problems under a structure-less sample model can be effectively solved using vanilla first-order methods provided with a good initial point generated by spectral methods.
	\item We propose a simple numerical method called ``rank reduction", inherited and modified from \cite{FangEDMCNewton}\cite{HuangPhaseLift}\cite{GaoRRMC}, to generate empirically good initial points for \eqref{eq_s_stress_gmap_EDMC}. This pre-process is robust to both unit ball sampling scheme and random initialization. Simulation results on a synthetic SNL problem verify its effectiveness when combined with the RCG-HZLS framework. Performance close to the global rigidity lower bound is obtained.
\end{itemize} 
Other contribution may include a Manifold-ADMM\cite{KovnatskyMADMM} (MADMM) based outliers elimination circuit which aims to
deal with Non Line-of-sight (NLoS) distance measurements.

\emph{Organization:} In Section \ref{secIIbackground} we revisit tools to formulate \eqref{eq_FNMOptFunc_LRMC} into non-convex EDMC problems on a quotient manifold. We present our algorithm, i.e., RCG with Riemannian HZLS under $\mathbb{R}^{n\times d}_{*}/\mathrm{O}(d)$ in Section \ref{section_III_RCG_HZLS}. Basin characterization and ``rank reduction" will be detailed at the end of this section. In Section \ref{secIV_OutlierQuestion} we turn to the outliers elimination problem. Section \ref{secV_NumExps} provides numerical results on proposed algorithms. The paper is concluded in Section \ref{sec_conclusion} with a discussion.

\emph{Notions:} \textbf{Bold} face lower letters and capital letters represent vectors or matrices, respectively. $\mathcal{M}$, $\mathbb{R}_{*}^{n\times d}$ and $\mathcal{S}^{n}_{+}$ denote the general Riemannian manifold, the set of full column rank $n\times d$ matrices and the set of $n\times n$ PSD matrices, respectively. $\odot$ is the Hadamard product. $\mathrm{diag}(\mathbf{A})$ and  $\mathrm{diag}(\mathbf{a})$ is the column vector formed by the diagonal elements of $\mathbf{A}$ and the diagonal matrix formed by vector $\mathbf{a}$, respectively. $\mathrm{Skew}(d)$ and $\mathcal{S}(d)$ denote the sets of skew-symmetric and symmetric matrices of size $d\times d$. $\mathrm{Skew}(\mathbf{A})$ is the skew-symmetric part of $\mathbf{A}$. $\mathrm{O}(d)$ represents the set of orthogonal matrices of size $d\times d$. The adjoint of a linear transformation $T$ is denoted by $T^{*}$. $\mathcal{I}$ denotes the identity operator. 
$\sigma_i/\lambda_i(\mathbf{A})$ is the i-th largest singular/eigen value of $\mathbf{A}$. $\Vert\mathbf{a}\Vert_2$ is the vector $l_2$ norm while $\Vert\mathbf{A}\Vert$, $\Vert\mathbf{A}\Vert_F$, $\Vert\mathbf{A}\Vert_1$, $\Vert\mathbf{A}\Vert_{2,\infty}:=\max_{i}\Vert\mathbf{e}_i^T\mathbf{A}\Vert_2$, $\Vert\mathbf{A}\Vert_{\infty}$ stand for the spectral, Frobenius, $l_1$, $l_2/l_{\infty}$, and entrywise $l_{\infty}$ norm of $\mathbf{A}$. $c$ is a positive constant and may differ from line to line. Commonly used abbreviations are summarized in Table \ref{table_Abbreviations}.
\section{BACKGROUND}
\label{secIIbackground}
We briefly revisit several well-known results in EDM community\cite{KrislockEDMandApp} in Section \ref{subsecII-A_EDM}
and concisely overview some of the key ingredients for optimizing on $\mathcal{S}_{+}^{d,n}$ with quotient geometry $\mathbb{R}^{n\times d}_{*}/\mathrm{O}(d)$\cite{BoumalIntrotoMani} in Section \ref{subsecII-B_RieQM}. We refer readers to\cite{AbsilOptonManifold}\cite{BoumalIntrotoMani} for a comprehensive treatment of optimization algorithms on matrix manifolds. 
\subsection{Euclidean Geometry and EDMC problem}
\label{subsecII-A_EDM}
There is a natural relationship between the set $\mathcal{S}^{n}_+$ and $\mathbb{EDM}^{n}$. Suppose $\mathbf{D}\in \mathbb{EDM}^n$ and its point realization is given by $\mathbf{Y}=[\mathbf{p}_1, \dots ,\mathbf{p}_n]^T \in \mathbb{R}^{n\times d}_{*}$. Let $\mathbf{G}=\mathbf{Y}\mathbf{Y}^T\in\mathcal{S}^{n}_+$. $g:\mathcal{S}(n) \to \mathcal{S}(n)$ is denoted as
\begin{equation}
	\setlength\belowdisplayskip{4pt}
	\setlength\abovedisplayskip{4pt}
	\label{eq_Gram_to_EDM}
	g(\mathbf{G}):=\mathrm{diag}(\mathbf{G})\mathbf{1}^T+\mathbf{1}\mathrm{diag}(\mathbf{G})^T-2\mathbf{G}.
\end{equation}
Clearly, $g$ maps the cone of semidefinite matrices onto $\mathbb{EDM}^n$\cite{KrislockEDMandApp}. We call $d$ the embedding dimension of this specific EDM. 
The inverse problem of turning an EDM back to the original collection of points is non-trivial. The absolute positions are ``lost" in the forward mapping since rigid transformations do not change the pairwise distance between nodes. Using the geometric centering matrix $\mathbf{J}=\mathbf{I}-\frac{1}{n}\mathbf{1}\mathbf{1}^T$, one can map the EDM back to a Gram matrix
\begin{equation}
	\setlength\belowdisplayskip{3pt}
	\setlength\abovedisplayskip{3pt}
	\label{eq_EDM_to_Gram_inv}
	\hat{\mathbf{G}}=-\frac{1}{2}\mathbf{J}\mathbf{D}\mathbf{J},\, \mathbf{D}\in \mathbb{EDM}^n, \, \hat{\mathbf{G}}\in \mathcal{S}^{n}_+.
\end{equation}
An alternative to state \eqref{eq_EDM_to_Gram_inv} is saying that a symmetric hollow matrix $\mathbf{D}$ is an EDM iff it is negative semidefinite on $\{\mathbf{1}\}^{\perp}$\cite{DokmanicEDMTheory}.
The point set can be found by the d-truncated eigenvalues decomposition (EVD) of the self-centered Gram matrix $\hat{\mathbf{G}}$
\begin{equation}
	\setlength\belowdisplayskip{5pt}
	\setlength\abovedisplayskip{5pt}
	\label{eq_cMDS}
	\hat{\mathbf{G}}=\mathbf{U}\mathbf{\Lambda}\mathbf{U}^T,\mathbf{\Lambda}=\mathrm{diag}(\lambda_1,\dots,\lambda_d),\,\hat{\mathbf{Y}}=\mathbf{U}\sqrt{\mathbf{\Lambda}}.
\end{equation} 
\eqref{eq_EDM_to_Gram_inv}, \eqref{eq_cMDS} is the so-called cMDS. After using anchors to recover the rotation and translation matrices, the absolute positions of all nodes can be extracted from $\hat{\mathbf{Y}}$.

By applying \eqref{eq_Gram_to_EDM}, one can transform \eqref{eq_FNMOptFunc_LRMC} into
\begin{equation}
	\label{eq_transed_EDMC_FNM}
	\setlength\belowdisplayskip{4pt}
	\setlength\abovedisplayskip{4pt}
	\begin{aligned}
		\min \, &h(\mathbf{G})=\frac{1}{2} \Vert \mathbf{W}\odot \mathcal{P}_{\Omega}(g(\mathbf{G})-\mathbf{D}_e) \Vert_F^2, \\
		\mathrm{s.t.} \, &\mathrm{rank}(\mathbf{G})=d,\,\mathbf{G} \succcurlyeq 0.	
	\end{aligned}	
\end{equation}
Nguyen et al.\cite{NguyenLRM-CG} tackled \eqref{eq_transed_EDMC_FNM} on $\mathcal{S}_{+}^{d,n}$ with an embedding geometry\cite{VandereyckenEmbgeoSnp} and proposed the LRM-CG algorithm. It requires a retraction step that involves QR-factorization of a $n\times d$ matrix, which can be computationally demanding when dealing with larger networks. Mishra et al.\cite{MishraRieEDMC} introduced the method of Burer-Monteiro type factorization to reformulate \eqref{eq_transed_EDMC_FNM} into the following unconstrained non-convex problem
\begin{equation}
	\label{eq_transed_MFed_EDMC_FNM}
		\setlength\belowdisplayskip{3pt}
	\setlength\abovedisplayskip{3pt}
	\min_{\mathbf{[Y]}\in \mathbb{R}_{*}^{n\times d}/\mathrm{O}(d)} f([\mathbf{Y}])=\frac{1}{2}\Vert \mathbf{W}\odot \mathcal{P}_{\Omega}(g(\mathbf{YY}^T)-\mathbf{D}_e) \Vert_F^2.
\end{equation}
Difference between notions in \eqref{eq_s_stress_gmap_EDMC} and \eqref{eq_transed_MFed_EDMC_FNM}, also $f$ in \eqref{eq_transed_MFed_EDMC_FNM} and $\bar{f}$ in \eqref{eq_transed_MFed_EDMC_FNM_EG_EH} will all become clear in Section \ref{subsecII-B_RieQM}, one can treat them equal currently. 
The Riemannian Trust-Region (RTR) on $\mathcal{S}_{+}^{d,n}$ is applied with quotient geometry to solve \eqref{eq_transed_MFed_EDMC_FNM} in\cite{MishraRieEDMC}. This approach is known to have the ability to escape from saddle points and local minima in practice\cite{AbsilRTR}.
The Fr\'{e}chet derivative is defined as $\langle \nabla f(\mathbf{X}),\mathbf{V} \rangle=\mathrm{D}f(\mathbf{X})[\mathbf{V}]=\lim_{t\to 0} \frac{f(\mathbf{X}+t\mathbf{V})-f(\mathbf{X})}{t}$. 
The adjoint operator of $g$ is $g^{*}(\mathbf{D})=2(\mathrm{diag}(\mathbf{D}\mathbf{1})-\mathbf{D})$\cite{KrislockEDMandApp}. We also present the Euclidean gradient and Hessian of \eqref{eq_transed_MFed_EDMC_FNM} here for completeness, they can be found in\cite{MishraRieEDMC}.
\begin{subequations}
	\label{eq_transed_MFed_EDMC_FNM_EG_EH}
	\setlength\belowdisplayskip{3pt}
	\setlength\abovedisplayskip{3pt}
	\begin{gather}
		\label{eq_transed_MFed_EDMC_FNM_Egrad1}
		\mathbf{H}=\mathcal{P}_{\Omega}(\mathbf{W}\odot\mathbf{W}), \, \mathbf{S}_1=\mathbf{H}\odot(g(\mathbf{YY}^T)-\mathbf{D}_e), \\
		\label{eq_transed_MFed_EDMC_FNM_EHess1}
		\mathbf{S}_2=\mathbf{H}\odot(g(\mathbf{YZ}^T+\mathbf{ZY}^T)),\\
		\label{eq_transed_MFed_EDMC_FNM_Egrad2}
		\nabla \bar{f}(\mathbf{Y})=2g^{*}(\mathbf{S}_1)\mathbf{Y},\\
		\label{eq_transed_MFed_EDMC_FNM_EHess2}
		\nabla^2\bar{f}(\mathbf{Y})[\mathbf{Z}]=2g^{*}(\mathbf{S}_1)\mathbf{Z}+2g^{*}(\mathbf{S}_2)\mathbf{Y}.
	\end{gather}
\end{subequations}
\subsection{Optimization over Riemannian Quotient Manifold}
\label{subsecII-B_RieQM}
There is an equivalence relation $\sim$ defined on $\mathbb{R}_{*}^{n\times d}$: $\mathbf{Y}_1 \sim \mathbf{Y}_2$ if and only if $\mathbf{Y}_2=\mathbf{Y_1Q}$ for $\mathbf{Q}\in \mathrm{O}(d)$. This equivalence class of $\mathbf{Y}\in \mathbb{R}_{*}^{n\times d}$ is denoted by $[\mathbf{Y}]=\{\mathbf{YQ}|\mathbf{Q}\in \mathrm{O}(d)\}$. Let the quotient set $\mathbb{R}_{*}^{n\times d}/\mathrm{O}(d)$ be defined as
\begin{equation*}
	\setlength\belowdisplayskip{3pt}
	\setlength\abovedisplayskip{3pt}
	\mathbb{R}_{*}^{n\times d}/\mathrm{O}(d):=\mathbb{R}_{*}^{n\times d}/\sim:=\{[\mathbf{Y}]|\mathbf{Y}\in \mathbb{R}_{*}^{n\times d}\}.
\end{equation*}
The quotient map is denoted by
\begin{equation*}
		\setlength\belowdisplayskip{3pt}
	\setlength\abovedisplayskip{3pt}
	\pi:\mathbb{R}_{*}^{n\times d}\to \mathbb{R}_{*}^{n\times d}/\mathrm{O}(d): \mathbf{Y}\mapsto \pi(\mathbf{Y})=[\mathbf{Y}],
\end{equation*} 
which is a continuous function w.r.t. the quotient topology. 
\newtheorem{theorem}{Theorem}[section]
\begin{theorem}
	\label{them_Snp_and_Rquotient_rela}
	Let $\mathcal{S}_{+}^{d,n}=\{\mathbf{YY}^T|\mathbf{Y}\in \mathbb{R}_{*}^{n\times d}\}$. The quotient manifold $\mathbb{R}_{*}^{n\times d}/\mathrm{O}(d)$ is diffeomorphic to $\mathcal{S}_{+}^{d,n}$\textnormal{\cite[Prop. 2.8]{MassartQuotientgeometry}}.
\end{theorem}
Two Riemannian metrics are considered on the total space. For $\mathbf{Z}_1,\mathbf{Z}_2 \in \mathrm{T}_{Y}\mathbb{R}_{*}^{n\times d}$, we define
\begin{subequations}
	\label{eq_Rie_metric_on_originalS}
	\setlength\belowdisplayskip{4pt}
	\setlength\abovedisplayskip{4pt}
	\begin{gather}
		\label{eq_normalRie_metric_on_originalS}
		g^1_{\mathbf{Y}}(\mathbf{Z}_1,\mathbf{Z}_2)=\mathrm{tr}(\mathbf{Z}_1^T\mathbf{Z}_2),\\
		\label{eq_HuangRie_metric_on_originalS}
		g^2_{\mathbf{Y}}(\mathbf{Z}_1,\mathbf{Z}_2)=\mathrm{tr}((\mathbf{Y}^T\mathbf{Y})\mathbf{Z}_1^T\mathbf{Z}_2).
	\end{gather}
\end{subequations} 
\eqref{eq_HuangRie_metric_on_originalS} is originated from\cite{MishraRiegeometryMetric}, the reason for choosing it will be further explained in Section \ref{subsecIII-D_rankReduction}.
$[\mathbf{Y}]$ is an embedding submanifold of $\mathbb{R}^{n\times d}_{*}$. Its tangent space $\mathrm{T}_{Y}[\mathbf{Y}]$ is a subspace of $\mathrm{T}_{Y}\mathbb{R}_{*}^{n\times d}$ known as the vertical space $\mathcal{V}_{\mathbf{Y}}$, which satisfies
\begin{equation}
	\setlength\belowdisplayskip{4pt}
	\setlength\abovedisplayskip{4pt}
	\label{eq_Vertical_Space}
	\mathcal{V}_{\mathbf{Y}}=\{\mathbf{Y\Omega}|\mathbf{\Omega}^{T}=-\mathbf{\Omega},\mathbf{\Omega}\in \mathbb{R}^{d\times d}\}.
\end{equation}
The orthogonal complement to $\mathcal{V}_{\mathbf{Y}}$ is the horizontal space $\mathcal{H}_{\mathbf{Y}}$ which is chosen w.r.t. the Riemannian metric. The standard $\mathcal{H}_{\mathbf{Y}}^1$ under the canonical matrix inner product \eqref{eq_normalRie_metric_on_originalS} satisfies
\begin{equation}
	\setlength\belowdisplayskip{4pt}
	\setlength\abovedisplayskip{4pt}
	\label{eq_Horizontal_Space_g1}
	\mathcal{H}_{\mathbf{Y}}^1=\{\mathbf{Z}\in \mathbb{R}^{n\times d}|\mathbf{Y}^T\mathbf{Z}-\mathbf{Z}^T\mathbf{Y}=0\}, 
\end{equation} 
while \eqref{eq_HuangRie_metric_on_originalS} results in\vspace{-4pt}
\begin{IEEEeqnarray}{l}
	\label{eq_Horizontal_Space_g2}
	\mathcal{H}_{\mathbf{Y}}^2=\{\mathbf{Z}\in \mathbb{R}^{n\times d}|(\mathbf{Y}^T\mathbf{Y})^{-1}\mathbf{Y}^T\mathbf{Z}=\mathbf{Z}^T\mathbf{Y}(\mathbf{Y}^T\mathbf{Y})^{-1}\}.
	\IEEEeqnarraynumspace\vspace{-6pt}
\end{IEEEeqnarray}
\newtheorem{proposition}{Proposition}[section]
\begin{proposition}\vspace{-15pt}
	\label{prop_RieMg1_project_on_HVspace}
	Under Riemannian metric $g^1_{\mathbf{Y}}$, the orthogonal projection of any matrix $\mathbf{Z}\in \mathbb{R}^{n\times d}$ onto $\mathcal{V}_{\mathbf{Y}}$ and $\mathcal{H}_{\mathbf{Y}}^1$ are given by\vspace{-4pt}
	\begin{equation*}
		\setlength\belowdisplayskip{4pt}
		\setlength\abovedisplayskip{4pt}
		\mathcal{P}_{\mathbf{Y}}^{\mathcal{V}}(\mathbf{Z})=\mathbf{Y\mathbf{\Omega}},\,\mathcal{P}_{\mathbf{Y}}^{\mathcal{H}^1}(\mathbf{Z})=\mathbf{Z}-\mathbf{Y\mathbf{\Omega}},
	\end{equation*}
	where $\mathbf{\Omega}\in \mathrm{Skew}(d)$ satisfying the Sylvester equation
	\begin{equation*}
		\setlength\belowdisplayskip{4pt}
		\setlength\abovedisplayskip{4pt}
		\label{eq_Sylvester}
		\mathbf{\Omega}\mathbf{Y}^T\mathbf{Y}+\mathbf{Y}^T\mathbf{Y\Omega}=\mathbf{Y}^T\mathbf{Z}-\mathbf{Z}^T\mathbf{Y}.\vspace{-4pt}
	\end{equation*}
\end{proposition}
\begin{proposition}
	\label{prop_RieMg2_project_on_HVspace}
	Under Riemannian metric $g^2_{\mathbf{Y}}$, the orthogonal projection of any matrix $\mathbf{Z}\in \mathbb{R}^{n\times d}$ onto $\mathcal{V}_{\mathbf{Y}}$ and $\mathcal{H}_{\mathbf{Y}}^2$ are given by\textnormal{\cite{MishraRiegeometryMetric}}
	\begin{equation*}
		\setlength\belowdisplayskip{4pt}
		\setlength\abovedisplayskip{4pt}
		\mathcal{P}_{\mathbf{Y}}^{\mathcal{V}}(\mathbf{Z})=\mathbf{Y}\mathrm{Skew}((\mathbf{Y}^T\mathbf{Y})^{-1}\mathbf{Y}^T\mathbf{Z}),\,\mathcal{P}_{\mathbf{Y}}^{\mathcal{H}^2}(\mathbf{Z})=\mathbf{Z}-\mathcal{P}_{\mathbf{Y}}^{\mathcal{V}}(\mathbf{Z}).
	\end{equation*}
\end{proposition}
Consider a point $\mathbf{Y}\in\mathbb{R}_{*}^{n\times d}$ and there is a unique tangent vector $\bar{\eta}_\mathbf{Y}\in \mathcal{H}_\mathbf{Y}$ that satisfies $\mathrm{D}\pi(\mathbf{Y})[\bar{\eta}_\mathbf{Y}]=\eta_{[\mathbf{Y}]}$ for each $\eta_{[\mathbf{Y}]} \in \mathrm{T}_{[Y]}\mathbb{R}_{*}^{n\times d}/\mathrm{O}(d)$, which called the horizontal lift of $\eta_{[\mathbf{Y}]}$ at $\mathbf{Y}$\cite[Def. 9.24]{BoumalIntrotoMani}, given by
\begin{equation}
	\label{eq_Horizontal_lift_QuoM}
		\setlength\belowdisplayskip{4pt}
	\setlength\abovedisplayskip{4pt}
	\bar{\eta}_\mathbf{Y}=(\mathrm{D}\pi(\mathbf{Y})|_{\mathcal{H}_\mathbf{Y}})^{-1}[\eta_{[\mathbf{Y}]}]=\mathrm{lift}_\mathbf{Y}(\eta_{[\mathbf{Y}]}).
\end{equation}
\eqref{eq_Rie_metric_on_originalS} also induces two Riemannian metrics on $\mathbb{R}_{*}^{n\times d}/\mathrm{O}(d)$.
\begin{equation}
	\label{eq_Rie_metric_on_QuoM}
	\setlength\belowdisplayskip{2pt}
	\setlength\abovedisplayskip{4pt}
	g_{[\mathbf{Y}]}^{i} (\eta_{[\mathbf{Y}]},\xi_{[\mathbf{Y}]}):=g_{\mathbf{Y}}^{i}(\bar{\eta}_{\mathbf{Y}},\bar{\xi}_{\mathbf{Y}}),\,i=1,2.
\end{equation}
\begin{theorem}\vspace{-12pt}
	\label{them_Riegrad_on_QuoM_using_lift}
	Consider $f:\mathbb{R}_{*}^{n\times d}/\mathrm{O}(d) \to \mathbb{R}$ and its lifted function $\bar{f}=f\circ\pi:\mathbb{R}_{*}^{n\times d}\to \mathbb{R}$ on the total space. Their Riemannian gradient is related as follows\textnormal{\cite[Prop. 9.38]{BoumalIntrotoMani}}
	\begin{equation}
		\label{eq_RieGrad_on_QuoM}
		\setlength\belowdisplayskip{3pt}
		\setlength\abovedisplayskip{3pt}
		\mathrm{lift}_{\mathbf{Y}}(\mathrm{grad}f([\mathbf{Y}]))=\mathrm{grad}\bar{f}(\mathbf{Y}).
	\end{equation}
\end{theorem}
\begin{proposition}
	\label{prop_Riegrad_lift_func_g12}
	Let $\bar{f}:\mathbb{R}_{*}^{n\times d}\to \mathbb{R}$ be the lifted cost function of \eqref{eq_transed_MFed_EDMC_FNM},
	 $\mathrm{grad}\bar{f}(\mathbf{Y})$ satisfies\textnormal{\cite[Prop. 4.13]{ZhengRQCGtheory}}
	\begin{equation*}
		\setlength\belowdisplayskip{3pt}
		\setlength\abovedisplayskip{3pt}
		\mathrm{grad}\bar{f}(\mathbf{Y})=
		\begin{cases}
			\nabla \bar{f}(\mathbf{Y}), &\text{under metric } g^1_{\mathbf{Y}} \\
			\nabla \bar{f}(\mathbf{Y})(\mathbf{Y}^T\mathbf{Y})^{-1}, &\text{under metric } g^2_{\mathbf{Y}}
		\end{cases}.
	\end{equation*}
where $\nabla \bar{f}(\mathbf{Y})$ is defined in \eqref{eq_transed_MFed_EDMC_FNM_Egrad2}.
\end{proposition}
\begin{theorem}
	\label{them_QuoM_lifted_retraction}
	The retraction on $\mathbb{R}_{*}^{n\times d}/\mathrm{O}(d)$ can be defined in terms of the retraction on $\mathbb{R}_{*}^{n\times d}$\textnormal{\cite[Sec. 2.8]{MassartQuotientgeometry}}, given by
	\begin{equation}
		\setlength\belowdisplayskip{4pt}
		\setlength\abovedisplayskip{4pt}
		\label{eq_QuoM_Retraction_totalS}
		\bar{\mathrm{R}}_{\mathbf{Y}}(t\bar{\eta}_\mathbf{Y}):=\mathbf{Y}+t\bar{\eta}_\mathbf{Y},
	\end{equation}
	where $\bar{\eta}_\mathbf{Y}\in \mathrm{T}_Y\mathbb{R}^{n\times d}_{*}$, and $t > 0$ is the step size. And
	\begin{equation}
		\setlength\belowdisplayskip{3pt}
		\setlength\abovedisplayskip{3pt}
		\label{eq_QuoM_Retraction}
		\mathrm{R}_{[\mathbf{Y}]}(t\eta_{[\mathbf{Y}]}):=\pi(\bar{\mathrm{R}}_\mathbf{Y} ( t \cdot \mathrm{lift}_{\mathbf{Y}}(\eta_{[\mathbf{Y}]}) ))=\pi(\mathbf{Y}+t\bar{\eta}_{\mathbf{Y}})
	\end{equation}
	defines a retraction on $\mathbb{R}_{*}^{n\times d}/\mathrm{O}(d)$.
\end{theorem}
\begin{proposition}
\label{prop_quotient_vectorTrans}
The vector transport on $\mathbb{R}_{*}^{n\times d}/\mathrm{O}(d)$ is defined as the projection onto horizontal space, given by
\begin{equation}
	\label{eq_QuoM_VecTrans}
	\setlength\belowdisplayskip{4pt}
	\setlength\abovedisplayskip{4pt}
	\mathrm{lift}_{\bar{\mathrm{R}}_{\mathbf{Y}}(\bar{\eta}_{\mathbf{Y}})}(\mathcal{T}_{\eta_{[\mathbf{Y}]}} \xi_{[\mathbf{Y}]}):=\mathcal{P}^{\mathcal{H}}_{\bar{\mathrm{R}}_{\mathbf{Y}}(\bar{\eta}_{\mathbf{Y}})}(\bar{\xi}_{\mathbf{Y}}).
\end{equation}
This vector transport satisfies\textnormal{\cite[Sec. 4.6]{ZhengRQCGtheory}}
\begin{equation*}
	\setlength\belowdisplayskip{4pt}
	\setlength\abovedisplayskip{4pt}
	\label{eq_dif_retrac_vector_trans}
	\mathrm{D}\mathrm{R}_{[\mathbf{Y}]}(\eta_{[\mathbf{Y}]})[\xi_{[\mathbf{Y}]}]=\mathrm{D}\pi(\mathbf{Y}+\bar{\eta}_{\mathbf{Y}})[\mathcal{P}^{\mathcal{H}}_{\mathbf{Y}+\bar{\eta}_{\mathbf{Y}}}(\bar{\xi}_{\mathbf{Y}})]=\mathcal{T}_{\eta_{[\mathbf{Y}]}}\xi_{[\mathbf{Y}]}.
\end{equation*}
Two distinct Riemannian metrics induce different vector transports, they will both be represented by
\begin{equation*}
	\setlength\belowdisplayskip{4pt}
	\setlength\abovedisplayskip{4pt}
	\mathcal{T}^{Y_{k-1}}_{Y_k}: \mathcal{H}_{\mathbf{Y}_{k-1}}\to \mathcal{H}_{\mathbf{Y}_{k}},\,\bar{\xi}_{\mathbf{Y}_{k-1}}\mapsto \mathcal{P}^{\mathcal{H}}_{\mathbf{Y}_k}(\bar{\xi}_{\mathbf{Y}_{k-1}})
\end{equation*}
for brevity in Section \ref{subsecIII-B_RCGQuoM}.
\end{proposition}

\section{EDMC: A Non-Convex Approach}
\label{section_III_RCG_HZLS}
This section contains all our major contributions. First, Hager-Zhang line search method is generalized to the quotient geometry $(\mathbb{R}^{n\times d}_{*}/\mathrm{O}(d),g^i)$, $i=1,2$ to couple with RCG. Second, we state and prove our main theorem. Last, the ``rank reduction" routine is proposed with some discussions. 
\subsection{Riemannian Hager-Zhang Line Search}
\label{subsecIII-A_RHZLS}
Hager and Zhang originally proposed this high-accuracy line search method in\cite{HZ_CG_DESCENT2}\cite{HZ_CGdecent1}. 
In line search algorithms, the Wolfe conditions are preferred to find a step size $\alpha_k$.
\begin{subequations}
	\label{eq_Wolfe_Cond_EucSpace}
		\setlength\belowdisplayskip{4pt}
	\setlength\abovedisplayskip{4pt}
	\begin{gather}
		\label{eq_Armijo_Cond_EucSpace}
		F(x_k+\alpha_k\xi_k)-F(x_k)\leq c_1 \alpha_k \xi_k^T \nabla F(x_k), \\
		\label{eq_Curv_Cond_EucSpace}
		\xi_k^T \nabla F(x_k+\alpha_k\xi_k) \geq c_2 \xi_k^T \nabla F(x_k),
	\end{gather}
\end{subequations}
where $F:\mathbb{R}^{m\times n}\to\mathbb{R}$ is a generic cost function, $0<c_1\leq c_2 <1$ are line search parameters. 
Set $\phi(\alpha)=F(x_k+\alpha\xi_k)$ one can reformulate \eqref{eq_Wolfe_Cond_EucSpace} into
\begin{equation}
	\label{eq_Wolfe_Cond_ES_use_phi}
	\setlength\belowdisplayskip{4pt}
	\setlength\abovedisplayskip{4pt}
	c_1 \alpha_k \phi^{\prime}(0)\geq\phi(\alpha_k)-\phi(0),\,\phi^{\prime}(\alpha_k)\geq c_2\phi^{\prime}(0).
\end{equation}
Hager and Zhang argued that $\phi(\alpha_k)-\phi(0)$ cannot be accurately calculated under finite machine precision when iteration reaches the ``flat region" near a local minimum of $F$. This causes step sizes returned from the line search to vanish. To address this issue, they used the derivative $\phi^{\prime}(\alpha_k)$ directly and introduced the approximate Wolfe condition, given by
\begin{equation}
	\label{eq_Approxi_Wolfe_Cond_EucSpace}
	\setlength\belowdisplayskip{4pt}
	\setlength\abovedisplayskip{4pt}
	(2c_1-1)\phi^{\prime}(0)\geq\phi^{\prime}(\alpha_k)\geq c_2\phi^{\prime}(0),
\end{equation}
where $0<c_2\leq1,\,0<c_1<\min\{0.5,c_2\}$. 
The HZLS involves a bracketing process where secant and bisection methods are used to find a zero point of $\phi^{\prime}(\alpha)$ inside the interval $[a,b]$ that satisfies the opposite slope condition
\begin{equation}
	\label{eq_HZLS_OppoSlope_cond}
	\setlength\belowdisplayskip{4pt}
	\setlength\abovedisplayskip{4pt}
	\phi(a)\leq \phi(0)+\epsilon_k,\,\phi^{\prime}(a)<0,\,\phi^{\prime}(b)\geq 0,
\end{equation}
where $\epsilon_k=\epsilon|F(x_k)|$, $\epsilon$ is a small fixed parameter. The line search will terminate whenever \eqref{eq_Wolfe_Cond_ES_use_phi} is satisfied. When numerical error in \eqref{eq_Armijo_Cond_EucSpace} is large,
the line search switches to checking \eqref{eq_Approxi_Wolfe_Cond_EucSpace} permanently. Implementation details are referred to their papers. 
Under quotient geometry, we rewrite the lifted line search function as $\bar{\phi}:=f\circ\pi\circ\bar{\mathrm{R}}_{\mathbf{Y}}$, thus additional proof is required instead of using existing generalization for embedding geometry\cite[Sec. 4]{SuttiRieHZLS} directly. 
\begin{theorem}
	\label{them_HZLS_QuoM_Expression}
	Under quotient geometry $(\mathbb{R}^{n\times d}_{*}/\mathrm{O}(d),g^i)$, $i=1,2$ as defined in \eqref{eq_Rie_metric_on_originalS}, we have
	\begin{subequations}
		\label{eq_phi_and_dphi_QuoM}
		\setlength\belowdisplayskip{3pt}
		\setlength\abovedisplayskip{3pt}
		\begin{gather}
			\label{eq_phi_QuoM}
			\bar{\phi}(\alpha)=\bar{f}(\bar{\mathrm{R}}_{\mathbf{Y}}(\alpha \bar{\eta}_{\mathbf{Y}})),\\
			\label{eq_dphi_QuoM}
			\bar{\phi}^{\prime}(\alpha)=g^1_{\mathbf{Y}}(\nabla \bar{f}(\bar{\mathrm{R}}_{\mathbf{Y}}(\alpha \bar{\eta}_{\mathbf{Y}})),\bar{\eta}_{\mathbf{Y}}).
		\end{gather}
	\end{subequations}
\end{theorem}
\begin{proof}
	 \eqref{eq_phi_QuoM} is from the definition of $\bar{f}:=f\circ\pi$ as described in Theorem \ref{them_Riegrad_on_QuoM_using_lift} and \eqref{eq_QuoM_Retraction_totalS}. To compute \eqref{eq_dphi_QuoM}, we use the chain rule for $\bar{\phi}(\cdot):=f\circ\pi\circ\bar{\mathrm{R}}_{\mathbf{Y}}(\bar{\eta}_{\mathbf{Y}},\cdot)$, it gives
	\begin{subequations}
		\setlength\belowdisplayskip{7pt}
		\setlength\abovedisplayskip{7pt}
		\begin{align*}
			\mathrm{LHS}:=\bar{\phi}^{\prime}(\alpha)
			=&\mathrm{D}(f\circ\pi)(\bar{\mathrm{R}}_{\mathbf{Y}}(\alpha\bar{\eta}_{\mathbf{Y}}))[\frac{\mathrm{d}}{\mathrm{d}t}\bar{\mathrm{R}}_{\mathbf{Y}}(\alpha\bar{\eta}_{\mathbf{Y}})]\\
			=&\mathrm{D}(f\circ\pi)(\bar{\mathrm{R}}_{\mathbf{Y}}(\alpha\bar{\eta}_{\mathbf{Y}}))[\bar{\eta}_\mathbf{Y}].
		\end{align*}
	\end{subequations}
	We denote $\bar{\mathrm{R}}_{\mathbf{Y}}(\alpha\bar{\eta}_{\mathbf{Y}})$ as $\mathbf{X}$ and apply Theorem \ref{them_Snp_and_Rquotient_rela}, we have $\pi(\mathbf{X})=[\mathbf{X}]$ correspond to a point $\mathbf{XX}^T \in \mathcal{S}^{n\times d}_{+}$. Using chain rule again, we have
	\begin{subequations}
		\label{eq_proof_Thenm3-1}
		\setlength\belowdisplayskip{7pt}
		\setlength\abovedisplayskip{7pt}
		\begin{align}
			&\mathrm{LHS}=\mathrm{D}f(\pi(\mathbf{X}))[\mathrm{D}\pi(\mathbf{X})[\bar{\eta}_{\mathbf{Y}}]]\nonumber \\
			\label{eq_proof_Thenm3-1_a}
			&=\mathrm{D}f(\pi(\mathbf{X}))[\mathrm{D}\pi(\mathbf{Y}+\alpha\bar{\eta}_\mathbf{Y})[\mathcal{P}^{\mathcal{H}}_{\mathbf{Y}+\alpha\bar{\eta}_\mathbf{Y}}(\bar{\eta}_{\mathbf{Y}})]]\\
			&\hspace{-1.1pt}\overset{(a)}{=}g^i_{[\mathbf{X}]}(\nabla f([\mathbf{X}]), \mathcal{T}_{\alpha\eta_{[\mathbf{Y}]}}\eta_{[\mathbf{Y}]} )=g^i_{[\mathbf{X}]}(\mathrm{grad}f([\mathbf{X}]), \mathcal{T}_{\alpha\eta_{[\mathbf{Y}]}}\eta_{[\mathbf{Y}]})\nonumber \\
			\label{eq_proof_Thenm3-1_b}
			&=g^i_{\mathbf{X}}(\mathrm{grad}\bar{f}(\mathbf{X}), \mathrm{lift}_{\mathbf{X}}(\mathcal{T}_{\alpha\eta_{[\mathbf{Y}]}}\eta_{[\mathbf{Y}]}))\\
			\label{eq_proof_Thenm3-1_c}
			&=g^i_{\mathbf{X}}( \mathrm{grad} \bar{f}(\mathbf{X}),\mathcal{P}_{\bar{\mathrm{R}}_{Y}(\alpha\bar{\eta}_{\mathbf{Y}})}^{\mathcal{H}^i}(\bar{\eta}_{\mathbf{Y}}))\\
			\label{eq_proof_Thenm3-1_d}
			&\hspace{-1.1pt}\overset{(b)}{=}g^i_{\mathbf{X}}( \mathrm{grad} \bar{f}(\mathbf{X}),\bar{\eta}_{\mathbf{Y}})=g^1_{\mathbf{X}}(\nabla \bar{f}(\bar{\mathrm{R}}_{\mathbf{Y}}(\alpha \bar{\eta}_{\mathbf{Y}})),\bar{\eta}_{\mathbf{Y}}).
		\end{align}
	\end{subequations}
\eqref{eq_proof_Thenm3-1_a} and $(a)$ use Proposition \ref{prop_quotient_vectorTrans}, then followed by the fact that $\mathcal{T}_{\alpha\eta_{[\mathbf{Y}]}}\eta_{[\mathbf{Y}]} \in\mathrm{T}_{[X]}\mathcal{S}^{n\times d}_{+}$. \eqref{eq_proof_Thenm3-1_b} from \eqref{eq_Rie_metric_on_QuoM} and \eqref{eq_proof_Thenm3-1_c} dues to \eqref{eq_QuoM_VecTrans}. Then followed by $\mathrm{grad} \bar{f}(\mathbf{X})\in \mathcal{H}_{\mathbf{X}}^i$ we get $(b)$.
For $i=1$, the last equation is trivial. For $i=2$, we note that the $\mathbf{Y}^T\mathbf{Y}$ term is canceled-out by using \eqref{eq_HuangRie_metric_on_originalS} and Proposition \ref{prop_Riegrad_lift_func_g12}.
\end{proof}\vspace{-4pt}
For RCG on generic Riemannian manifolds, its local convergence analysis utilizes the Riemannian Wolfe conditions.\vspace{-10pt}

\begin{small}
	\begin{IEEEeqnarray}{l}
		\IEEEyesnumber\label{eq_RieWolfe_cond_GenM}
		\IEEEyessubnumber*
		\label{eq_RieWolfe_cond_GenM1}
		F_{\mathcal{M}}(x_{k+1})-F_{\mathcal{M}}(x_k)\leq c_1\alpha_k\langle\mathrm{grad}F_{\mathcal{M}}(x_k),\eta_k\rangle_{x_k},\\
		\label{eq_RieWolfe_cond_GenM2}
		\langle\mathrm{grad}F_{\mathcal{M}}(x_{k+1}),\mathcal{T}^{x_k}_{x_{k+1}}(\eta_k)\rangle_{x_{k+1}}\leq c_2 \langle\mathrm{grad}F_{\mathcal{M}}(x_k),\eta_k\rangle_{x_k},
		\IEEEeqnarraynumspace\vspace{-4pt}
	\end{IEEEeqnarray}
\end{small}
where $F_{\mathcal{M}}:\mathcal{M}\to\mathbb{R}$, $x_{k+1}=\mathrm{R}_{x_k}(\alpha_k\eta_k)$, $\eta_k\in \mathrm{T}_{x_k}\mathcal{M}$ and $\langle\cdot,\cdot\rangle_{x_k}$ is the Riemannian metric at $x_k$. By noticing \eqref{eq_proof_Thenm3-1_c}, we suggest that the Riemannian version of \eqref{eq_Wolfe_Cond_ES_use_phi} is the same as \eqref{eq_RieWolfe_cond_GenM} under $(\mathbb{R}^{n\times d}_{*}/\mathrm{O}(d),g^i)$ while the convergence analysis of CG methods based on \eqref{eq_Approxi_Wolfe_Cond_EucSpace} is highly open\cite[Sec. 3]{HZ_CG_DESCENT2}\cite[Sec. 5.5]{SuttiRMGLSPhDthis}.
Local convergence analysis of our RCG-HZLS framework on $(\mathbb{R}^{n\times d}_{*}/\mathrm{O}(d),g^i)$ can be conducted by using \eqref{eq_RieWolfe_cond_GenM} as in\cite{RingRCGframeWork}. This proof is quite standard and thus omitted.
\subsection{RCG on the Quotient Manifold $\mathbb{R}^{n\times d}_{*}/\mathrm{O}(d)$}
\label{subsecIII-B_RCGQuoM}
The use of RCG serves a twofold purpose. Firstly, numerical studies have demonstrated its effectiveness in addressing the LRMC problem\cite{VandereyckenFixrankMani_CG}.
Secondly, it is analogous to Vanilla Gradient Descent (VGD). 
We adopt the standard RCG framework\cite{ZhengRQCGtheory}, and our algorithm is listed in Algorithm \ref{alg_RCG_on_QuoM}. 
We choose to use Hager-Zhang (HZ+) updating rule described in\cite{HZ_CGdecent1} for $\beta_k$ in Line 14, which has been generalized to RCG and shown to always provide sufficient descent in\cite{Sakai_HZRCG}. A similar line search switching technique mentioned in\cite{HZ_CG_DESCENT2} is also implemented from Line 3 to 11. 
	\begin{subequations}
		\label{eq_HZLS_LSError_condition}
		\setlength\belowdisplayskip{4pt}
		\setlength\abovedisplayskip{12pt}
		\begin{gather}
			\label{eq_HZLS_LSError_condition1}
			|\bar{f}(\mathbf{Y}_{k+1})-\bar{f}(\mathbf{Y}_{k})|\leq \omega C_k, \\
			\label{eq_HZLS_LSError_condition2}
			\begin{cases}
				Q_k=1+Q_{k-1}\Delta, \, &Q_{-1}=0,\\
				C_k=C_{k-1}+(|\bar{f}(\mathbf{Y}_{k})|-C_{k-1})/Q_k, \, &C_{-1}=0.
			\end{cases}
		\end{gather}
	\end{subequations}
It begins with a simple Armijo backtracking method\footnote{The initial step size of Armijo backtracking (Line 4, Algorithm \ref{alg_RCG_on_QuoM}) is computed via a similar method as in LRMC problem. Please see\cite[Sec. 6.3]{ZhengRQCGtheory} or\cite[Sec. 3]{VandereyckenFixrankMani_CG} for details.} until \eqref{eq_HZLS_LSError_condition} is satisfied or the backtracking fails to converge after a finite time of shrinkage, then it will switch to Riemannian HZLS permanently. $F_{\mathrm{stop}}$ is the stopping criteria which will be detailed in Section \ref{secV_NumExps}.
If both the low-rank matrix and sampling scheme lack inherent structure, e.g., matrix generated by the random orthogonal model in\cite{LRMC_Can1} and sampled uniformly at random, then spectral initialized VGD leads to guaranteed success given near-optimal sample complexity\cite{MaImplicitRegularNonCVX_GD}. But for the SNL problem, we find that line search algorithms can be compelled to take vanishingly small step size and quit erroneously with high probability when initiated from random starting points in a sparsely connected network. Methods using ideas directly from LRMC like SVD-MDS yield only marginal improvement while MAP-MDS does not generalize its success to irregularly-shaped networks. The performance loss of spectral methods on practical SNL problems is mainly because their ``estimators" get biased. For SVD-MDS, a tight bound on $\Vert\frac{1}{p}\mathcal{P}_{\Omega}\mathbf{D}_e-\mathbf{D}_e\Vert$, where $p=|\Omega|/n^2$, does not exist under unit ball model since $\Omega$ only samples elements with small intensity values. While the shortest path algorithm fails to approximate the Euclidean distance when the topology of the network becomes complicated\footnote{We refer readers to Appendix \ref{App_D_Shape_of_Network} for preliminary discussions.}. We now further address this in the next two sections.
\subsection{The Attractive Basin of EDMC Under Bernoulli Rule}
\label{subsecIII-C_Basin}
In this section, we show that the ill-posed nature of solving SNL problems via non-convex methods comes from the unit ball sample rule but not the non-orthonormality of $\boldsymbol{\omega}_{\boldsymbol{\alpha}}$, since the local landscape of s-stress is quite benign under Bernoulli sampling scheme. Before we step into the analysis of EDMC, we briefly recall what the non-convex LRMC counterpart\cite{ZhengLaffertyNonCVXFR}\cite{SL15NonCVXFR}\cite{MaImplicitRegularNonCVX_GD} tells us. Loosely speaking, their main result is threefold
: (i) matrix factorization together with FNM cost function enjoy great numerical superiority in solving the LRMC problem; (ii) this cost function has a small region of attraction around its global minima provided with large enough samples, and VGD is guaranteed to converge when initialized inside this region; (iii) the spectral method helps one enter this region when used as initialization\footnote{This is not the only way to seek guaranteed success for non-convex LRMC. Another list of works consists of landscape analysis to show the non-existence of spurious local minima\cite{RongGeSpuriousLocalMinima} and followed by a saddle point escaping algorithm\cite{pmlr-v70-jin17a}. We refer interested readers to recent survey\cite{ChiYNonCVX_Facor_overview}.}. Intuitively, one may hope all these three results hold for EDMC. Unfortunately, the latter two seem to need either highly non-trivial modification or much more dedicated design to tolerate the unit ball rule. 
\begin{algorithm}[!t]
	\caption{RCG-HZLS on $\mathbb{R}_{*}^{n\times d}/\mathrm{O}(d)$}
	\label{alg_RCG_on_QuoM}
	\begin{algorithmic}[1]
		\REQUIRE initial point $\mathbf{Y}_0\in[\mathbf{Y}_0]$, cost function $f([\mathbf{Y}])$ and its gradient as in \eqref{eq_transed_MFed_EDMC_FNM} \eqref{eq_transed_MFed_EDMC_FNM_EG_EH}, stop function $F_{\mathrm{stop}}$, max iterations limit $\mathrm{IMAX}$, line search type flag $\mathrm{LS}_F=1$ for Armijo backtracking at first few iterations and $2$ for pure HZLS, Riemannian metric $g^i$.
		\STATE set initial direction $\boldsymbol{\xi}_0: \boldsymbol{p}_0=\boldsymbol{\xi}_0=-\mathrm{grad}\bar{f}(\mathbf{Y}_0)$
		\FOR{$k=0,1,2,\dots$ to $\mathrm{IMAX}$}
		\IF{$\mathrm{LS}_F=1$}
		\STATE Initial step size $\alpha_k=\arg\min_\alpha \bar{f}(\bar{\mathrm{R}}_{\mathbf{Y}_k}(\alpha \boldsymbol{\xi}_k))$
		\STATE Use Armijo backtracking to find a step size $0.5^t\alpha_k$ and the smallest positive integer $t$ that
		$\bar{f}(\mathbf{Y}_k)-\bar{f}(\bar{\mathrm{R}}_{\mathbf{Y}_k}(0.5^t \alpha_k\boldsymbol{\xi}_k))\geq -c_1*0.5^t \alpha_k\,g^i_{\mathbf{Y}_k}(\boldsymbol{p}_k,\boldsymbol{\xi}_k)$
		\IF{$t \geq L$}
		\STATE Switch to Riemannian HZLS for $\alpha_k$, set $\mathrm{LS}_F=2$
		\ENDIF
		\ELSIF{$\mathrm{LS}_F=2$}
		\STATE Use Riemannian HZLS to select a step size $\alpha_k$ with the line search function \eqref{eq_phi_and_dphi_QuoM}
		\ENDIF
		\STATE Calculate the retraction $\mathbf{Y}_{k+1}=\bar{\mathrm{R}}_{\mathbf{Y}_k}(\alpha_k\boldsymbol{\xi}_k)$
		\STATE Calculate Riemannian gradient $\boldsymbol{p}_{k+1}=\mathrm{grad}\bar{f}(\mathbf{Y}_{k+1})$
		\STATE Set $\boldsymbol{\xi}_{k+1}=-\boldsymbol{p}_{k+1}+\beta^{HZ+}_{k+1}\mathcal{T}^{Y_k}_{Y_{k+1}}(\boldsymbol{\xi}_k)$
		\STATE Calculate \eqref{eq_HZLS_LSError_condition} to switch $\mathrm{LS}_F$
		\IF{$F_{\mathrm{stop}}$} 
		\RETURN $\mathbf{Y}_{k}$
		\ENDIF
		\ENDFOR
	\end{algorithmic}
\end{algorithm}
\theoremstyle{plain}
\begin{theorem}
	\label{them_distance_tirival_bounded}
	Suppose the ground truth point set is self-centered\footnote{The translation ambiguity can be easily removed using self-centered initialization point set\cite{GluntEmbeddEDM}. From now on we assume any point realization is self-centered.} and denoted by $\mathbf{Y}^{\star}\in\mathbb{R}^{n\times d}_{*}$, $\mathbf{G}^{\star}=\mathbf{Y}^{\star}\mathbf{Y}^{{\star}T}=\mathbf{U}^{\star}\boldsymbol{\Sigma}^{\star}\mathbf{U}^{{\star}T}$, $\mathbf{U}^{\star}\in\mathbb{R}^{n\times d}$, which satisfies the standard incoherence condition in LRMC
	\begin{equation*}
		\setlength\belowdisplayskip{3pt}
		\setlength\abovedisplayskip{3pt}
		\Vert\mathbf{U}^{\star}\Vert_{2,\infty}^2\leq\frac{\mu d}{n},\,\Vert\mathbf{Y}^{\star}\Vert_{2,\infty}^2\leq\frac{\mu d\sigma_1^{\star}}{n},\,\sigma_1^{\star}=\sigma_1(\mathbf{Y}^{\star}\mathbf{Y}^{{\star}T}).
	\end{equation*}
	Let $\Delta:=\mathbf{Y}-\mathbf{Y}^{\star}\boldsymbol{\psi}^{\star}$ and $\boldsymbol{\psi}^{\star}\in O(d)$ as defined in Lemma \ref{eq_geodesic_distance}. There exists an incoherent and attractive region $\mathcal{B}:=\{\Delta|\Vert\Delta\Vert_F^2\leq\frac{\sigma_d^{\star}}{c_3},\,\Vert\Delta\Vert_{2,\infty}^2\leq\frac{c_4\sigma_1^{\star}}{\kappa n}\}$, where $\sigma_d^{\star}=\sigma_d(\mathbf{Y}^{\star}\mathbf{Y}^{{\star}T})$ and $\kappa$ is the condition number of $\mathbf{Y}^{\star}\mathbf{Y}^{{\star}T}$. Suppose the ground truth EDM $\mathbf{D}^{\star}$ is sampled by the Bernoulli rule with:
	\begin{itemize}
		\item[(1)] $p\geq \max\{C_T\beta (\mu d)^3,C_D(\mu d)^2\}\log n/n$, for some large enough constant $C_T,C_D\gg\beta>1$ independent of $n$, $\mu$, $d$. Inside $\mathcal{B}$, it holds that $\langle\nabla\bar{f}(\mathbf{Y}),\Delta\rangle\geq pc_5\sigma_d^{\star}\Vert\Delta\Vert_F^2$ with probability at least $1-n^{1-\beta}-\frac{1}{2}n^{-4}-2n^{-8}$.
		\item[(2)] $p\geq C_D(\mu d)^2\log n/n$ for some large enough constant $C_D$. Inside $\mathcal{B}$, it holds that $\Vert\nabla\bar{f}(\mathbf{Y})\Vert_F\leq pC_P\mu d\sigma_1^{\star}\Vert\Delta\Vert_F$ for some large constant $C_P$ independent of $\mu$, $d$, $n$, and $\kappa$ with probability at least $1-4n^{-8}-\frac{1}{2}n^{-4}$.
	\end{itemize} 
	Where $\nabla\bar{f}(\mathbf{Y})$ is defined in \eqref{eq_transed_MFed_EDMC_FNM_Egrad2}, and numerical constants $c_3$, $c_4$, $c_5$ will be specified in Appendix \ref{Appdi_B_Bound}.
\end{theorem}
\begin{proof}
	The proof is based on several prior arts\cite{ZhengLaffertyNonCVXFR}\cite{SL15NonCVXFR}\cite{TasissaEDMCProof}. Please see Appendix \ref{Appdi_B1_Bound} for (1) and Appendix \ref{Appdi_B2_Bound} for (2), respectively.
\end{proof}
\begin{figure}[t]
	\centering
	\setlength{\abovecaptionskip}{0cm}
	\setlength{\belowcaptionskip}{-3cm}
	\includegraphics[width=9cm]{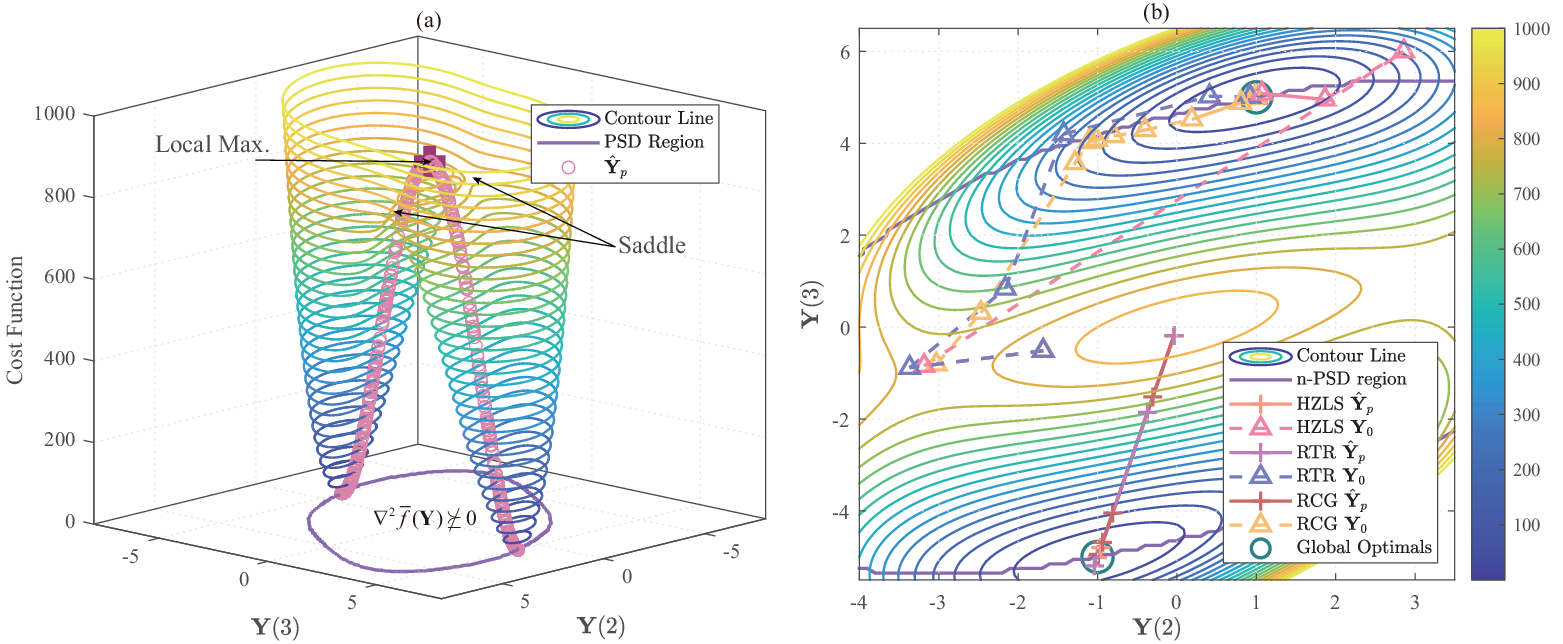}
	\caption{The ground truth nodes positions are given by $\mathbf{Y}^{\star}_1=[0,1,5]^T$ or $[0,-1,-5]^T$. We plot the level contours of the s-stress function along the last two dimensions in (a). There is a local maximum at the origin indicated by a purple plus sign and two saddle points aside from the local maximum\cite[Ch. 3]{ParhizkarPhDthsis}. Pink scatter points represent the $\hat{\mathbf{Y}}_p$ generated in 200 independent trials. Inside the purple ellipse on the xy-plane, the Hessian of s-stress function is not PSD numerically. (b) shows an instance of the trajectories generated by RCG-HZLS, RTR, and regular RCG. These iterations start from a random point $\mathbf{Y}_0\in\mathbb{R}^{3\times 2}$ or a recommended initial point $\hat{\mathbf{Y}}_p$. Starting from $\hat{\mathbf{Y}}_p$ gives sharp and straightforward convergence while starting from a random point results in wiggling.}
	\label{fig_toy_model}
\end{figure}
Two immediate corollaries can be obtained from Theorem \ref{them_distance_tirival_bounded}: (i) $\langle\nabla \bar{f}(\mathbf{Y})-\nabla \bar{f}(\mathbf{Y}^{\star}),\Delta\rangle>c\Vert\Delta\Vert_F^2$ will hold inside $\mathcal{B}$ given enough samples. This implies that any stationary point of $\bar{f}$ in $\mathcal{B}$ is a global minimum. (ii) VGD can achieve non-vanishing step sizes when optimizing \eqref{eq_s_stress_gmap_EDMC} starting inside $\mathcal{B}$, since $\Vert\nabla\bar{f}(\mathbf{Y})\Vert_F\leq c\Vert\Delta\Vert_F$ for $c$ independent of $n$\cite[Lemma 7.10]{PhaseretrievalWirtinger}. These facts also conjecture that SVD-MDS\cite{DrineasSVD_MDS}\cite{ZhangSVD_MDS} initialized VGD could enjoy good convergence on \eqref{eq_s_stress_gmap_EDMC} given enough Bernoulli samples and a proper regularization term\footnote{The regularization term as in\cite{SL15NonCVXFR}\cite{RongGeSpuriousLocalMinima}, can be used to force iteration to stay incoherent, while revealing the step size selection and explicit converge speed of VGD need further analysis (please see Appendix \ref{subsec_discuss_remark_appB} for discussions). We also notice that the ``implicit regularization" phenomenon\cite{MaImplicitRegularNonCVX_GD} seems to still exist here, i.e., when $p$ is large enough there is no need to use a regularization term.}, indicating that the Bernoulli sampling scheme makes the EDMC problem easier to solve when compared with the unit ball rule. A preliminary test fixing $d=2$ shows that to obtain no failure in 20 trials\footnote{The point set is generated by standard Gaussian distribution. We claim a success if the EDM recover rate (see Section \ref{subsecV_I_ImplementDetails}) falls below $10^{-3}$. Its empirical phase transition is referred to Appendix \ref{App_D_Phase_Transition_SVD_MDS_GD}.}, SVD-MDS initialized GD with HZLS step size needs $p\geq c\log n/n$ where $c$ is about $4$ to $2$ when $n$ varying from $100$ to $1000$. 
Compared with rigidity theory, an Erd\"os-R\'enyi graph tends to become generically globally rigid in $\mathbb{R}^2$ around $p>(\log n+3\log\log n+w(n))/n$ for $\lim_{n\to\infty}w(n)=\infty$ \cite[Thm. 2.6]{SingerTheBound}, indicating that the sample complexity in Theorem \ref{them_distance_tirival_bounded} reasonably falls within the information theory lower bound. The incoherence assumption on $\mathbf{Y}^{\star}$ introduced here can be satisfied by some statistical models analyzed in\cite[Sec. I-E]{CandTaoLRMC} (please also see\cite{ChenIncoOptimalMC}), and has been utilized in\cite{TasissaEDMCProof} to demonstrate that trace minimization achieves near-optimal sample complexity in EDMC problem when given a structure-less sample mask. Experimental results in\cite[Sec. V]{TasissaEDMCProof} suggest that trace minimization succeeds on some complicated point sets extracted from highly irregularly-shaped manifolds given a few randomly chosen inter-point distance measurements. Also, for nodes distributed inside some convex polyhedrons via regular manner, the incoherence assumption seems to be satisfied empirically. For example, if the convex polyhedron is near-isotropic and the points are dropped uniformly at random inside this convex body, then the corresponding position matrix $\mathbf{Y}^{\star}$ tends to enjoy low coherence. This assumption may be viewed as a ``margin" or ``tolerance" since the final sample complexity scales proportionally with the coherence parameter $\mu$. As $\mu$ grows, the distribution of intensity in $\mathbf{Y}^{\star}$ gets concentrated, rendering the characterization of the average sample number less informative. Thus, we suggest that the incoherence assumption tend to be more general than the ``regularly distributed in convex polyhedron" assumption used in\cite{KarbasiOhMAPMDSTri}, since it gives tolerance towards the irregularity of the distribution of the point set (at least empirically)\cite[Remark. 3]{TasissaEDMCProof}.
However, if one changes the sampling strategy to unit ball rule, then all distortion bounds in Appendix \ref{Appdi_B_Bound} do not hold anymore, since only small distance measurements are available now, causing Theorem \ref{them_distance_tirival_bounded} to break down. While the attractive region seems to still exist empirically\footnote{We refer readers to Appendix \ref{App_D_attractiveRegion_untiball} for preliminary discussions.}. We next resort to a ``rank reduction" routine to obtain near-optimal performance for the practical SNL problem. 
\begin{figure}[!t]
	\centering
	\setlength{\abovecaptionskip}{0cm}
	\setlength{\belowcaptionskip}{-3cm}
	\includegraphics[width=9cm]{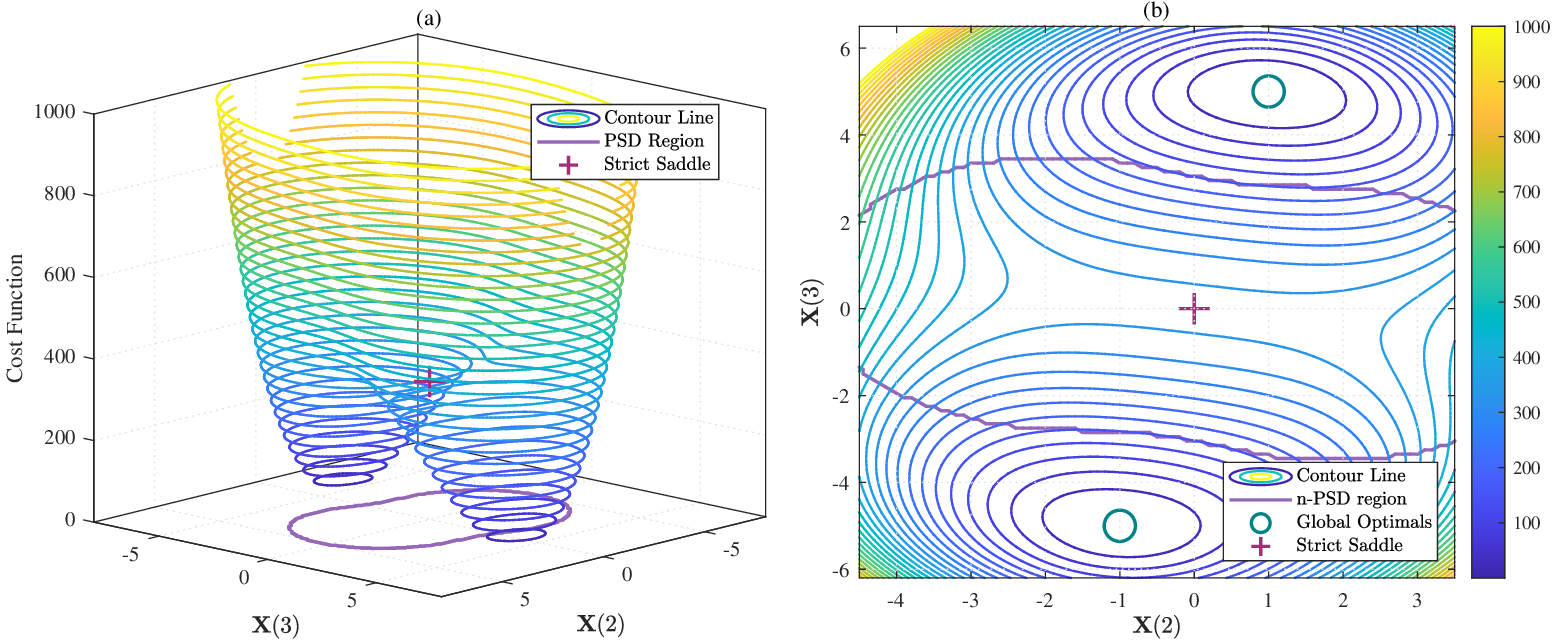}
	\caption{The landscape of $\Vert\mathbf{XX}^T-\mathbf{M}^{\star}\Vert_F^2$, where $\mathbf{M}^{\star}=\mathbf{X}^{\star}\mathbf{X}^{{\star}T}$ and $\mathbf{X}^{\star}=[0,1,5]^T$ or $[0,-1,-5]^T$. We plot the level contours along the last two dimensions. The non-PSD region is of the same meaning as in Fig. \ref{fig_toy_model}. It is known that under the rank $1$ setting there will be a visibly strongly convex region around the global minima\cite{LiSymmSaddleLandscape}.}
	\label{fig_originalMatrixFR}
\end{figure}
\subsection{The ``Rank Reduction" Routine}
\label{subsecIII-D_rankReduction}
The long-lasting race to find good initial points for the SNL has proceeded for decades\cite{KearsleySstress}. As aforementioned, the difficulty mainly lies in the sample model rather than the s-stress itself. Historically, commonly-used methods include MAP-MDS\cite{ImprovedMDSShang}, Biswas-Ye SDR\cite{BiswasTR_sdp}, randomized initialization, dimensionality relaxation\cite{FangEDMCNewton}. The dimensionality relaxation shares a notable resemblance to ``rank reduction", i.e., lifting or over-parameterization helps to solve the SNL problem. However, the original approach of dimensionality relaxation is much more complicated as it involves solving extra optimization problems, while our method is framed within either Riemannian optimization\cite{GaoRRMC}
\cite{ZhengRQCGtheory} or the recent advances in low-rank recovery\cite{ZhangPrecondGD}\cite{XuScaledGD}. Interested readers are referred to\cite{FangEDMCNewton} for a combined use of several strategies aforementioned and their performance in solving some hard molecule configuration problems. We next show the effect of over-parameterization.

\theoremstyle{definition}
\newtheorem{example}{Example}[section]
\begin{example}
	\vspace{-6pt}
	\label{exp_toy_example_rank1}
	Consider an EDM approximation problem with embedding dim one and consisting of three nodes\cite[Ex. 3.4]{ParhizkarPhDthsis} as shown in Fig. \ref{fig_toy_model}.  By Solving this toy-model from a random initial point $\mathbf{Y}_0\in \mathbb{R}^{3\times 2}_{*}$ using RTR solver in Manopt\cite{manopt} under $(\mathbb{R}^{n\times d}_{*}/\mathrm{O}(d),g^1)$, it returns a stationary point $\hat{\mathbf{Y}}_0$, then we use SVD to truncate it to rank one:
	\begin{equation*}
		\setlength\belowdisplayskip{5pt}
		\setlength\abovedisplayskip{5pt}
		\hat{\mathbf{Y}}_0 \hat{\mathbf{Y}}^{T}_0=\mathbf{U}_1\mathbf{\Sigma}_1\mathbf{U}_1^T+\mathbf{U}_2\mathbf{\Sigma}_2\mathbf{U}_2^T,\,\hat{\mathbf{Y}}_p=\mathbf{U}_1\sqrt{\mathbf{\Sigma}_1}.
	\end{equation*}
	We plot these $(\hat{\mathbf{Y}}_p,\bar{f}(\hat{\mathbf{Y}}_p))$ in Fig. \ref{fig_toy_model}(a), and these $\hat{\mathbf{Y}}_p$ seem to fall onto a line segment. Obviously, it's easier to converge when starting from $\hat{\mathbf{Y}}_p$.
\end{example}\vspace{-5pt}

\begin{figure}[!t]
	\centering
	\setlength{\abovecaptionskip}{0cm}
	\setlength{\belowcaptionskip}{-3cm}
	\includegraphics[width=9cm]{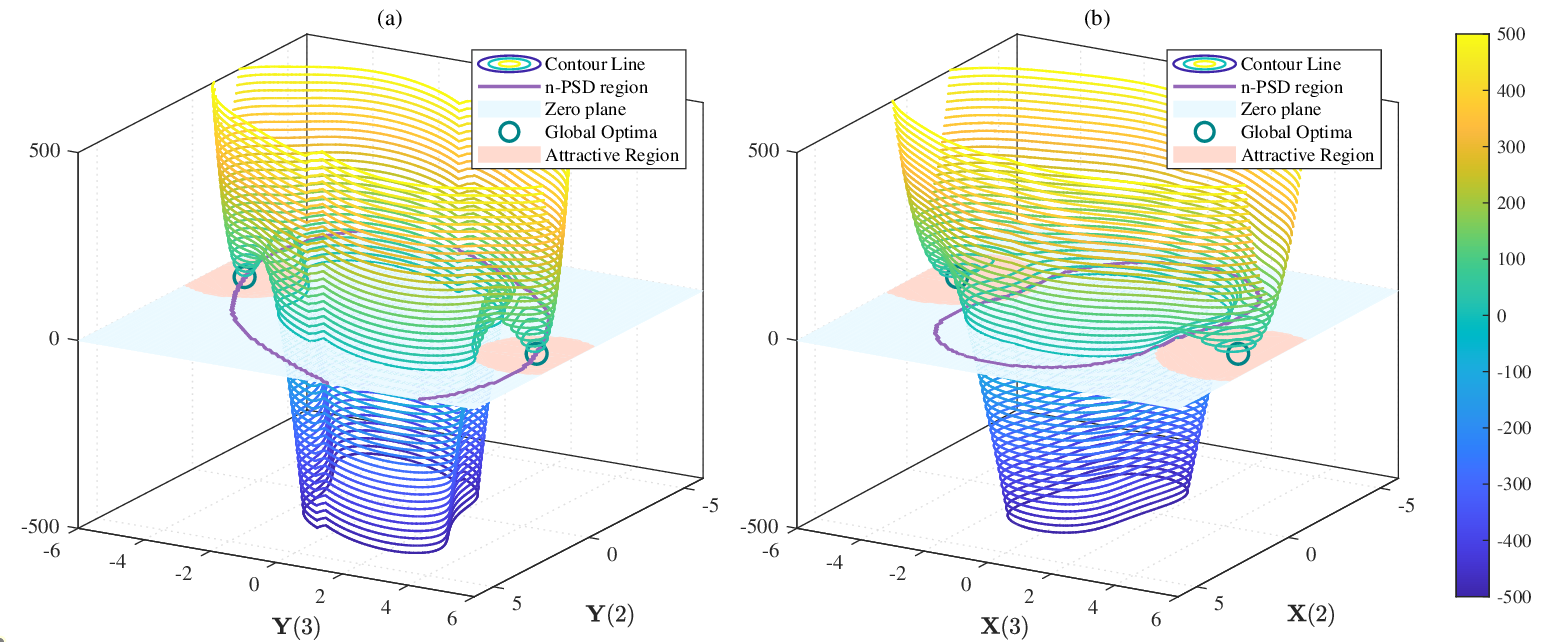}
	\caption{The attractive region in Example \ref{exp_toy_example_rank1} (a) and the vanilla matrix factorization problem (b) as in Fig. \ref{fig_originalMatrixFR}. The counter line is the value of $R=\Delta^T\nabla^2\bar{f}(\mathbf{Y})[\Delta]$. The attractive region (shown by the red circle on the zero plane) can be roughly regarded as the small neighborhood around the global optima insides which $R>0$\cite{LiSymmSaddleLandscape}\cite{MaImplicitRegularNonCVX_GD}. This characterization is stronger than what has been proved in Theorem \ref{them_distance_tirival_bounded}.}
	\label{fig_RIC_EDMC_originalMatrixFR}
\end{figure}
Fig. \ref{fig_toy_model} also illustrates the perturbation of the EDM mapping $g$ in \eqref{eq_Gram_to_EDM} to the original landscape when compared with the vanilla matrix factorization problem shown in Fig. \ref{fig_originalMatrixFR}. $g$ causes the strict saddle at the origin to break down into local maximum and symmetric saddles, which is analogous to the case when $\mathrm{rank}(\mathbf{M}^{\star})>1$\cite{LiSymmSaddleLandscape} or the phase retrieval under Gaussian ensembles\cite{SunQuWrightPRLandScape}. The most pathological part in Fig. \ref{fig_toy_model}(b) is that the non-PSD region almost reaches the global optima even in the rank one case, while the attractive region seems not to suffer much degradation when compared with the vanilla matrix factorization problem, as shown in Fig. \ref{fig_RIC_EDMC_originalMatrixFR}. The sharp convergence of RCG in Example \ref{exp_toy_example_rank1} when starting inside the non-PSD region (also outside the attractive region) with a ``correctly aligned" point $\hat{\mathbf{Y}}_p$ suggests that over-parameterization might be a lightweight way to greatly enhance the performance of vanilla first order methods when applied to the SNL problem. As for the algorithm side, the proposed ``rank reduction" routine (Algorithm \ref{alg_RankReduct_Initial}) is inspired by Huang et al.\cite{HuangPhaseLift} on solving phase retrieval using Riemannian strengthened Wirtinger Flow.
But as pointed out in \cite{PhaseretrievalWirtinger} and later works\cite{ChenYXPRRandominit}\cite[Lemma 1]{MaImplicitRegularNonCVX_GD}, the local landscape of phase retrieval under Gaussian measurements is quite benign, and VGD can converge even initialized randomly with a near-optimal sample complexity. While the non-convex SNL apparently relies on a dedicated initial design given a merely connected network. In Line 1 of Algorithm \ref{alg_RankReduct_Initial} we utilize SVD-MDS to bring $\mathbf{Y}_0$ close to a rank deficient point, then Algorithm \ref{alg_RCG_on_QuoM} seems to learn the correct rank via a currently unknown mechanism. This convergence is also robust to random initialization. As proved by Zheng et al.\cite{ZhengRQCGtheory}, metric $g^2$ results in better conditioned Riemannian Hessian than $g^1$ when iteration points reach the boundary of the quotient manifold, as can occur in over-parameterized scenarios. Thus, deploying $g^2$ in Line 4 leads to faster numerical convergence. We choose to use a ``clever" singular value shrinkage strategy in\cite{GaoRRMC} to reduce the rank as in Line 6 and stop Algorithm \ref{alg_RCG_on_QuoM} in Line 4 before it reaches a strictly rank-deficient point. Therefore, due to the perturbation introduced when truncating non-zero singular values, the ``rank reduction" exhibits a non-monotone descent behavior when stepping from $\hat{\mathbf{Y}}_0^k$ to $\hat{\mathbf{Y}}_p^r$.
\begin{algorithm}[!t]
	\caption{Rank Reduction}
	\label{alg_RankReduct_Initial}
	\begin{algorithmic}[1]
		\REQUIRE Real embedding dimension $d$, cost function $f([\mathbf{Y}])$ and its gradient as in \eqref{eq_transed_MFed_EDMC_FNM} \eqref{eq_transed_MFed_EDMC_FNM_EG_EH}.
		\STATE Use SVD-MDS\cite{ZhangSVD_MDS} to find $\mathbf{Y}_0\in\mathbb{R}^{n\times (d+2)}$:
		Use SVD to truncate $\frac{1}{p}\mathcal{P}_{\Omega}(\mathbf{D}_e)$ to rank $d+2$, $\hat{\mathbf{D}}=\mathcal{T}_{d+2}(\frac{1}{p}\mathcal{P}_{\Omega}(\mathbf{D}_e))$, then apply SVD to $-\frac{1}{2}\mathbf{J}\hat{\mathbf{D}}\mathbf{J}=\mathbf{Q}\boldsymbol{\Lambda}\mathbf{Q}^T$, $\mathbf{Y}_0=\mathbf{J}\mathbf{Q}\boldsymbol{\Lambda}^{1/2}$. 
		\STATE Set $k=d+2$, $\mathbf{Y}_0^k=\mathbf{Y}_0$.
		\WHILE{$k> d$}
		\STATE Call Algorithm \ref{alg_RCG_on_QuoM} with rank $k$, initial point $\mathbf{Y}_0^k$ and metric $g^2$, set $\mathrm{IMAX}=\mathrm{N}_1$, denote its output as $\hat{\mathbf{Y}}_0^k$.
		\STATE $[\mathbf{U}_{k},\mathbf{S}_{k},\mathbf{V}_{k}]=\mathrm{svds}(\hat{\mathbf{Y}}_0^k,k)$, $\mathbf{S}_k=\mathrm{diag}(s_1,\dots,s_k)$.
		\STATE Calculate the maximum singular value gap and its index $r$ as in \cite[Alg. 3]{GaoRRMC}: $r=\arg\max_i (s_i-s_{i+1})/s_i$.
		\STATE Shrink: $\hat{\mathbf{Y}}_p^{r}=\mathbf{U}_{k}(:,1:r)\mathbf{S}_{k}(1:r,1:r) \in \mathbb{R}^{n\times r}$
		\STATE Set $k=r$ and $\mathbf{Y}_0^k=\hat{\mathbf{Y}}_p^{r}$
		\ENDWHILE
		\STATE Call Algorithm \ref{alg_RCG_on_QuoM} with rank $d$, initial point $\mathbf{Y}_0^d$ and metric $g^1$, set $\mathrm{IMAX}=\mathrm{N}_2$, denote its output as $\hat{\mathbf{Y}}$.
		\RETURN $\hat{\mathbf{Y}}$
	\end{algorithmic}
\end{algorithm} 

The reason why ``rank reduction" shows tolerance to the unit ball sample model, i.e., the convergence on the lifted space given near-optimal sample complexity without a dedicated initial design, remains mystery. This routine also works well under the Bernoulli sampling scheme, but compared with SVD-MDS, it is time-consuming. Experiment results suggest that lifting $k$ to $d+2$ is enough. If one lifts the rank to be sufficiently large, as suggested by either EDM property\cite{DokmanicEDMTheory} or SDR theory\cite{LuoNonCVX_QCQP}, then the problem will become benign. But now it needs extra low-rank induction regularizer, i.e., using $(\mathbb{R}^{n\times d}_{*}/\mathrm{O}(d),g^2)$ alone is not strong enough to trigger convergence towards the correct rank within limited iterations. Recent advances\cite{ZhangPrecondGD}\cite{XuScaledGD} lead to guaranteed success of over-parameterization on matrix sensing model. We refer interested readers to their works and references within. 
\section{Outlier Control}
\label{secIV_OutlierQuestion}
We resort to Robust Matrix Completion (RMC)\cite{MaRPCA_RMCSurvey} techniques originally developed in Robust Principal Component Analysis\cite{CandesRPCCA} context to address the NLoS measurements.
\begin{equation}
	\setlength\belowdisplayskip{4pt}
	\setlength\abovedisplayskip{4pt}
	\label{eq_REE_MADMM_SNLproblem}
	\begin{gathered}
		\min_{\mathbf{Y}\in \mathbb{R}^{n\times d}} \Vert\mathcal{P}_{\Omega}(\mathbf{Z}-\mathbf{D}_e)\Vert_1
		+\frac{\lambda}{2}\Vert\mathcal{P}_{\bar{\Omega}}(g(\mathbf{YY}^T))\Vert_F^2 \\
		\mathrm{s.t.}\,g(\mathbf{YY}^T)=\mathbf{Z},\,[\mathbf{Y}]\in \mathcal{S}_{+}^{d,n}.
	\end{gathered}
\end{equation}
\eqref{eq_REE_MADMM_SNLproblem} can be viewed as a regularized version of LMaFit\cite{ShenLMaFit}. The term $\Vert\mathcal{P}_{\bar{\Omega}}(g(\mathbf{YY}^T))\Vert_F^2$ is used to prevent over fitting of noisy samples in $\mathcal{P}_{\Omega}(\mathbf{D}_e)$, where $\bar{\Omega}$ is the complement of $\Omega$. MADMM\cite{KovnatskyMADMM} is utilized to solve this. The augmented Lagrangian of \eqref{eq_REE_MADMM_SNLproblem} is given by\vspace{-6pt}
\begin{align}
	\label{eq_augmented_Lagrangian}
	\setlength\belowdisplayskip{-5pt}
	\setlength\abovedisplayskip{-5pt}
	\mathcal{L}_{\rho}(\mathbf{Z},\mathbf{Y},\mathbf{U})&=\Vert\mathcal{P}_{\Omega}(\mathbf{Z}-\mathbf{D}_e)\Vert_1
	+\frac{\lambda}{2}\Vert\mathcal{P}_{\bar{\Omega}}(g(\mathbf{YY}^T))\Vert_F^2\nonumber\\
	&+\frac{\rho}{2}\Vert g(\mathbf{YY}^T)-\mathbf{Z}+\mathbf{U}\Vert_F^2.\vspace{-12pt}
\end{align}
\eqref{eq_augmented_Lagrangian} results in a two-phrase ADMM with the primal and dual residuals given by
\begin{equation*}
	\setlength\belowdisplayskip{5pt}
	\setlength\abovedisplayskip{5pt}
	\mathbf{r}_k=g([\mathbf{Y}_k])-\mathbf{Z}_k,\,\mathbf{d}_k=\rho_k(g([\mathbf{Y}_{k-1}])-[\mathbf{Y}_k]).
\end{equation*} 
To simplify notions, we denote $\mathbf{YY}^T$ as $[\mathbf{Y}]$ later. $\mathbf{Z}$ subproblem corresponds to solving a proximal operator of the $l_1$ norm. The solution is given by\cite{ShenLMaFit}
\begin{equation}
	\label{eq_sol_ADMM_Z}
	\setlength\belowdisplayskip{6pt}
	\setlength\abovedisplayskip{6pt}
	\begin{gathered}
		\mathcal{P}_{\Omega}(\mathbf{Z}_{k+1})=\mathcal{P}_{\Omega}(\mathcal{S}_{\frac{1}{\rho}}(g([\mathbf{Y}_{k}])-\mathbf{D}_e+\mathbf{U}_k)+\mathbf{D}_e),\\
		\mathcal{P}_{\bar{\Omega}}(\mathbf{Z}_{k+1})=\mathcal{P}_{\bar{\Omega}}(g([\mathbf{Y}_{k}])+\mathbf{U}_k).
	\end{gathered}
\end{equation}  
Where $\mathcal{S}_{\frac{1}{\rho}}(x)=\mathrm{sgn}(x)\max(|x|-\frac{1}{\rho},0)$ is the element-wise soft-thresholding operator. $\mathbf{Y} $ subproblem is solved by Algorithm \ref{alg_RCG_on_QuoM}, its Euclidean gradient is given by
	\begin{gather*}
		\setlength\belowdisplayskip{5pt}
		\setlength\abovedisplayskip{5pt}
		\label{eq_EgradADMM_Y1}
		\mathbf{D}_c=g([\mathbf{Y}]),\,\mathbf{S}_{Y}=\mathbf{D}_c-(\mathbf{Z}_{k+1}-\mathbf{U}_k),\\
		\label{eq_EgradADMM_Y2}
		\nabla_{\mathbf{Y}} \mathcal{L}_\rho(\mathbf{Y})=2\lambda g^{*}(\mathcal{P}_{\bar{\Omega}}(\mathbf{D}_c))\mathbf{Y}+2\rho g^{*}(\mathbf{S}_{Y})\mathbf{Y}.
	\end{gather*}
A continuation technique on $\rho$ is enough to establish convergence. It starts with solving \eqref{eq_augmented_Lagrangian} using a relatively small $\rho_0$, then increase $\rho$ and solve again. We only update $\rho$ every $T_f$ step. A widely used stopping rule\cite{pmlr-v54-xu17a} is chosen, it gives
\begin{equation}
	\setlength\belowdisplayskip{5pt}
	\setlength\abovedisplayskip{5pt}
	\begin{gathered}
		\label{eq_ADMM_stopping_cria}
		\Vert \mathbf{r}_k \Vert_F \leq \epsilon^{\mathrm{tol}} \max\left\{ \Vert \mathbf{Z}_k\Vert_F, \Vert g([\mathbf{Y}_k]) \Vert_F \right\}, \\
		\Vert \mathbf{d}_k \Vert_F \leq \epsilon^{\mathrm{tol}} \max\left\{ \Vert \rho_k\mathbf{U}_k\Vert_F \right\}.
	\end{gathered}
\end{equation}
Algorithm \ref{alg_RADMM_SNL} shows the complete MADMM procedure. 
\section{Numerical experiments}
\label{secV_NumExps}
We present numerical results on a synthetic scene to illustrate the performance of the proposed algorithms. LRM-CG\cite{NguyenLRM-CG}, RTR\cite{MishraRieEDMC}, regular RCG, TNNR\cite{HuTNNR}, LMaFit\cite{ShenLMaFit}, MVU-SDP\cite{BiswasTR_sdp} and the ADMM version of trace minimization, BB-ADMM\cite{TasissaEDMCProof} are tested here. We use RHZLS, r-RHZLS, and RADMM to denote Algorithm \ref{alg_RCG_on_QuoM} with $g^1$, Algorithm \ref{alg_RankReduct_Initial} and \ref{alg_RADMM_SNL} in all experiments, respectively. For both RTR and LMaFit, we turn off the rank estimation since we assume the embedding dimension is known beforehand. For TNNR, we use the APGL approach since\cite{HuTNNR} suggests that it is more resistant to noisy samples. We also test SVD-MDS\cite{DrineasSVD_MDS} initialized RHZLS.
\subsection{Implementation Details}
\label{subsecV_I_ImplementDetails}
\begin{algorithm}[!t]
	\caption{Riemannian ADMM for Robust SNL}
	\label{alg_RADMM_SNL}
	\begin{algorithmic}[1]
		\REQUIRE Initial point $\mathbf{Y}_0$, initial $\rho_0$, regularization parameter $\lambda$, stopping tolerance $\epsilon^{\mathrm{tol}}$. 
		\STATE Initial ADMM variables as $\mathbf{U}_0=\mathcal{P}_{\Omega}(\mathbf{D}_e-g([\mathbf{Y}_0]))/\rho_0$,
		
		$\mathbf{Z}_0=g([\mathbf{Y}_0])$, $\mathcal{P}_{\Omega}(\mathbf{Z}_0)=\mathcal{P}_{\Omega}(\mathbf{D}_e)$.
		\FOR{$k=0,1,\dots,N$}
		\STATE Solve \eqref{eq_sol_ADMM_Z} to get $\mathbf{Z}_{k+1}$.
		\STATE Solve $\mathbf{Y}$ subproblem by Algorithm \ref{alg_RCG_on_QuoM} with $g^1$, $\mathrm{LS}_F=2$, and $\mathrm{IMAX}=2$ to get $\mathbf{Y}_{k+1}$.
		\STATE Update $\mathbf{U}_{k+1}$ by $\mathbf{U}_{k+1}=\mathbf{U}_k+g([\mathbf{Y}_{k+1}]-\mathbf{Z}_{k+1})$
		\IF{\eqref{eq_ADMM_stopping_cria} is satisfied}
		\RETURN $\mathbf{Y}_k$
		\ENDIF
		\IF{$\mathrm{mod}(k,T_f)=0$ and $\rho < \rho_{\mathrm{max}}$}
		\STATE $\rho_k=\tau\rho_k$,\,$\mathbf{U}_k=\mathbf{U}_k/\tau$.
		\ENDIF
		\ENDFOR
		\RETURN $\mathbf{Y}_k$
	\end{algorithmic}
\end{algorithm}
For the HZLS part, we employ the implementation provided in RMGLS\cite{SuttiRieHZLS}. The regular RCG and the RTR come from the default solvers in Manopt. The stopping criterion $F_{\mathrm{stop}}$ is triggered when the gradient tolerance $\Vert \boldsymbol{p}_k \Vert_F \leq \epsilon^{\mathrm{grad}}$ or the step length tolerance $\alpha_k\Vert \boldsymbol{\xi}_k \Vert_F \leq \epsilon^{\mathrm{ls}}$ is satisfied. For TNNR, LMaFit, and BB-ADMM, we use the implementation provided by their original authors. MVU-SDP formulation is solved using CVX\cite{cvx}.  
All non-manifold methods are carefully tuned to achieve the best performance. Experiments are run on a dual-socket Intel Xeon Gold 6226r server with 256 GB of RAM, Ubuntu 20.04.4, MATLAB 2022a. Under each different scenario setup, we run all these algorithms for 1000 independent trials. The synthetic scene used in this study is a square with a side length of $1$ and an embedding dimension of $2$. It has four anchor nodes located at $(-0.5,-0.5)$, $(-0.5,0.5)$, $(0.5,0.5)$, $(0.5,-0.5)$. $100$ sensor nodes are randomly dropped inside this square. The EDM is then sampled according to the unit ball rule. We assume that the distances between anchors are always known exactly (This means that we add a cliques structure formed by anchors into the sample mask $\Omega$ generated by the unit ball rule.).  Noisy samples of the EDM $d^e_{ij}$ are generated according to a widely used path loss model\cite{MaoSNLSurvey}:
\begin{equation*}
	\setlength\belowdisplayskip{4pt}
	\setlength\abovedisplayskip{4pt}
	d^e_{ij}=d_{ij}\exp\left\{-\frac{X_{\sigma}}{\eta\gamma}-\frac{\sigma^2}{2\eta^2\gamma^2}\right\},
\end{equation*}
where $\eta=\frac{10}{\ln 10},\,X_{\sigma}\sim \mathcal{N}(0,\,\sigma^2)$ is a random variable, and $d_{ij}$ is the ground truth distance. The path loss factor $\gamma$ is assumed to be $2$ in all experiments, and the value of $\sigma$ is varied for the noise of different intensities. We simply set $w_{ij}=\exp(-|d^e_{ij}-d_{ij}|^{1/4})$ for $(i,\,j)$ element in the weight matrix $\mathbf{W}$. Outliers are modeled as an additive sparse matrix $\mathbf{S}$, and its non-zero elements are generated from a uniform distribution between $[1,1+v_{\mathrm{out}}],\,v_{\mathrm{out}}>0$ as we only consider outliers caused by NLoS links. Let $p_{\mathrm{out}}=\Vert\mathcal{P}_{\Omega}(\mathbf{S})\Vert_0/|\Omega|$ denote the ratio of outliers. We assume that a noisy EDM should still be symmetric and have zero elements on its diagonal. The default parameter values for the proposed algorithms are shown in Table \ref{table_ParametersValue}. Where N.L. and N.Y. represent noiseless and noisy tests, respectively. For RTR, $\mathrm{IMAX}$ is set to $150$ since it has inner iterations.
The EDM recovery rate (RE) and Mean Square Localization Error (MSLE) are used to evaluate performance.
\begin{equation*}
	\setlength\belowdisplayskip{4pt}
	\setlength\abovedisplayskip{4pt}
	\label{eq_RE_and_MSLE}
	\mathrm{RE}=\frac{\Vert g(\hat{\mathbf{Y}}\hat{\mathbf{Y}}^T)-\mathbf{D}^{\star} \Vert_F}{\Vert\mathbf{D}^{\star}\Vert_F},\,
	\mathrm{MSLE}=\frac{\Vert\hat{\mathbf{Y}}_{\mathrm{UE}}-\mathbf{Y}_{\mathrm{UE}}^{\star}\Vert_F}{n-4}.
\end{equation*} 
Here, $\mathbf{D}^{\star}$ and $\mathbf{Y}_{\mathrm{UE}}^{\star}$ represent the ground truth EDM and sensor nodes' positions, respectively. The total number of nodes, $n$, is fixed at $104$ for all experiments. We also evaluate the Hessian matrix of \eqref{eq_transed_MFed_EDMC_FNM} at the point where the stopping criteria are triggered. Please see Appendix \ref{Appdi_A_Hess} for the actual Hessian formulation.
\subsection{Noiseless Scenario}
\label{subsecV_II_Noiseless}
In noiseless scenario we first focus on the impact of different initialization strategies and compare the behaviors of several powerful tools from optimization machinery. Fig. \ref{fig_result1_RE_grad_HessPSD} presents the percentage of $\Vert\nabla \bar{f}(\hat{\mathbf{Y}})\Vert_F<10^{-12}$ and getting a PSD Hessian matrix at the solution point $\hat{\mathbf{Y}}$ for RCG, RHZLS, RTR when random initialized, or using dedicated initialization strategies. Two confusing phenomena appear: (i) the RTR\footnote{The performance of RTR may be further enhanced by explicitly using the negative curvature direction to escape these saddles when the gradient vanishes. Please see \cite[Sec. 4]{SunQuWrightPRLandScape} for more discussions.} gets trapped at Euclidean saddle point in $45\%$ of all the trials when $r=0.3$; (ii) the SVD-MDS initialized RHZLS show non-monotone behavior from $r=0.3$ to $r=0.6$. These indicate that: (i) the landscape of s-stress under unit ball sampling rule is quite distinct from both the one under Bernoulli sampling scheme and the well-studied low-rank recovery problems\cite{LiSymmSaddleLandscape}\cite{RongGeSpuriousLocalMinima}, even though no spurious local minima have been found for it both numerically or theoretically\cite[Ch. 3]{ParhizkarPhDthsis}; (ii) the distortion bound $\Vert\frac{1}{p}\mathcal{P}_{\Omega}\mathbf{D}^{\star}-\mathbf{D}^{\star}\Vert$ is seriously biased under unit ball rule. 
We next provide the detailed comparison of all test algorithms in Fig. \ref{fig_result2_RE_MSLE_against_radioR}.
Both LaMFit and TNNR fail to achieve high recovery accuracy under low sample rates, while the MVU-SDP and BB-ADMM approach are the most resilient to radio coverage range decay as shown in Fig. \ref{fig_result2_RE_MSLE_against_radioR}(a), (b). But BB-ADMM only approximates the ground truth\footnote{BB-ADMM is tuned to be speed priority, even though it can reach high recovery rate after sufficient iterations (often as time-consuming as the MVU-SDP). The performance of trace minimization and MVU-SDP should be nearly equal, here the difference is caused by the ADMM formulation.}, which is natural since ADMM iterations are inescapable from jitter and non-monotonic decrease. 
While the MVU-SDP and BB-ADMM outperform r-RHZLS when $r$ is exceedingly low, our approach is computationally lighter, yields high accuracy solution, and requires only slightly larger $r$. In Fig. \ref{fig_result2_RE_MSLE_against_radioR}(c), we plot the success rate ($\mathrm{RE}<10^{-5}$) of all algorithms we tested against radio coverage range, and a rigidity lower bound drawn from 5000 numerical simulations is also included. 
The running times of all algorithms are plotted in Fig. \ref{fig_result2_RE_MSLE_against_radioR}(d). Overall, except MVU-SDP, all solvers perform far from the rigidity lower bound, the BB-ADMM can't attain such accuracy until $r>0.3$. It takes CVX about 10 seconds to find a feasible solution when $r=0.25$, and its computational cost doubles as the problem becomes well-conditioned. In contrast, r-RHZLS becomes quite competitive from both recovery rate and time cost as soon as $r\geq0.25$.
\begin{table}[t]
	\renewcommand{\arraystretch}{1.2}
	\vspace{-4pt}
	\caption{Default Value of the Parameters}
	\label{table_ParametersValue}
	\centering\vspace{-4pt}
	\scalebox{1}{\begin{tabular}{c c | c c | c c}
			\hline
			\bfseries Parameter & \bfseries Value & \bfseries Parameter & \bfseries Value & \bfseries Parameter & \bfseries Value\\
			\hline 
			Armijo $c_1$ & 0.5 & HZLS $\alpha_{\mathrm{max}}$ & 200 & HZLS $c_1$ & 0.1 \\
			HZLS $c_2$ & 0.1 & HZLS $\epsilon$ & $10^{-14}$ & \eqref{eq_HZLS_LSError_condition} $\omega$ & 0.005 \\
			\eqref{eq_HZLS_LSError_condition} $\Delta$ & 0.7 & Alg. \ref{alg_RankReduct_Initial} $\mathrm{N_1}$& 300 & Alg. \ref{alg_RankReduct_Initial} $\mathrm{N_2}$ & 300 \\
			$\mathrm{IMAX}$ & 600 & N.L. $\epsilon^{\mathrm{grad}}$ & $10^{-15}$ & N.Y. $\epsilon^{\mathrm{grad}}$ & $10^{-6}$ \\
			N.Y. $\epsilon^{\mathrm{ls}}$ & $10^{-10}$ & $\rho_0$ & $10^{-3}$ & $\lambda$ & $10^{-6}$ \\
			$\epsilon^{\mathrm{tol}}$ & 0.02 & Alg. \ref{alg_RADMM_SNL} $\mathrm{N}$ & 600 & $T_f$ & 2 \\
			$\tau$ & 1.05 & $\rho_{\mathrm{max}}$ & 100 &  &   \\
			\hline
	\end{tabular}}
\end{table}
\subsection{Noisy Scenario}
\label{subsec_V_III_noisySence}
When the distance measurements are corrupted by RSSI noise, position accuracy is considered instead of EDM recovery rate. All solvers except LMaFit\footnote{For LMaFit, we set the weight matrix to all ones, since its update structure does not allow one to change the sample operator.} share the same weight matrix as discussed in Section\ref{subsecV_I_ImplementDetails}.
Fig. \ref{fig_result3_RE_MSLE_against_noise_ADMM} shows the average MSLE performance for various noise intensities $\sigma$ and different values of $r$. 
As the radio outage range increases, all algorithms exhibit almost the same convergence behavior as in the noiseless cases, with r-RHZLS continuing to achieve the best MSLE performance.
\begin{figure}[!t]
	\raggedright
	\setlength{\abovecaptionskip}{0cm}
	\setlength{\belowcaptionskip}{-3cm}
	\includegraphics[width=9cm]{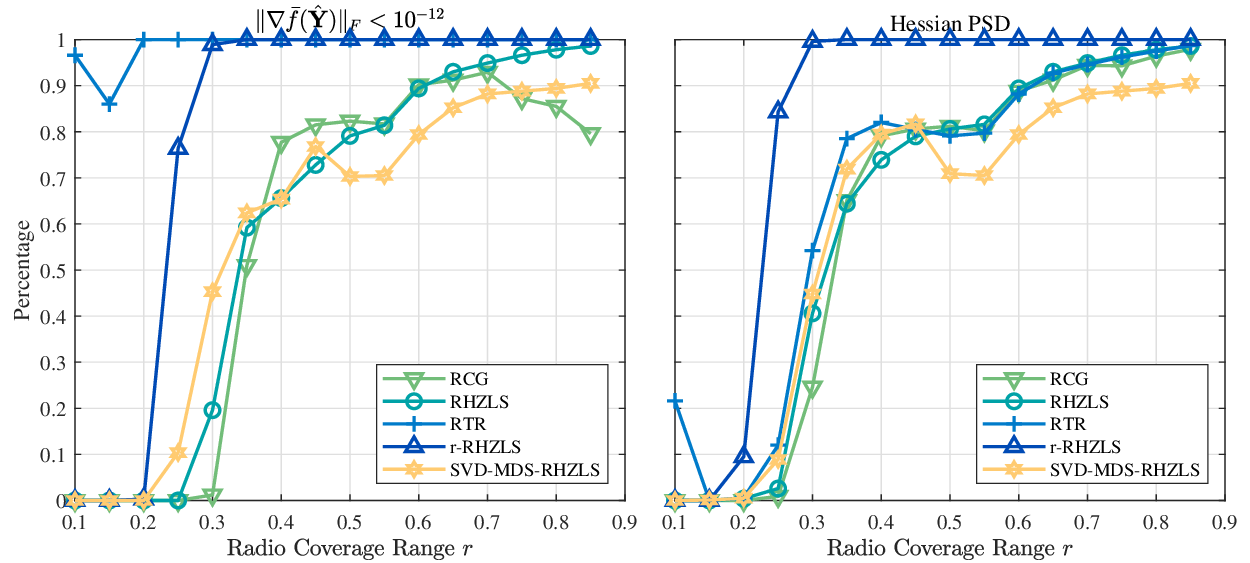}
	\caption{The percentage of achieving $\Vert\nabla \bar{f}(\hat{\mathbf{Y}})\Vert_F<10^{-12}$ and a PSD Hessian matrix at the solution point under different radio coverage range, the noiseless situation.}
	\label{fig_result1_RE_grad_HessPSD}
\end{figure}
\begin{figure*}[!h]
	\raggedright
	\setlength{\abovecaptionskip}{0cm}
	\setlength{\belowcaptionskip}{-3cm}
	\includegraphics[width=18cm]{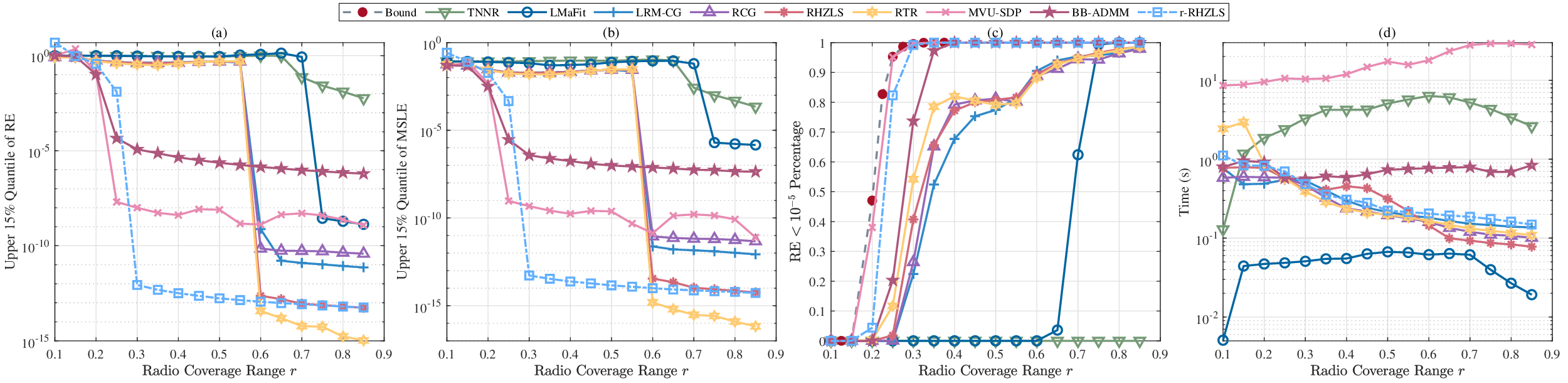}
	\caption{Upper $15\%$ quantile of RE and MSLE performance when varying radio coverage range, the noiseless situation (a)(b). The generically global rigidity lower bound and success rate of different algorithms are plotted in (c). The Average computation time costs are plotted in (d). The rigidity bound is obtained through a numerical procedure including QR decomposition and rank estimation\cite{SingerTheBound} instead of analytical expressions.}
	\label{fig_result2_RE_MSLE_against_radioR}
	\vspace{-5pt}
\end{figure*}
\begin{figure*}[!h]
	\raggedright
	\setlength{\abovecaptionskip}{0cm}
	\setlength{\belowcaptionskip}{-3cm}
	\includegraphics[width=18cm]{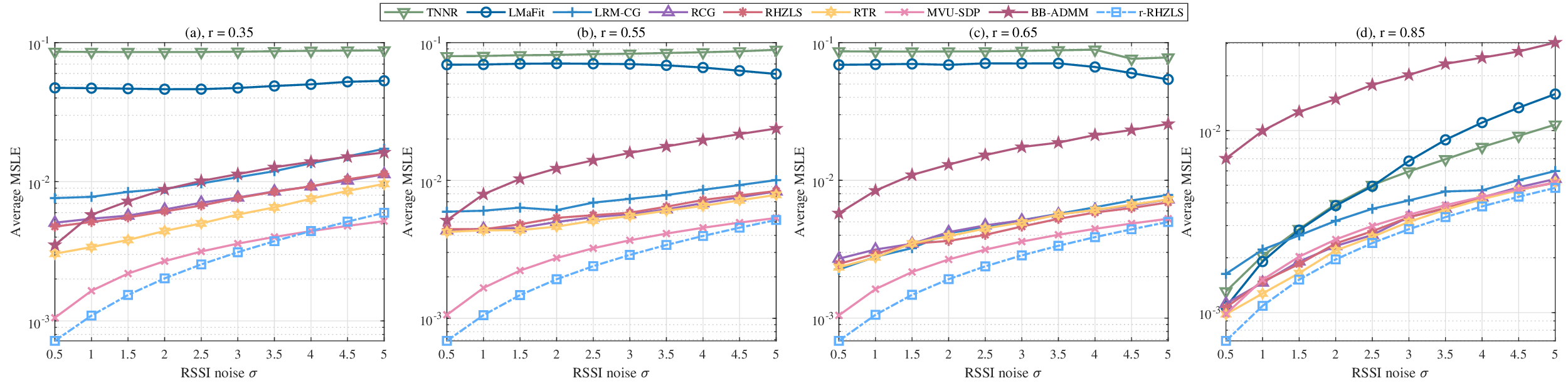}
	\caption{Average MSLE Performance when varying RSSI noise intensities and radio coverage range.}
	\label{fig_result3_RE_MSLE_against_noise_ADMM}
	\vspace{-10pt}
\end{figure*}
\subsection{Outliers Elimination}
\label{subsec_V_D_outliers}
We compare the MSLE and time performance of RADMM, SDP, and r-RHZLS in the presence of outliers in distance measurements in Fig. \ref{fig_result4_MSLE_2_ADMM}. When $p_{\mathrm{out}}<0.1$, the $l_1$ norm cost function in RADMM effectively identifies and eliminates incorrect measurements. RADMM is also robust to the changes in the value of outliers as long as $p_{\mathrm{out}}$ is less than a typical threshold. Also, its computational burden is light.
\begin{figure*}[!t]
	\raggedright
	\setlength{\abovecaptionskip}{0cm}
	\setlength{\belowcaptionskip}{-3cm}
	\includegraphics[width=18cm]{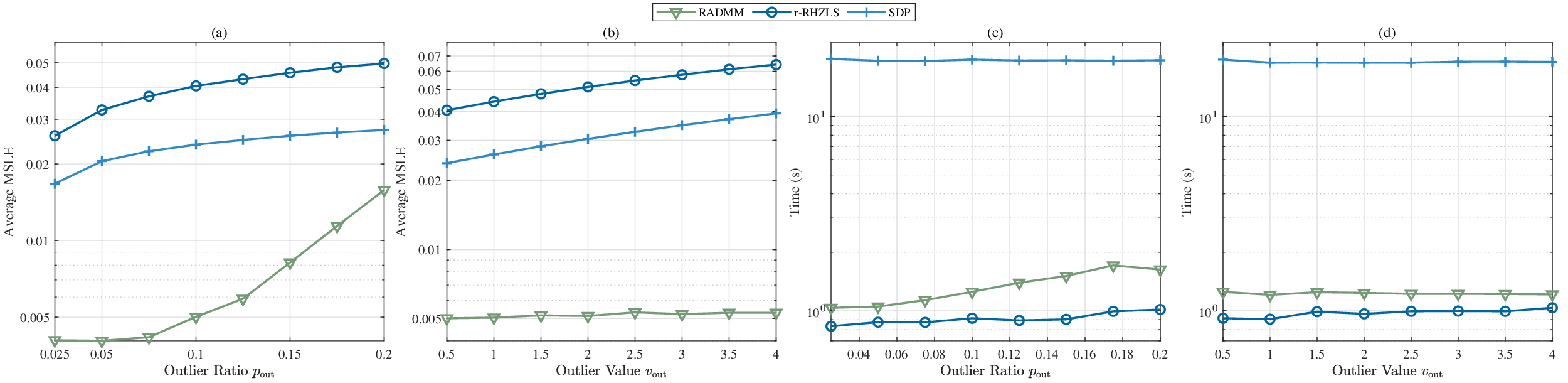}
	\caption{Average MSLE Performance and time cost of RADMM, r-RHZLS and SDP when varying outlier ratio (a)(c) and outlier value (b)(d). Radio coverage range is fixed to $0.35$ while RSSI noise $\sigma$ is fixed to $1$. For (a)(c) we set outlier value to $0.5$, for (b)(d) we set outlier ratio to $0.1$.\vspace{-5pt}}
	\label{fig_result4_MSLE_2_ADMM}
\end{figure*}
\section{Conclusion and Discussion}
\label{sec_conclusion}
This paper proposes a Riemannian Conjugate Gradient method with Hager-Zhang line search on a specific quotient manifold to solve the multi-hop distance-based Sensor Network Localization problem. 
A local attractive basin of the s-stress function under Bernoulli sampling model is analyzed for the first time. We conjecture that EDMC problems under structure-less sample masks can be effectively solved using spectral method as initialization, followed by simple first-order methods. 
A ``rank reduction" pre-process which facilitates the generation of a reasonable initialization point and improves the global convergence probability of first-order Riemannian optimization methods for the SNL problem is then introduced, which greatly enhances the practical performance of the RCG-HZLS approach. 
There are numerous open problems for future investigation, we list a few of them:
\begin{itemize}
	\item \textbf{Full characterization of the Bernoulli model}: It is known that vanilla gradient descent enjoys interesting properties on these non-convex statistical matrix factorization models\cite{MaImplicitRegularNonCVX_GD}\cite{ChenYXPRRandominit}. Fully characterizing the behavior of GD in solving EDMC problems with either spectral or random initial points is of future interest.
	\item \textbf{Sampling scheme}: Practical Euclidean distance problems, including molecule configuration and manifold learning, are based on unit ball sample model. 
	Analyzing the local landscape of the s-stress under unit ball sampling scheme may uncover the ill-posed nature of the SNL problem from non-convex algorithm side.
	\item \textbf{Spurious local minima}: Demonstrating the non-existence of spurious local minima of the s-stress function and applying saddle point escapable algorithms will lead to another way of theoretically guaranteed success of EDMC problems. This is tantamount to asking the global landscape of the s-stress function under both sampling schemes.
	\item \textbf{Over-parameterization}: The ``rank reduction" seems to be robust towards both sampling scheme and initial point. To fully characterize its behavior during the convergence to a rank deficient point under this set of RIP-(Restricted Isometry Property) less non-orthonormal basis and unit ball sampling model is meaningful but open.  
\end{itemize}
\section*{ACKNOWLEDGMENT}
\addcontentsline{toc}{section}{ACKNOWLEDGMENT}
The authors would like to thank the anonymous reviewers and the Associate Editor for their constructive comments that have helped to improve the presentation of this paper.
{\appendices
	\section{Hessian of S-stress function}
	\label{Appdi_A_Hess}
	We slightly modify the Hessian formulation in\cite{chuHessianSStress} to fit our problem. The Hessian is expressed in a blocked matrix with $n\times n$ blocks and each block is of size $d\times d$. 
	Assuming that the argument of s-stress function $\bar{f}_1$ is $\mathbf{P}=\mathbf{Y}^T=[\mathbf{p}_1,\mathbf{p}_2,\dots,\mathbf{p}_n]$, 
	then the partial gradient $\mathbf{g}_k=\frac{\partial \bar{f}_1}{\partial \mathbf{p}_k}$ is 
		\begin{subequations}
			\setlength\belowdisplayskip{3pt}
			\setlength\abovedisplayskip{3pt}
			\label{eq_App_Partial_Grad}
			\begin{gather}
				f_{ij}=
				\begin{cases}
					\Vert \mathbf{p}_i-\mathbf{p}_j\Vert_F^2-d^2_{ij}, & \text{if} (i,j)\in \Omega\\
					0,& \text{else}
				\end{cases}, \\
				\mathbf{p}_{ij}=
				\begin{cases}
					\mathbf{p}_i-\mathbf{p}_j, & \text{if} (i,j)\in \Omega\\
					\mathbf{0}_{d},&\text{else} \\
				\end{cases},\\
				\label{eq_App_Partial_Grad_3}
				\frac{\partial \bar{f}_1}{\partial \mathbf{p}_k}=4\sum_{
					\substack{j=1\\j\neq k}
				}^{n} \mathbf{p}_{kj}f_{kj},\,\nabla\bar{f}_1(\mathbf{P})=\left[\frac{\partial \bar{f}_1}{\partial \mathbf{p}_1},\dots,\frac{\partial \bar{f}_1}{\partial \mathbf{p}_n}\right],
			\end{gather}
		\end{subequations}
	where $\mathbf{0}_d$ is the zero vector in $\mathbb{R}^d$. Since Hessian matrix will only be examined in noiseless scenario, here we assume the weight matrix to be all ones and thus omitted. Then using \eqref{eq_App_Partial_Grad_3}, the $(k, i)$-block of the Hessian $\nabla^2 \bar{f}_1(\mathbf{P})$ is
			\begin{equation}
			\setlength\belowdisplayskip{4pt}
			\setlength\abovedisplayskip{4pt}
			\frac{\partial \mathbf{g}_k}{\partial \mathbf{p}_i}=
			\begin{cases}
				\sum^{n}_{j=1,j\neq k} [4f_{kj}\mathbf{I}_d+8\mathbf{p}_{kj}\mathbf{p}_{kj}^T],\,&\text{if }i=k\\
				4f_{kj}\mathbf{I}_d-8\mathbf{p}_{kj}\mathbf{p}_{kj}^T,\,&\text{if }i\neq k
			\end{cases},
		\end{equation}
	where $\mathbf{I}_d$ is the identity matrix of size $d\times d$. Recall \eqref{eq_transed_MFed_EDMC_FNM_EHess2}, it satisfies 
	\begin{equation*}
		\setlength\belowdisplayskip{4pt}
		\setlength\abovedisplayskip{4pt}
		\mathrm{tr}(\mathrm{vec}^T(\mathbf{V}) \nabla^2 \bar{f}_1(\mathbf{P}) \mathrm{vec}(\mathbf{V}))=\mathrm{tr}(\mathbf{Z}^T \nabla^2\bar{f}(\mathbf{Y})[\mathbf{Z}]),
	\end{equation*}
	where $\mathbf{Z}=\mathbf{V}^T$. Thus, we are able to check whether
	\begin{equation*}
		\setlength\belowdisplayskip{4pt}
		\setlength\abovedisplayskip{4pt}
		\forall\, \mathbf{Z}\in\mathbb{R}^{n\times d},\,\mathrm{tr}(\mathbf{Z}^T \nabla^2\bar{f}(\mathbf{Y})[\mathbf{Z}])\geq 0,
	\end{equation*}
	holds through checking $\nabla^2 \bar{f}_1(\mathbf{P})$ is a PSD matrix or not. 
	\section{Proof of Theorem \ref{them_distance_tirival_bounded}}
	\label{Appdi_B_Bound}
	\newtheorem{lemma}{Lemma}[section]
	We first introduce some notions. For the fixed rank $d$ Gram matrix $\mathbf{G}^{\star}=\mathbf{Y}^{\star}\mathbf{Y}^{\star T}$, let $\mathbf{U}^{\star}\boldsymbol{\Sigma}^{\star}\mathbf{U}^{\star T}$ denote its thin SVD. The tangent space and norm space at $\mathbf{G}^{\star}$ of $\mathcal{S}^{d,n}_+$ under embedding geometry
	\cite[Ch. 7]{BoumalIntrotoMani} are given by \eqref{eq_embedded_space}, and the projections onto these two spaces are denoted by $\mathcal{P}_{\mathbb{T}}$ and $\mathcal{P}_{\mathbb{T}^{\perp}}$, respectively.
	\begin{subequations}
		\setlength\belowdisplayskip{4pt}
		\setlength\abovedisplayskip{4pt}
		\label{eq_embedded_space}
		\begin{gather}
			\label{eq_tangent_space}
			\mathbb{T}=T_{\mathbf{G}^{\star}}\mathcal{S}^{d,n}_+=\{\mathbf{U}^{\star}\mathbf{W}_1^T+\mathbf{W}_1\mathbf{U}^{\star T}\},\\
			\label{eq_norm_space}
			\mathbb{T}^{\perp}=N_{\mathbf{G}^{\star}}\mathcal{S}^{d,n}_+=\{\mathbf{U}^{\star}_{\perp}\mathbf{W}_2\mathbf{U}^{\star T}_{\perp}\},\\
			\label{eq_tangent_proj}
			\mathcal{P}_{\mathbb{T}}(\mathbf{X})=\mathcal{P}_{\mathbf{U}}\mathbf{X}+\mathbf{X}\mathcal{P}_{\mathbf{U}}-\mathcal{P}_{\mathbf{U}}\mathbf{X}\mathcal{P}_{\mathbf{U}},
		\end{gather}
	\end{subequations}
	where $\mathcal{P}_{\mathbf{U}}=\mathbf{U}^{\star}\mathbf{U}^{\star T}$. Next, we rescale the gradient \eqref{eq_transed_MFed_EDMC_FNM_Egrad2} and reuse $f=\bar{f}$ to denote the lifted function defined on $\mathbb{R}^{n\times d}_*$ for the sake of simplicity
	\begin{equation}
		\setlength\belowdisplayskip{3pt}
		\setlength\abovedisplayskip{3pt}
		\label{eq_EDMC_sstress_grad}
		\nabla f(\mathbf{Y})=\frac{2}{p}g^*(\mathcal{P}_{\Omega}\circ g(\mathbf{YY}^T-\mathbf{Y}^{\star}\mathbf{Y}^{\star T}))\mathbf{Y},
	\end{equation}
	where we set $\mathbf{W}=\mathbf{11}^T$ since noiseless. 

	\subsection{Restricted Strong Convexity}
	\label{Appdi_B1_Bound}
	The proof is divided into three parts: (i) We show equivalence between choosing $\mathbb{I}=[n]^2$ and $\mathbb{I}=\mathbb{L}=\{(i,j):1\leq i\leq j\leq n\}$ under Bernoulli model. Since any EDM is a hollow matrix, one only needs to consider the off-diagonal samples. Then this claim follows directly from the decoupling method\cite[Thm. D.2]{GeNIPS2016_7fb8ceb3} and omits here for brevity. (ii) Modify what has been developed in\cite{ZhengLaffertyNonCVXFR}\cite{SL15NonCVXFR}, i.e., lower bounding $\langle\nabla \bar{f}(\mathbf{Y}),\Delta\rangle$ using components belonging to $\mathbb{T}$ and $\mathbb{T}^{\perp}$, we concisely revisit it here for completeness. (iii) Develop new bounds for these two components respectively. The proof is then concluded with some discussions. Define $\mathbf{YY}^T-\mathbf{G}^{\star}=\mathbf{Z}$ and $\mathcal{P}_{\Omega}\circ g=\mathcal{R}_{\Omega}$, $\mathrm{LHS}:=\langle\nabla f(\mathbf{Y}),\Delta\rangle$ we have\vspace{-3pt}
	\begin{align}
		\setlength\belowdisplayskip{-10pt}
		\setlength\abovedisplayskip{3pt}
		\label{eq_expanding_TN_Copments}
		&\mathrm{LHS}=\langle\frac{2}{p}g^*(\mathcal{P}_{\Omega}\circ g(\mathbf{YY}^T-\mathbf{Y}^{\star}\mathbf{Y}^{\star T}))\mathbf{Y},\mathbf{Y}-\mathbf{Y}^{\star}\boldsymbol{\psi}^{\star}\rangle \nonumber\\
		&=\langle\frac{1}{p}\mathcal{P}_{\Omega}\circ g(\mathbf{YY}^T-\mathbf{Y}^{\star}\mathbf{Y}^{{\star}T}),g(\Delta\mathbf{Y}^T+\mathbf{Y}\Delta^T)\rangle \\
		&\overset{(i)}{=}\langle\frac{1}{p}\mathcal{R}_{\Omega}(\mathbf{YY}^T-\mathbf{G}^{\star}),\mathcal{R}_{\Omega}(\mathbf{YY}^T-\mathbf{G}^{\star}+\Delta\Delta^T)\rangle\nonumber\\
		&\overset{(ii)}{=}\langle\frac{1}{p}\mathcal{R}_{\Omega}(\mathbf{Z}_{\mathbb{T}}+\mathbf{Z}_{\mathbb{T}}^{\perp}),\mathcal{R}_{\Omega}(\Delta\bar{\mathbf{Y}}^{{\star}T}+\bar{\mathbf{Y}}^{\star}\Delta^T+2\Delta\Delta^T)\rangle,\nonumber
	\end{align}
	where $\Delta$ and $\boldsymbol{\psi}^{\star}$ are defined in Lemma \ref{eq_geodesic_distance} and $\bar{\mathbf{Y}}^{\star}=\mathbf{Y}^{\star}\boldsymbol{\psi}^{\star}$. $(i)$ from \eqref{eq_distance_rela_4} and
	$(ii)$ by noticing that $\mathbf{Z}_{\mathbb{T}}=\mathcal{P}_{\mathbb{T}}\mathbf{Z}=\Delta\bar{\mathbf{Y}}^{{\star}T}+\bar{\mathbf{Y}}^{\star}\Delta^T$ and $\mathbf{Z}_{\mathbb{T}}^{\perp}=\mathcal{P}_{\mathbb{T}^{\perp}}\mathbf{Z}=\Delta\Delta^T$. By expanding $\mathrm{LHS}$, we have\vspace{-3pt}
	\begin{align}
		\setlength\belowdisplayskip{-3pt}
		\setlength\abovedisplayskip{3pt}
		\label{eq_EDMC_sstress_regular_cvx2}
		\mathrm{LHS}&=\frac{1}{p}\Vert\mathcal{R}_{\Omega}\mathbf{Z}_{\mathbb{T}}\Vert_F^2+\frac{2}{p}\Vert\mathcal{R}_{\Omega}\mathbf{Z}_{\mathbb{T}}^{\perp}\Vert_F^2+\frac{3}{p}\langle\mathcal{R}_{\Omega}\mathbf{Z}_{\mathbb{T}},\mathcal{R}_{\Omega}\mathbf{Z}_{\mathbb{T}}^{\perp}\rangle\nonumber \\
		&\overset{(i)}{\geq}\frac{1}{2p}\underbrace{\Vert\mathcal{R}_{\Omega}\mathbf{Z}_{\mathbb{T}}\Vert_F^2}_{\zeta_1}-\frac{5}{2p}\underbrace{\Vert\mathcal{R}_{\Omega}\mathbf{Z}_{\mathbb{T}}^{\perp}\Vert_F^2}_{\zeta_2},\vspace{-5pt}
	\end{align}
	where $(i)$ follows from \cite[App. C.1]{ZhengLaffertyNonCVXFR}. To show \eqref{eq_EDMC_sstress_regular_cvx2} can be bound away from $0$, we need a lower bound on $\zeta_1$ and an upper bound on $\zeta_2$. 

	To bound $\zeta_1$, we first adopt the ``dual basis" representation method in\cite{TasissaEDMCProof}. Note that  $\mathcal{R}_{\Omega}\mathbf{X}=\sum_{\boldsymbol{\alpha}\in\mathbb{I}}\delta_{\boldsymbol{\alpha}}\langle\mathbf{X},\boldsymbol{\omega}_{\boldsymbol{\alpha}}\rangle\mathbf{e}_i\mathbf{e}_j^T$, where $\delta_{\boldsymbol{\alpha}}$ and $\boldsymbol{\omega}_{\boldsymbol{\alpha}}$ are defined as in Section \ref{subsec_math_setup}. 
	Minor linear algebra shows that $g^*g=\sum_{\boldsymbol{\alpha}\in\mathbb{I}}\langle\cdot,\boldsymbol{\omega}_{\boldsymbol{\alpha}}\rangle\boldsymbol{\omega}_{\boldsymbol{\alpha}}$ and
	\begin{equation}
		\setlength\belowdisplayskip{3pt}
		\setlength\abovedisplayskip{3pt}
		\label{eq_EDMC_1pROmegROmeg_opt}
		\frac{1}{p}\mathcal{R}_{\Omega}^*\mathcal{R}_{\Omega}\mathbf{X}=\frac{1}{p}\sum_{\boldsymbol{\alpha}\in\mathbb{I}}\delta_{\boldsymbol{\alpha}}\langle\mathbf{X},\boldsymbol{\omega}_{\boldsymbol{\alpha}}\rangle\boldsymbol{\omega}_{\boldsymbol{\alpha}}.
	\end{equation}
	\eqref{eq_EDMC_1pROmegROmeg_opt} is the same as the ``restricted frame operator" defined in\cite{TasissaEDMCProof}.
	We now show the following two-side distortion bound holds with high probability. 
\begin{lemma}
	\label{lemma_restricted_region_inco_T_RIP}
	Under standard incoherence assumption, i.e., $\Vert\mathbf{U}^{\star}\Vert_{2,\infty}^2\leq\frac{\mu d}{n}$ for $\epsilon>\sqrt{\frac{10240\beta(\mu d)^{3}\log n}{3np}}$,
	\begin{equation}
		\setlength\belowdisplayskip{3pt}
		\setlength\abovedisplayskip{3pt}
		\label{eq_restricted_region_inco_T_RIP}
		\Vert\frac{1}{p}\mathcal{P}_{\mathbb{T}}\mathcal{R}_{\Omega}^*\mathcal{R}_{\Omega}\mathcal{P}_{\mathbb{T}}-\mathcal{P}_{\mathbb{T}}g^*g\mathcal{P}_{\mathbb{T}}\Vert\leq \epsilon<1
	\end{equation} 
	holds with probability at least $1-n^{1-\beta}$ as soon as $p\geq C_T\beta(\mu d)^{3}\log n/n$ for sufficient large constant $C_T$ and $\beta>1$. The proof is referred to Appendix \ref{subsec_ProofofLemmaA1}.
\end{lemma}
	We next bound $\zeta_2$. Inspired by\cite[Lemma 9]{ZhengLaffertyNonCVXFR} and\cite[Prop. 4.3]{SL15NonCVXFR}, we have the following estimate.
	\begin{lemma}
		\label{eq_NromalSpace_tight_bound_randomGraph}
		If the sample complexity $p\geq \frac{c_r \log n}{\delta_r^2 n}$ for some constant $c_r>3$ and $\delta_r\in(0,1]$, then with probability at least $1-\frac{1}{2}n^{-4}-2n^{-8}$, uniformly for all $\Delta\in\mathbb{R}^{n\times d}$, it holds that
		
		\begin{small} 
		\begin{equation*}
			\setlength\belowdisplayskip{3pt}
			\setlength\abovedisplayskip{-3pt}
			\frac{1}{p}\Vert\mathcal{R}_{\Omega}\Delta\Delta^T\Vert_F^2\leq \left[(16n+c_g\sqrt{\frac{n}{p}})\Vert\Delta\Vert_{2,\infty}^2+8(1+\delta_r)\Vert\Delta\Vert_F^2\right]\Vert\Delta\Vert_F^2
		\end{equation*}
		\end{small}
		for some constant $c_g$ independent of $n,\,d$.
	\end{lemma}
	\begin{proof}
		By splitting $\frac{1}{p}\Vert\mathcal{R}_{\Omega}\Delta\Delta^T\Vert_F^2=\frac{1}{p}\sum_{\boldsymbol{\alpha}\in\Omega}\langle\mathbf{z}_{\boldsymbol{\alpha}}^T\Delta,\mathbf{z}_{\boldsymbol{\alpha}}^T\Delta\rangle^2$ into square and cross terms, we have
		\begin{align*}
			\setlength\belowdisplayskip{2pt}
			\setlength\abovedisplayskip{2pt}
			&\frac{1}{p}\Vert\mathcal{R}_{\Omega}\Delta\Delta^T\Vert_F^2
			\leq\frac{1}{p}\sum_{\boldsymbol{\alpha}\in\Omega}\Vert\mathbf{z}_{\boldsymbol{\alpha}}^T\Delta\Vert_2^2\Vert\mathbf{z}_{\boldsymbol{\alpha}}^T\Delta\Vert_2^2\\
			&\overset{(i)}{\leq} \frac{1}{p}\sum_{\boldsymbol{\alpha}\in\Omega}4(\Vert\mathbf{e}_i^T\Delta\Vert_2^2+\Vert\mathbf{e}_j^T\Delta\Vert_2^2)^2\\
			&=\underbrace{\frac{4}{p}\sum_{\boldsymbol{\alpha}\in\Omega}(\Vert\mathbf{e}_i^T\Delta\Vert_2^4+\Vert\mathbf{e}_j^T\Delta\Vert_2^4)}_{\eta_1}+\underbrace{\frac{8}{p}\sum_{\boldsymbol{\alpha}\in\Omega}\Vert\mathbf{e}_i^T\Delta\Vert_2^2\Vert\mathbf{e}_j^T\Delta\Vert_2^2}_{\eta_2},
		\end{align*}
		where $(i)$ from $\Vert\mathbf{z}_{\boldsymbol{\alpha}}^T\Delta\Vert_2^2=\Vert(\mathbf{e}_i-\mathbf{e}_j)^T\Delta\Vert_2^2\leq2(\Vert\mathbf{e}_i^T\Delta\Vert_2^2+\Vert\mathbf{e}_j^T\Delta\Vert_2^2)$. $\eta_1$ can be bounded by $l_{2,\infty}$ analysis. Separating the sample set $\Omega$ into rows, i.e., $\Omega=[\Omega_1^T,\dots,\Omega_n^T]^T$, we have
		\begin{align}
			\setlength\belowdisplayskip{2pt}
			\setlength\abovedisplayskip{2pt}
			\label{eq_normal_space_comp_eta_1}
			\eta_1&=\frac{4}{p}\sum_{i=1}^n\sum_{j\in\Omega_i}\Vert\mathbf{e}_i^T\Delta\Vert_2^4+\Vert\mathbf{e}_j^T\Delta\Vert_2^4\overset{(i)}{\leq}\frac{8}{p}\sum_{i=1}^n\sum_{j\in\Omega_i}\Vert\mathbf{e}_i^T\Delta\Vert_2^4\nonumber\\
			&\leq\frac{8}{p}\sum_{i=1}^n\sum_{j\in\Omega_i}\Vert\mathbf{e}_i^T\Delta\Vert_2^2\Vert\Delta\Vert_{2,\infty}^2\overset{(ii)}{\leq} 16n\Vert\Delta\Vert_{2,\infty}^2\sum_{i=1}^n\Vert\mathbf{e}_i^T\Delta\Vert_2^2\nonumber\\
			&\overset{(iii)}{=}16n\Vert\Delta\Vert_{2,\infty}^2\Vert\Delta\Vert_F^2,
		\end{align}
		where $(i)$ by symmetry, i.e., one can separate $\Omega$ by column to bound the $\mathbf{e}_j$ term, and $(ii)$ followed by Chernoff inequality. It is straightforward to check $|\Omega_i|\leq 2np$ holds with probability at least $1-2e^{-cnp}$. $(iii)$ follows from an union bound. For $p>c\log n/n$, $c>3$, the probability of failure when controlling $\eta_1$ is less or equals to $1-2n^{-8}$. The bound on $\eta_2$ is a direct use of Lemma \ref{random_graph_lemma} by setting $\mathbf{x}=\mathbf{y}=[\Vert\mathbf{e}_1^T\Delta\Vert_2^2,\dots,\Vert\mathbf{e}_n^T\Delta\Vert_2^2]^T$, we have
		\begin{equation*}
			\setlength\belowdisplayskip{3pt}
			\setlength\abovedisplayskip{3pt}
			\eta_2\leq8(1+\delta_r)\Vert\Delta\Vert_F^4+c_g\sqrt{\frac{n}{p}}\Vert\Delta\Vert_{2,\infty}^2\Vert\Delta\Vert_F^2,
		\end{equation*}
		holds with probability at least $1-\frac{1}{2}n^{-4}$ as soon as $p\geq \frac{c_r \log n}{\delta_r^2 n}$. And we conclude the proof.
	\end{proof}

	What remains is straightforward. Following a standard derivation detailed in Appendix \ref{subsec_ProofofLemmaA1} and setting $\epsilon=1/2$ in Lemma \ref{lemma_restricted_region_inco_T_RIP}, we have
	\begin{equation}
		\label{eq_finallBound_on_TSpace_Comp}
		\setlength\belowdisplayskip{3pt}
		\setlength\abovedisplayskip{3pt}
		\frac{7}{2}\Vert\mathbf{Z}_{\mathbb{T}}\Vert_F^2\leq(4-\epsilon)\Vert \mathbf{Z}_{\mathbb{T}}\Vert_F^2\overset{(i)}{\leq}\frac{1}{p}\Vert\mathcal{R}_{\Omega}\mathbf{Z}_\mathbb{T}\Vert_F^2,
	\end{equation}
	where $(i)$ by using the fact that the smallest eigenvalue of $g^*g$ is $4$\cite[Cor. 2.2]{LichtenbergTasissaEDMC}\footnote{According to \cite{LichtenbergTasissaEDMC}, $g^*g$ has only three distinct eigenvalues: $4$, $2n$, $4n$.}. 
	For convenience of analysis, we define event $\mathrm{E}_1:=\{\Delta\in\mathcal{B},\,p\geq C_D\mu^2 d^2\log n/(\delta_r^2 n)\}$ for some sufficient large constant number $C_D$, and $\delta_r$ (which is defined in Lemma \ref{eq_NromalSpace_tight_bound_randomGraph}).  Here
	\begin{equation}
		\label{eq_theB_region_and_sampleComplex}
		\setlength\belowdisplayskip{3pt}
		\setlength\abovedisplayskip{3pt}
		\mathcal{B}:=\{\Delta\,|\,\Vert\Delta\Vert_F^2\leq\frac{\sigma_d^{\star}}{120},\,\Vert\Delta\Vert_{2,\infty}^2\leq\frac{\mu d\sigma_1^{\star}\delta_t}{16\kappa n}\},
	\end{equation}
	where $\delta_t>0$ is a constant. Conditioned on $\mathrm{E}_1$, it holds that
	\begin{equation*}
		\setlength\belowdisplayskip{3pt}
		\setlength\abovedisplayskip{3pt}
		\rho_1:=(16n+c_g\sqrt{\frac{n}{p}})\Vert\Delta\Vert_{2,\infty}^2\leq \mu d\sigma_d^{\star}\delta_t+\sqrt{\frac{c_g^2\sigma_1^{{\star}2}\delta_r^2\delta_t^2}{C_D \kappa^2\log n}}.
	\end{equation*}
	For $C_D\gg c_g$, this gives $\rho_1\leq(\delta_r+\mu d)\sigma_d^{\star}\delta_t\leq2\mu d\sigma_d^{\star}\delta_t$ since $\delta_r\leq1$, $\mu d \geq1$\cite{LRMC_Can1}\cite{ChenIncoOptimalMC}. Substituting \eqref{eq_theB_region_and_sampleComplex}, \eqref{eq_finallBound_on_TSpace_Comp}, bound on $\rho_1$, and Lemma \ref{eq_NromalSpace_tight_bound_randomGraph} into \eqref{eq_EDMC_sstress_regular_cvx2}, we have
	\begin{align}
		\label{eq_final_the_bound}
		\setlength\belowdisplayskip{3pt}
		\setlength\abovedisplayskip{3pt}
		\mathrm{LHS}&\overset{(i)}{\geq} \frac{7}{4}\Vert\mathbf{Z}_{\mathbb{T}}\Vert_F^2-(5\mu d\sigma_d^{\star}\delta_t-\frac{(1+\delta_r)\sigma_d^{\star}}{6})\Vert\Delta\Vert_F^2\nonumber\\
		&\overset{(ii)}{\geq} (\frac{7}{2}-1-\frac{1}{5})\sigma_d^{\star}\Vert\Delta\Vert_F^2=\frac{23}{10}\sigma_d^{\star}\Vert\Delta\Vert_F^2,
	\end{align}
	where $(i)$ from setting $\Vert\Delta\Vert_F^2\leq\frac{\sigma_d^{\star}}{120}$ then substituting the bound on $\rho_1$, and $(ii)$ by first expanding the term $\Vert\mathbf{Z}_{\mathbb{T}}\Vert_F^2=\Vert\Delta\bar{\mathbf{Y}}^{{\star}T}+\bar{\mathbf{Y}}^{\star}\Delta^T\Vert_F^2$
	\begin{align*}
		\setlength\belowdisplayskip{3pt}
		\setlength\abovedisplayskip{3pt}
		\Vert\Delta\bar{\mathbf{Y}}^{{\star}T}+\bar{\mathbf{Y}}^{\star}\Delta^T\Vert_F^2&=2\Vert\bar{\mathbf{Y}}^{\star}\Delta^T\Vert_F^2+\mathrm{tr}(\bar{\mathbf{Y}}^{\star}\Delta^T\bar{\mathbf{Y}}^{\star}\Delta^T)\\
		&\overset{(a)}{=}2\Vert\bar{\mathbf{Y}}^{\star}\Delta^T\Vert_F^2+\mathrm{tr}(\bar{\mathbf{Y}}^{\star}\bar{\mathbf{Y}}^{\star T}\Delta\Delta^T)\\
		&\overset{(b)}{\geq}2\Vert\bar{\mathbf{Y}}^{\star}\Delta^T\Vert_F^2\geq 2\sigma_d^{\star}\Vert\Delta\Vert_F^2,
	\end{align*} 
	where $(a)$ from Lemma \ref{eq_geodesic_distance}, and $(b)$ by noticing that $\mathrm{tr}(\bar{\mathbf{Y}}^{\star}\bar{\mathbf{Y}}^{\star T}\Delta\Delta^T)=\Vert\bar{\mathbf{Y}}^{\star T}\Delta\Vert_F^2\geq 0$, and then setting $\delta_t=\frac{1}{5\mu d},\,\delta_r=\frac{1}{5}$, we get \eqref{eq_final_the_bound}.
	
	Thus, Theorem \ref{them_distance_tirival_bounded}-(1) holds with probability at least $1-n^{1-\beta}-\frac{1}{2}n^{-4}-2n^{-8}$ under 
	\begin{equation*}
		\setlength\belowdisplayskip{3pt}
		\setlength\abovedisplayskip{3pt}
		p\geq \max\{C_T\beta(\mu d)^3\log n/n, C_D(\mu d)^{2}\log n/n\},
	\end{equation*}
	 for some large enough constant $C_T$, $C_D$, $\beta>1$, and inside the region $\mathcal{B}$. This sample complexity is sub-optimal when compared with the convex approach as in\cite{TasissaEDMCProof} unless $\mu$, $d$ are of order $\mathcal{O}(1)$. And we conclude the proof. \qed
	 
	 Replacing Lemma \ref{eq_NromalSpace_tight_bound_randomGraph} by a simpler approach as in\cite[Lemma 43]{MaImplicitRegularNonCVX_GD} gives $\Vert\frac{1}{p}\mathcal{R}_{\Omega}(\Delta\Delta^T)\Vert_F\leq 4\Vert\Delta\Vert_{2,\infty}^2\sqrt{\frac{2}{p}}n$ holds with probability at least $1-2e^{-2cn^2p}$ and cause the final radius of the region to be $\mathcal{O}(\frac{1}{\sqrt{n}})$. 
	 One can also use Lemma \ref{eq_NromalSpace_tight_bound} to bound $\zeta_2$, it will simplify analysis but lose the Frobenius norm bound on $\Delta$. It is also interesting to investigate whether the stronger version of restricted strong convexity as that in\cite[Lemma 7]{MaImplicitRegularNonCVX_GD} still holds on the s-stress function, and whether the Leave-One-Out analysis can be paralleled to uncover the ``implicit regularization" under this set of basis. We leave these to future work. 
	\subsection{Restricted Smoothness}
	\label{Appdi_B2_Bound}
	As in\cite[Sec. VII-G]{PhaseretrievalWirtinger}\cite[App. C.2]{ZhengLaffertyNonCVXFR}, we have $\Vert\nabla f(\mathbf{Y})\Vert_F^2=|\sup_{\Vert\mathbf{W}\Vert_F^2=1}\langle\nabla f(\mathbf{Y}),\mathbf{W}\rangle|^2$. Then the term $\mathrm{LHS}:=|\langle\nabla f(\mathbf{Y}),\mathbf{W}\rangle|^2$ can be decomposed via a similar method as in \eqref{eq_expanding_TN_Copments}
	\begin{align*}
		\setlength\belowdisplayskip{3pt}
		\setlength\abovedisplayskip{3pt}
		&\mathrm{LHS}=\left|\frac{1}{p}\langle\mathcal{R}_{\Omega}(\mathbf{Z}_{\mathbb{T}}+\mathbf{Z}_{\mathbb{T}}^{\perp}),\mathcal{R}_{\Omega}(\mathbf{W}_{\Delta}+\mathbf{W}_Y)\rangle\right|^2\\
		&\overset{(i)}{\leq}\frac{4}{p^2}\langle\mathcal{R}_{\Omega}\mathbf{Z}_{\mathbb{T}},\mathcal{R}_{\Omega}\mathbf{W}_{\Delta}\rangle^2+
		\frac{4}{p^2}\langle\mathcal{R}_{\Omega}\mathbf{Z}_{\mathbb{T}},\mathcal{R}_{\Omega}\mathbf{W}_Y\rangle^2\\
		&+\frac{4}{p^2}\langle\mathcal{R}_{\Omega}\mathbf{Z}_{\mathbb{T}}^{\perp},\mathcal{R}_{\Omega}\mathbf{W}_{\Delta}\rangle^2+
		\frac{4}{p^2}\langle\mathcal{R}_{\Omega}\mathbf{Z}_{\mathbb{T}}^{\perp},\mathcal{R}_{\Omega}\mathbf{W}_{Y}\rangle^2\\
		&\overset{(ii)}{\leq}\frac{4}{p^2}(\Vert\mathcal{R}_{\Omega}\mathbf{Z}_{\mathbb{T}}\Vert_F^2+\Vert\mathcal{R}_{\Omega}\mathbf{Z}_{\mathbb{T}}^{\perp}\Vert_F^2)(\Vert\mathcal{R}_{\Omega}\mathbf{W}_{\Delta}\Vert_F^2+\Vert\mathcal{R}_{\Omega}\mathbf{W}_{Y}\Vert_F^2).
	\end{align*}
	Here $\mathbf{Z}_{\mathbb{T}}$, $\mathbf{Z}_{\mathbb{T}}^{\perp}$ are defined as in \eqref{eq_EDMC_sstress_regular_cvx2}. Let $\mathbf{W}_{\Delta}=\mathbf{W}\Delta^T+\Delta\mathbf{W}^T$, $\mathbf{W}_Y=\mathbf{W}\bar{\mathbf{Y}}^{{\star}T}+\bar{\mathbf{Y}}^{\star}\mathbf{W}^T$. $(i)$ by using elementary inequality $(a+b+c+d)^2\leq 4(a^2+b^2+c^2+d^2)$, and $(ii)$ follows from Cauchy–Schwarz. Using parallelogram inequality again, we have
	\begin{align}
		\label{eq_Gradient_upper_bound_final1}
		\setlength\belowdisplayskip{3pt}
		\setlength\abovedisplayskip{3pt}
		\mathrm{LHS}&\leq4(\frac{4}{p}\Vert\mathcal{R}_{\Omega}\Delta\bar{\mathbf{Y}}^{{\star}T}\Vert_F^2+\frac{1}{p}\Vert\mathcal{R}_{\Omega}\Delta\Delta\Vert_F^2)\cdot(\frac{4}{p}\Vert\mathcal{R}_{\Omega}\mathbf{W}\Delta^T\Vert_F^2\nonumber\\
		&+\frac{4}{p}\Vert\mathcal{R}_{\Omega}\mathbf{W}\bar{\mathbf{Y}}^{{\star}T}\Vert_F^2)=4(4\mathrm{I}_1+\mathrm{I}_2)(4\mathrm{I}_3+4\mathrm{I}_4).
	\end{align}
	As in Appendix \ref{Appdi_B1_Bound}, $\mathrm{I}_2$ can be bounded by Lemma \ref{eq_NromalSpace_tight_bound_randomGraph}. Inspired by\cite[Lemma 9]{ZhengLaffertyNonCVXFR}, we develop the following estimate for $\mathrm{I}_1$, $\mathrm{I}_3$ and $\mathrm{I}_4$.
	\begin{lemma}
		\label{eq_NromalSpace_tight_bound}
		If the sample complexity $p>\frac{c\log n}{n}$ for some constant $c>3$, then with probability at least $1-4n^{-8}$, uniformly for all matrices $\mathbf{A}$, $\mathbf{B}\in\mathbb{R}^{n\times d}$, we have
		\begin{equation*}
			\setlength\belowdisplayskip{3pt}
			\setlength\abovedisplayskip{3pt}
			\frac{1}{p}\Vert\mathcal{R}_{\Omega}(\mathbf{AB}^T)\Vert_F^2\leq 32n\min\left\{\Vert \mathbf{A}\Vert_F^2\Vert\mathbf{B}\Vert_{2,\infty}^2,\Vert\mathbf{B}\Vert_F^2\Vert\mathbf{A}\Vert_{2,\infty}^2\right\}
		\end{equation*}
		holds for some small constant $c$ independent of $n,\,d$.
	\end{lemma}
	\begin{proof}
		This is established using similar argument as in \eqref{eq_normal_space_comp_eta_1}, \eqref{eq_gamma1_simpley_bounded}. Setting $\Omega=[\Omega_1^T,\dots,\Omega_n^T]^T$, we have
		\begin{align*}
			\setlength\belowdisplayskip{3pt}
			\setlength\abovedisplayskip{3pt}
			&\frac{1}{p}\Vert\mathcal{R}_{\Omega}(\mathbf{AB}^T)\Vert_F^2\leq\frac{1}{p}\sum_{\boldsymbol{\alpha}\in\Omega}\Vert\mathbf{z}_{\boldsymbol{\alpha}}^T\mathbf{A}\Vert_2^2\Vert\mathbf{z}_{\boldsymbol{\alpha}}^T\mathbf{B}\Vert_2^2\\
			&\overset{(i)}{\leq}\frac{8}{p}\sum_{i=1}^n\sum_{j\in\Omega_i}(\Vert\mathbf{e}_i^T\mathbf{A}\Vert_2^2 + \Vert\mathbf{e}_j^T\mathbf{A}\Vert_2^2) \Vert\mathbf{B}\Vert_{2,\infty}^2\\
			&\overset{(ii)}{\leq}\frac{16}{p}\Vert\mathbf{B}\Vert_{2,\infty}^2\sum_{i=1}^n\sum_{j\in\Omega_i}\Vert\mathbf{e}_i^T\mathbf{A}\Vert_2\overset{(iii)}{\leq}32n\Vert\mathbf{A}\Vert_F^2\Vert\mathbf{B}\Vert_{2,\infty}^2,
		\end{align*}
		where $(i),\,(ii)$ comes from Lemma \ref{lemma_incoherence} and \eqref{eq_gamma1_simpley_bounded}. $(iii)$ follows directly from \eqref{eq_normal_space_comp_eta_1} and holds with probability at least $1-2n^{-8}$. It is straightforward to check $\frac{1}{p}\Vert\mathcal{R}_{\Omega}(\mathbf{AB}^T)\Vert_F^2\leq32n\Vert\mathbf{B}\Vert_F^2\Vert\mathbf{A}\Vert_{2,\infty}^2$ holds with the same probability. And we conclude the proof.
	\end{proof}
	We note that the upper bound of $\frac{1}{p}\Vert\mathcal{R}_{\Omega}\mathbf{Z}_{\mathbb{T}}\Vert_F^2$ obtained by Lemma \ref{lemma_restricted_region_inco_T_RIP} is on the same order of $n$ as in Lemma \ref{eq_NromalSpace_tight_bound}\footnote{This means that the distortion brought by basis $\boldsymbol{\omega}_{\boldsymbol{\alpha}}$, whose correlation matrix has spectrum norm $2n$\cite[Lemma 18]{TasissaEDMCProof}, can be tightly bounded by utilizing the $l_2/l_{\infty}$ norm of the ground truth.}. When conditioned on $\mathrm{E}_1$ and some $\delta_r,\,\delta_t\leq1/2$ independent of $\mu$, $d$, Lemma \ref{eq_NromalSpace_tight_bound_randomGraph} and \ref{eq_NromalSpace_tight_bound} give
	\begin{subequations}
		\label{eq_upperbound_NormTan_I1234}
		\begin{gather}
			\setlength\belowdisplayskip{3pt}
			\setlength\abovedisplayskip{3pt}
			\mathrm{I}_1\leq32n\Vert\Delta\Vert_F^2\Vert\bar{\mathbf{Y}}^{\star}\Vert_{2,\infty}^2,\,\mathrm{I}_3\leq32n\Vert\Delta\Vert_{2,\infty}^2,\\
			\mathrm{I}_2\leq(\mu d\sigma_d^{\star}+12\Vert\Delta\Vert_F^2)\Vert\Delta\Vert_F^2,\,\mathrm{I}_4\leq 32n\Vert\bar{\mathbf{Y}}^{\star}\Vert_{2,\infty}^2,
		\end{gather}
	\end{subequations}
	where we use the fact that $\Vert\mathbf{W}\Vert_F^2=1$. Since $\kappa\geq1$, it holds that $\Vert\Delta\Vert_{2,\infty}^2\leq\frac{c\mu d\sigma_1^{\star}}{n}$ inside $\mathcal{B}$. By Substituting this and \eqref{eq_upperbound_NormTan_I1234}, $\Vert\bar{\mathbf{Y}}^{\star}\Vert_{2,\infty}^2\leq\frac{\mu d\sigma_1^{\star}}{n}$, $\Vert\Delta\Vert_F^2\leq\frac{\sigma_d^{\star}}{120}$ into \eqref{eq_Gradient_upper_bound_final1}, we have
	\begin{equation*}
		|\langle\nabla f(\mathbf{Y}),\mathbf{W}\rangle|^2\leq C_P \mu^2d^2\sigma_1^{{\star}2}\Vert\Delta\Vert_F^2
	\end{equation*}
	holds for some sufficient large constant $C_P$ independent of $\mu$, $d$, $n$ with probability at least $1-4n^{-8}-\frac{1}{2}n^{-4}$ as soon as $p\geq C_D(\mu d)^2\log n/n$. And we conclude the proof. \qed
\subsection{Proof of Lemma \ref{lemma_restricted_region_inco_T_RIP}}
\label{subsec_ProofofLemmaA1}
	This is established using matrix Bernstein inequality.  Let $\mathbf{S}_{\boldsymbol{\alpha}}:=(\frac{1}{p}\delta_{\boldsymbol{\alpha}}-1)\langle\cdot,\mathcal{P}_{\mathbb{T}}\boldsymbol{\omega}_{\boldsymbol{\alpha}}\rangle\mathcal{P}_{\mathbb{T}}\boldsymbol{\omega}_{\boldsymbol{\alpha}}$ be a zero mean, self-adjoint operator. One finds
\begin{equation*}
	\setlength\belowdisplayskip{3pt}
	\setlength\abovedisplayskip{3pt}
	\frac{1}{p}\mathcal{P}_{\mathbb{T}}\mathcal{R}_{\Omega}^*\mathcal{R}_{\Omega}\mathcal{P}_{\mathbb{T}}-\mathcal{P}_{\mathbb{T}}g^*g\mathcal{P}_{\mathbb{T}}=\sum_{\boldsymbol{\alpha}\in\mathbb{I}}\mathbf{S}_{\boldsymbol{\alpha}}.
\end{equation*}
To use Lemma \ref{eq_matrix_Bernstein}, we need bounds on $B=\Vert\mathbf{S}_{\boldsymbol{\alpha}}\Vert$, $\sigma^2=\Vert\mathbb{E}\sum_{\boldsymbol{\alpha}\in\mathbb{I}}\mathbf{S}_{\boldsymbol{\alpha}}^2\Vert$. We start form bounding $B$. For any $\mathbf{H}\in\mathbb{T}$,
\begin{equation*}
	\setlength\belowdisplayskip{3pt}
	\setlength\abovedisplayskip{3pt}
	\Vert\mathbf{S}_{\boldsymbol{\alpha}}\mathbf{H}\Vert_F\leq\frac{1}{p}\Vert\mathbf{H}\Vert_F\Vert\mathcal{P}_{\mathbb{T}}\boldsymbol{\omega}_{\boldsymbol{\alpha}}\Vert_F^2\overset{(i)}{\leq}\frac{8\mu d}{np}\Vert\mathbf{H}\Vert_F,
\end{equation*}
where $(i)$ from incoherence assumption and Lemma \ref{lemma_incoherence}. For $\sigma^2$, the major difficulty lies in bounding $\Vert\mathcal{P}_{\mathbb{T}}g^*g\mathcal{P}_{\mathbb{T}}\Vert$.
We achieve this by using the variational characterization of the spectral norm. For any $\mathbf{H}\in\mathbb{T}$, we have\vspace{-5pt}

\begin{small} 
	\begin{align}
		\setlength\belowdisplayskip{-5pt}
		\setlength\abovedisplayskip{-5pt}
		\label{eq_matrix_Bernstn_sigma2_upper}
		&\Vert\sum_{\boldsymbol{\alpha}\in\mathbb{I}}\mathbb{E}\mathbf{S}_{\boldsymbol{\alpha}}^2\Vert:=\sup_{\Vert\mathbf{H}\Vert_F=1}|\langle\mathbf{H},\sum_{\boldsymbol{\alpha}\in\mathbb{I}}\mathbb{E}\mathbf{S}_{\boldsymbol{\alpha}}^2\mathbf{H}\rangle|\nonumber\\
		&=\sup_{\Vert\mathbf{H}\Vert_F=1}|\langle\mathbf{H},\mathbb{E}\sum_{\boldsymbol{\alpha}\in\mathbb{I}}(\frac{1}{p}\delta_{\boldsymbol{\alpha}}-1)^2\Vert\mathcal{P}_{\mathbb{T}}\boldsymbol{\omega}_{\boldsymbol{\alpha}}\Vert_F^2\langle\mathbf{H},\mathcal{P}_{\mathbb{T}}\boldsymbol{\omega}_{\boldsymbol{\alpha}}\rangle\mathcal{P}_{\mathbb{T}}\boldsymbol{\omega}_{\boldsymbol{\alpha}}\rangle|\nonumber\\
		&\leq\sup_{\Vert\mathbf{H}\Vert_F=1}\frac{1-p}{p}\max_{\boldsymbol{\alpha}}\Vert\mathcal{P}_{\mathbb{T}}\boldsymbol{\omega}_{\boldsymbol{\alpha}}\Vert_F^2\underbrace{\sum_{\boldsymbol{\alpha}\in\mathbb{I}}\langle\mathbf{H},\mathcal{P}_{\mathbb{T}}\boldsymbol{\omega}_{\boldsymbol{\alpha}}\rangle^2}_{\gamma_1}.\vspace{-10pt}
	\end{align}
\end{small}
We give an elementary proof which can bound $\gamma_1$ to order $\mathrm{poly}(\mu d)\Vert\mathbf{H}\Vert_F^2$. This is possible since $\gamma_1$ is a restriction of basis $\boldsymbol{\omega}_{\boldsymbol{\alpha}}$ on low dimensional and incoherent space $\mathbb{T}$. Observe that $\gamma_1\leq\sum_{\boldsymbol{\alpha}\in\mathbb{I}}8a_{\boldsymbol{\alpha}}^2+2b_{\boldsymbol{\alpha}}^2$, where $a_{\boldsymbol{\alpha}}:=|\langle\mathbf{H},\mathcal{P}_{\mathbf{U}}\boldsymbol{\omega}_{\boldsymbol{\alpha}}\rangle|$, $b_{\boldsymbol{\alpha}}:=|\langle\mathbf{H},\mathcal{P}_{\mathbf{U}}\boldsymbol{\omega}_{\boldsymbol{\alpha}}\mathcal{P}_{\mathbf{U}}\rangle|$, since $\langle\mathbf{H},\mathcal{P}_{\mathbf{U}}\boldsymbol{\omega}_{\boldsymbol{\alpha}}\rangle=\langle\mathbf{H},\boldsymbol{\omega}_{\boldsymbol{\alpha}}\mathcal{P}_{\mathbf{U}}\rangle$. We start from $b_{\boldsymbol{\alpha}}$.
\begin{align}
	\setlength\belowdisplayskip{3pt}
	\setlength\abovedisplayskip{3pt}
	\label{eq_balpha_upperbound}
	b_{\boldsymbol{\alpha}}^2&\leq(\Vert\mathcal{P}_{\mathbf{U}}\mathbf{H}\mathcal{P}_{\mathbf{U}}\Vert_{\infty}\Vert\boldsymbol{\omega}_{\boldsymbol{\alpha}}\Vert_1)^2\leq(4\max_{i,j}\mathbf{e}_i^T\mathcal{P}_{\mathbf{U}}\mathbf{H}\mathcal{P}_{\mathbf{U}}\mathbf{e}_j)^2\nonumber\\
	&\leq(4\max_i\Vert\mathbf{e}_i^T\mathbf{U}^{\star}\mathbf{U}^{{\star}T}\Vert_2\Vert\mathbf{H}\Vert_F\max_j\Vert\mathbf{e}_j^T\mathbf{U}^{\star}\mathbf{U}^{{\star}T}\Vert_2)^2\nonumber\\
	&\overset{(i)}{\leq}16\Vert\mathbf{U}^{\star}\mathbf{U}^{{\star}T}\Vert_{\infty}^2\Vert\mathbf{H}\Vert_F^2\leq16(\mu d)^2\Vert\mathbf{H}\Vert_F^2/n^2,
\end{align}
where $(i)$ by noticing that $\max_i\Vert\mathbf{e}_i^T\mathbf{U}^{\star}\mathbf{U}^{{\star}T}\Vert_2^2=\max_i\mathbf{e}_i^T\mathbf{U}^{\star}\mathbf{U}^{{\star}T}\mathbf{e}_i\leq\Vert\mathbf{U}^{\star}\mathbf{U}^{{\star}T}\Vert_{\infty}$. As for $a_{\boldsymbol{\alpha}}$, we utilize the following splitting trick.
\begin{align}
	\setlength\belowdisplayskip{3pt}
	\setlength\abovedisplayskip{3pt}
	\label{eq_gamma1_simpley_bounded}
	\sum_{\boldsymbol{\alpha}\in\mathbb{I}}a_{\boldsymbol{\alpha}}^2&=\sum_{\boldsymbol{\alpha}\in\mathbb{I}}\langle\mathbf{z}_{\boldsymbol{\alpha}}^T\mathbf{H},\mathbf{z}_{\boldsymbol{\alpha}}^T\mathcal{P}_{\mathbf{U}}\rangle^2\nonumber\\
			&\leq\sum_{\boldsymbol{\alpha}\in\mathbb{I}}\Vert(\mathbf{e}_i-\mathbf{e}_j)^T\mathbf{H}\Vert_2^2\Vert(\mathbf{e}_i-\mathbf{e}_j)^T\mathbf{U}^{\star}\mathbf{U}^{{\star}T}\Vert_2^2\nonumber\\
			&\overset{(i)}{\leq}8\sum_{\boldsymbol{\alpha}\in\mathbb{I}}(\Vert\mathbf{e}_i^T\mathbf{H}\Vert_2^2+\Vert\mathbf{e}_j^T\mathbf{H}\Vert_2^2)\max_i\Vert\mathbf{e}_i^T\mathbf{U}^{\star}\mathbf{U}^{{\star}T}\Vert_2^2\nonumber\\
			&\leq 16\Vert\mathbf{U}^{\star}\mathbf{U}^{{\star}T}\Vert_{\infty}\sum_{i=1}^n\Vert\mathbf{H}\Vert_F^2\overset{(ii)}{\leq}16\mu d\Vert\mathbf{H}\Vert_F^2,
\end{align}
where $(i),\,(ii)$ by Lemma \ref{lemma_incoherence}. By Substituting \eqref{eq_balpha_upperbound} and \eqref{eq_gamma1_simpley_bounded} into \eqref{eq_matrix_Bernstn_sigma2_upper}, we have $\gamma_1\leq 160(\mu d)^2\Vert\mathbf{H}\Vert_F^2$, and 
\begin{equation*}
	\setlength\belowdisplayskip{3pt}
	\setlength\abovedisplayskip{3pt}
	\label{eq_gamma1_upperbound}
	\Vert\sum_{\boldsymbol{\alpha}\in\mathbb{I}}\mathbb{E}\mathbf{S}_{\boldsymbol{\alpha}}^2\Vert\leq \frac{1}{p}\max_{\boldsymbol{\alpha}}\Vert\mathcal{P}_{\mathbb{T}}\boldsymbol{\omega}_{\boldsymbol{\alpha}}\Vert_F^2\cdot 160(\mu d)^2 \overset{(i)}{\leq} \frac{1280(\mu d)^3}{np},
\end{equation*}
where $(i)$ from Lemma \ref{lemma_incoherence}. Using Lemma \ref{eq_matrix_Bernstein}, we find $\Vert\sum_{\alpha\in\mathbb{I}}\mathbf{S}_{\boldsymbol{\alpha}}\Vert\leq t$ holds with probability at least $1-n^{1-\beta}$, $\beta>1$ for
\begin{equation*}
\setlength\belowdisplayskip{3pt}
\setlength\abovedisplayskip{3pt}
t>\max\{\sqrt{\frac{8}{3}\beta\log n\sigma^2}, \frac{8}{3}\beta\log n B\}.
\end{equation*}
Thus, letting $p>C_T\beta(\mu d)^3\log n/n$ for some sufficient large numerical constant $C_T>10240/(3t^2)$ independent of $n,\,d$, the claim holds with probability at least $1-n^{1-\beta}$ for $\epsilon=t<1$. And we conclude the proof.  \qed

\subsection{Discussions and Remarks}
\label{subsec_discuss_remark_appB}
\emph{Remarks for Lemma \ref{lemma_restricted_region_inco_T_RIP}:} 
Compared with \cite[Lemma 10]{TasissaEDMCProof}, Lemma \ref{lemma_restricted_region_inco_T_RIP} is stronger but sub-optimal unless both $\mu$, $d$ are of order $\mathcal{O}(1)$\footnote{\cite[Lemma 10]{TasissaEDMCProof} has developed a similar estimate for $\zeta_1$ with sample complexity $p=\mathcal{O}(\mu d \log n/n)$, i.e., lower bounding $\lambda_{\min}(\frac{1}{p}\mathcal{P}_{\mathbb{T}}\mathcal{R}_{\Omega}^*\mathcal{R}_{\Omega}\mathcal{P}_{\mathbb{T}})$ by matrix Chernoff bound\cite[Thm. 1.1]{TroppUserFriendMCI}. This one-side bound has sample complexity in the same order as the Rudelson selection estimate in LRMC\cite[Thm. 4.1]{LRMC_Can1}\cite[Lemma 1]{ChenIncoOptimalMC}.}. And\cite[Lemma 10]{TasissaEDMCProof} can be directly used to replace Lemma \ref{lemma_restricted_region_inco_T_RIP}, thus improving the sample complexity required by Theorem \ref{them_distance_tirival_bounded}-(1). Using Lemma \ref{lemma_restricted_region_inco_T_RIP}, it is straightforward to check that for any $\mathbf{H}\in\mathbb{T}$
\begin{align}
	\setlength\belowdisplayskip{3pt}
	\setlength\abovedisplayskip{3pt}
	\label{eq_lemmaB1_direct_coro}
	\Vert &g(\mathbf{H})\Vert_F^2-\Vert\mathbf{H}\Vert_F^2\Vert\frac{1}{p}\mathcal{P}_{\mathbb{T}}\mathcal{R}_{\Omega}^*\mathcal{R}_{\Omega}\mathcal{P}_{\mathbb{T}}-\mathcal{P}_{\mathbb{T}}g^*g\mathcal{P}_{\mathbb{T}}\Vert\leq\nonumber\\
	&\langle\mathbf{H},(\frac{1}{p}\mathcal{P}_{\mathbb{T}}\mathcal{R}_{\Omega}^*\mathcal{R}_{\Omega}\mathcal{P}_{\mathbb{T}}\nonumber-\mathcal{P}_{\mathbb{T}}g^*g\mathcal{P}_{\mathbb{T}})\mathbf{H}\rangle+\Vert g(\mathbf{H})\Vert_F^2\\
	&\leq \Vert g(\mathbf{H})\Vert_F^2+\Vert\mathbf{H}\Vert_F^2\Vert\frac{1}{p}\mathcal{P}_{\mathbb{T}}\mathcal{R}_{\Omega}^*\mathcal{R}_{\Omega}\mathcal{P}_{\mathbb{T}}-\mathcal{P}_{\mathbb{T}}g^*g\mathcal{P}_{\mathbb{T}}\Vert
\end{align}
holds. Since $\langle\mathbf{H},(\frac{1}{p}\mathcal{P}_{\mathbb{T}}\mathcal{R}_{\Omega}^*\mathcal{R}_{\Omega}\mathcal{P}_{\mathbb{T}}\nonumber-\mathcal{P}_{\mathbb{T}}g^*g\mathcal{P}_{\mathbb{T}})\mathbf{H}\rangle+\Vert g(\mathbf{H})\Vert_F^2=\frac{1}{p}\Vert\mathcal{R}_{\Omega}\mathbf{H}\Vert_F^2$, we have
\begin{align*}
	\setlength\belowdisplayskip{3pt}
	\setlength\abovedisplayskip{3pt}
	\lambda_{\min}(g^*g)&\Vert\mathbf{H}\Vert_F^2-\epsilon\Vert\mathbf{H}\Vert_F^2\leq \Vert g(\mathbf{H})\Vert_F^2-\epsilon\Vert\mathbf{H}\Vert_F^2\\
	&\leq\frac{1}{p}\Vert\mathcal{R}_{\Omega}\mathbf{H}\Vert_F^2\leq\Vert g(\mathbf{H})\Vert_F^2+\epsilon\Vert \mathbf{H}\Vert_F^2\\
	&\leq \lambda_{\max}(g^*g)\Vert\mathbf{H}\Vert_F^2+\epsilon\Vert\mathbf{H}\Vert_F^2,
\end{align*} 
holds with high probability uniformly over all $\mathbf{H}\in\mathbb{T}$. According to\cite[Cor. 2.2]{LichtenbergTasissaEDMC}, we have $\lambda_{\min}(g^*g)=4$ and $\lambda_{\max}(g^*g)=4n$, i.e., the lower bound of $\frac{1}{p}\Vert\mathcal{R}_{\Omega}\mathbf{H}\Vert_F^2$ is good enough while the upper bound is too loose. From \eqref{eq_balpha_upperbound} and \eqref{eq_gamma1_simpley_bounded}, it holds that $\Vert g(\mathbf{H})\Vert_F^2\leq160(\mu d)^2\Vert\mathbf{H}\Vert_F^2$, which forces the upper bound of $\frac{1}{p}\Vert\mathcal{R}_{\Omega}\mathbf{H}\Vert_F^2$ to be independent of $n$. And the l.h.s. of \eqref{eq_lemmaB1_direct_coro} gives slightly better lower bound when compared with \cite[Lemma 10]{TasissaEDMCProof}. While the large constant $C_T$ can be improved by normalizing $\boldsymbol{\omega}_{\boldsymbol{\alpha}}$, it is interesting to investigate whether Lemma \ref{lemma_restricted_region_inco_T_RIP} (or its variants as in\cite[Lemma 23]{LeeUnify_CS_GenTransDom}) still holds in the region $p\geq O(\mu d \log n)/n$.

\emph{Remarks for Lemma \ref{eq_NromalSpace_tight_bound_randomGraph}:} 
Insightful readers may have already pointed out that Lemma \ref{eq_NromalSpace_tight_bound_randomGraph} is sub-optimal in the order of $n$. Since the basic idea of controlling the norm space component of the residual ($\Delta\Delta^T$) in non-convex LRMC\cite[Sec. IV]{ChenJNonCVX_rec_IOOcm} is to bound
\begin{equation*}
	\setlength\belowdisplayskip{3pt}
	\setlength\abovedisplayskip{3pt}
	\Vert\frac{1}{p}\mathcal{P}_{\Omega}\mathbf{11}^T-\mathbf{11}^T\Vert\leq C\sqrt{\frac{n}{p}},
\end{equation*}
which holds with probability at least $1-n^{-c}$ if $p\geq\frac{C\log n}{n}$ (known as the random graph lemma). When compared with a direct use of Lemma \ref{eq_matrix_Bernstein}, its r.h.s. improves by a $\log n$ factor\cite[Lemma 2.2]{VuRandomGraph}\footnote{Please see\cite[Lemma 12]{ChenLRMC_ErrorsErasures}\cite[Lemma 2]{ChenIncoOptimalMC} for the bound on $\Vert\frac{1}{p}\mathcal{P}_{\Omega}\mathbf{M}-\mathbf{M}\Vert$, where $\mathbf{M}\in\mathbb{R}^{n\times n}$ is an arbitrary but fixed matrix. The random graph lemma can be established by a strengthened Bernstein-type concentration result for matrices with independent random entries\cite[Thm. 3.4]{ChenSpectralMethods}.}. When the sample basis has internal structure, e.g., the set of Hankel basis $\mathbf{H}_{\alpha}$ (the $\alpha$-th skew-diagonal of $\mathbf{11}^T$), Cai et. al.\cite[Lemma 5]{CaiSPHankelMC_PGD} prove the following bound
\begin{equation*}
	\setlength\belowdisplayskip{3pt}
	\setlength\abovedisplayskip{3pt}
	\Vert\sum_{\alpha=1}^{2n-1}\frac{1}{p}\delta_{\alpha}\mathbf{H}_{\alpha}-\mathbf{11}^T\Vert\leq\sqrt{\frac{cn\log n}{p}}
\end{equation*}
holds with high probability provided $p\geq \frac{c\log n}{n}$. However, currently we cannot find such a nice structure when handling the $\boldsymbol{\omega}_{\boldsymbol{\alpha}}$ used in EDMC problem. By noticing $\sum_{\boldsymbol{\alpha}\in\mathbb{I}}\boldsymbol{\omega}_{\boldsymbol{\alpha}}=2(n\mathbf{I}-\mathbf{11}^T)$ and $\sum_{\boldsymbol{\alpha}\in\mathbb{I}}\boldsymbol{\omega}_{\boldsymbol{\alpha}}^{\prime}:=\sum_{\boldsymbol{\alpha}\in\mathbb{I}}(\mathbf{e}_i+\mathbf{e}_j)(\mathbf{e}_i+\mathbf{e}_j)^T=2(n\mathbf{I}+\mathbf{11}^T)$, we suggest that Lemma \ref{eq_NromalSpace_tight_bound_randomGraph} plays an analogous role as bounding $\Vert\sum_{\boldsymbol{\alpha}\in\mathbb{I}}\frac{1}{p}\delta_{\boldsymbol{\alpha}}\boldsymbol{\omega}_{\boldsymbol{\alpha}}^{\prime}-2(\mathbf{11}^T+n\mathbf{I})\Vert$ by observing
\begin{align}
	\setlength\belowdisplayskip{3pt}
	\setlength\abovedisplayskip{3pt}
	\label{eq_upperbound_trans_basis_prime}
	&\sum_{\boldsymbol{\alpha}\in\mathbb{I}}\langle\Delta\Delta^T,\boldsymbol{\omega}_{\boldsymbol{\alpha}}\rangle^2\nonumber\\
	&\leq4\sum_{\boldsymbol{\alpha}\in\mathbb{I}}(\Vert\mathbf{e}_i^T\Delta\Vert_2^4+\Vert\mathbf{e}_j^T\Delta\Vert_2^4+2\Vert\mathbf{e}_i^T\Delta\Vert_2^2\Vert\mathbf{e}_j^T\Delta\Vert_2^2)\nonumber\\
	&=4\sum_{\boldsymbol{\alpha}\in\mathbb{I}}
	\begin{bmatrix}
		\Vert\mathbf{e}_1^T\Delta\Vert_2^2\\
		\vdots\\
		\Vert\mathbf{e}_n^T\Delta\Vert_2^2
	\end{bmatrix}^T
	\boldsymbol{\omega}_{\boldsymbol{\alpha}}^{\prime}
	\begin{bmatrix}
		\Vert\mathbf{e}_1^T\Delta\Vert_2^2\\
		\vdots\\
		\Vert\mathbf{e}_n^T\Delta\Vert_2^2
	\end{bmatrix}=4\mathbf{x}^T(\sum_{\boldsymbol{\alpha}\in\mathbb{I}}\boldsymbol{\omega}_{\boldsymbol{\alpha}}^{\prime})\mathbf{y},
\end{align} 
where $\mathbf{x}=\mathbf{y}=[\Vert\mathbf{e}_1^T\Delta\Vert_2^2,\dots,\Vert\mathbf{e}_n^T\Delta\Vert_2^2]^T$.
For $\mathbf{D}_{\boldsymbol{\alpha}}:=(\frac{1}{p}\delta_{\boldsymbol{\alpha}}-1)\boldsymbol{\omega}_{\boldsymbol{\alpha}}^{\prime}$, $\mathbb{E}\mathbf{D}_{\boldsymbol{\alpha}}=\mathbf{0}$, it is not hard to show 
\begin{equation*}
	\setlength\belowdisplayskip{3pt}
	\setlength\abovedisplayskip{3pt}
	\Vert\mathbf{D}_{\boldsymbol{\alpha}}\Vert\leq \frac{1}{p}\Vert\boldsymbol{\omega}_{\boldsymbol{\alpha}}^{\prime}\Vert\leq \frac{1}{p}\Vert\mathbf{e}_i+\mathbf{e}_j\Vert_2^2\leq\frac{4}{p}.
\end{equation*}
Also
\begin{align*}
	\setlength\belowdisplayskip{3pt}
	\setlength\abovedisplayskip{3pt}
	\Vert\mathbb{E}\sum_{\boldsymbol{\alpha}\in\mathbb{I}}\mathbf{D}_{\boldsymbol{\alpha}}^2\Vert&\leq\Vert\sum_{\boldsymbol{\alpha}\in\mathbb{I}} \mathbb{E}(\frac{1}{p}\delta_{\boldsymbol{\alpha}}-1)^2\boldsymbol{\omega}_{\boldsymbol{\alpha}}^{\prime 2}\Vert\nonumber\\
	&\overset{(i)}{\leq}\frac{4}{p}\Vert\sum_{\boldsymbol{\alpha}\in\mathbb{I}}\boldsymbol{\omega}_{\boldsymbol{\alpha}}^{\prime}\Vert=\frac{4}{p}\Vert2n\mathbf{I}+2\mathbf{11}^T\Vert\leq\frac{16n}{p},
\end{align*}
where $(i)$ by noticing that $\boldsymbol{\omega}_{\boldsymbol{\alpha}}^{\prime 2}\preccurlyeq4\boldsymbol{\omega}_{\boldsymbol{\alpha}}^{\prime},\,\forall\,\boldsymbol{\alpha}\in[n]^2$. By applying Lemma \ref{eq_matrix_Bernstein}, we get
\begin{equation}
		\setlength\belowdisplayskip{3pt}
	\setlength\abovedisplayskip{3pt}
	\label{eq_MBI_upperbound_on_norm_spectral}
	\Vert\sum_{\boldsymbol{\alpha}\in\mathbb{I}}\frac{1}{p}\delta_{\boldsymbol{\alpha}}\boldsymbol{\omega}_{\boldsymbol{\alpha}}^{\prime}-2(\mathbf{11}^T+n\mathbf{I})\Vert\leq\sqrt{\frac{128\beta n\log n}{3p}},
\end{equation}
holds with probability as least $1-n^{1-\beta}$, $\beta >1$, provided with $p\geq c\beta\log n/n$. By \eqref{eq_upperbound_trans_basis_prime} and \eqref{eq_MBI_upperbound_on_norm_spectral}, we have
\begin{align}
		\setlength\belowdisplayskip{3pt}
	\setlength\abovedisplayskip{3pt}
	\label{eq_MBI_upperbound_on_norm_thebound}
	&\frac{1}{p}\Vert\mathcal{R}_{\Omega}\Delta\Delta^T\Vert_F^2=\frac{1}{p}\sum_{\boldsymbol{\alpha}\in\Omega}\langle\Delta\Delta^T,\boldsymbol{\omega}_{\boldsymbol{\alpha}}\rangle^2\nonumber\\
	&\leq 4\left(2\mathbf{x}^T(\mathbf{11}^T+n\mathbf{I})\mathbf{y}+\mathbf{x}^T(\sum_{\boldsymbol{\alpha}\in\mathbb{I}}\mathbf{D}_{\boldsymbol{\alpha}})\mathbf{y}\right)\nonumber\\
	&\leq 8\Vert\mathbf{x}\Vert_1\Vert\mathbf{y}\Vert_1+\left(8n+\sqrt{\frac{c\beta n\log n}{p}}\right)\Vert\mathbf{x}\Vert_2\Vert\mathbf{y}\Vert_2,
\end{align}
which gives almost the same bound (up to $\log n$ factor) as Lemma \ref{eq_NromalSpace_tight_bound_randomGraph}, i.e., the result is $n$-times heavier on the diagonal. This sub-optimality of Lemma \ref{eq_NromalSpace_tight_bound_randomGraph} (also \eqref{eq_MBI_upperbound_on_norm_thebound} above) makes Theorem \ref{them_distance_tirival_bounded} weaker than the standard version of the regularity condition. Upon closer examination of the proof of Theorem \ref{them_distance_tirival_bounded}-(1), one realizes that the current version does not lead to a traceable algorithm (like projected/regularized GD in\cite{SL15NonCVXFR}\cite{ZhengLaffertyNonCVXFR}) since the attractive region is too small by means of $l_{2,\infty}$ norm. This is caused by the $16n$ or $4n$ term in Lemma \ref{eq_NromalSpace_tight_bound_randomGraph} or \eqref{eq_MBI_upperbound_on_norm_thebound}. Nevertheless, it is natural to conjecture that SVD-MDS initialized VGD could enjoy good convergence when solving s-stress, given Theorem \ref{them_distance_tirival_bounded} at hand, as suggested in Section \ref{subsecIII-C_Basin}. However, it remains a problem whether the $cn$ term can be improved to $\sqrt{cn}$ or not, even though numerical experiments reveal that the attractive region of s-stress is quite large. 
	\section{Auxiliary Lemmas}
	\label{Appdi_C_AuxiliaryLemma}
	\begin{lemma}[\cite{RongGeSpuriousLocalMinima}, Lemma 6]
		\label{eq_geodesic_distance}
		The well-known orthogonal Procrustes problem actually defines a geodesic distance on quotient manifold $\mathbb{R}^{n\times d}_*/O(d)$ if both two point sets are non-singular\textnormal{\cite[Prop. 4.4]{MassartQuotientgeometry}}. Define
		\begin{equation}
			\label{eq_orthogonal_Procrustes_problem}
			\boldsymbol{\psi}^{\star}=\arg\min_{\boldsymbol{\psi}\in O(d)}\Vert\mathbf{Y}-\mathbf{Y}^{\star}\boldsymbol{\psi}\Vert_F^2,
		\end{equation}
		for some $\mathbf{Y}$ and $\mathbf{Y}^{\star} \in \mathbb{R}^{n\times d}_*$. 
		\eqref{eq_orthogonal_Procrustes_problem} has a simple solution $\boldsymbol{\psi}^{\star}=\mathbf{AB}^T$, $\mathbf{Y}^{{\star}T}\mathbf{Y}=\mathbf{A}\mathbf{D}\mathbf{B}^T$. 
		Let $\Delta=\mathbf{Y}-\mathbf{Y}^{\star}\boldsymbol{\psi}^{\star}$, one finds that $\mathbf{Y}^T\mathbf{Y}^{\star}\boldsymbol{\psi}^{\star}\succcurlyeq \boldsymbol{0}$ and $\Delta^T\mathbf{Y}^{\star}\boldsymbol{\psi}^{\star}\in \mathcal{S}(d)$. Also, the following equation holds
		\begin{equation}
			\label{eq_distance_rela_4}
			\mathbf{Y}\Delta^T+\Delta\mathbf{Y}^{T}-\Delta\Delta^T=\mathbf{Y}\mathbf{Y}^T-\mathbf{Y}^{\star}\mathbf{Y}^{{\star}T}.
		\end{equation}
	\end{lemma}
	\begin{lemma}[\cite{TasissaEDMCProof}, Lemma 21]
		\label{lemma_incoherence}
		If the matrix $\mathbf{Y}^{\star}=\mathbf{U}^{\star}\boldsymbol{\Sigma}^{{\star}1/2}$ satisfies $\Vert\mathbf{U}^{\star}\Vert_{2,\infty}^2\leq \frac{\mu d}{n}$, $\Vert\mathbf{Y}^{\star}\Vert_{2,\infty}^2\leq \frac{\mu d\sigma_1^{\star}}{n}$, $\sigma_1^{\star}=\sigma_1(\mathbf{Y}^{\star}\mathbf{Y}^{{\star}T})$, then followings hold simultaneously
		\begin{gather*}
			\setlength\belowdisplayskip{4pt}
			\setlength\abovedisplayskip{4pt}
			\Vert\mathcal{P}_{\mathbb{T}}\boldsymbol{\omega}_{\boldsymbol{\alpha}}\Vert_F^2\leq\frac{8\mu d}{n},\,
			\max_{\boldsymbol{\alpha}\in\mathbb{I}}\Vert\mathbf{z}_{\boldsymbol{\alpha}}^T\mathbf{X}\Vert_2^2\leq4\Vert\mathbf{X}\Vert_{2,\infty}^2,\\
			\Vert\mathbf{U}^{\star}\mathbf{U}^{{\star}T}\Vert_{\infty}\leq\Vert\mathbf{U}^{\star}\Vert_{2,\infty}^2\leq\frac{\mu d}{n}.
		\end{gather*}
		The second inequality holds for any $\mathbf{X}\in\mathbb{R}^{n\times d}$ but for brevity we state it here.
	\end{lemma}
	\begin{lemma}[\cite{TroppUserFriendMCI}, Theorem 6.1]
		\label{eq_matrix_Bernstein}
		The matrix Bernstein inequality. If a finite sequence $\{\mathbf{S}_k\}$ of independent, random, self-adjoint matrices of dimension $\mathbb{R}^{n\times n}$ that satisfy
		\begin{equation*}
			\setlength\belowdisplayskip{3pt}
			\setlength\abovedisplayskip{3pt}
			\mathbb{E}\mathbf{S}_k=\mathbf{0},\,\Vert\mathbf{S}_k\Vert\leq B \,\,\text{almost surely.}
		\end{equation*}
		Set the norm of the total variance being $\sigma^2:=\Vert\sum_k\mathbb{E}\mathbf{S}_k^2\Vert$. Then for all all $t\geq0$ the following holds
		\begin{align*}
			\setlength\belowdisplayskip{3pt}
			\setlength\abovedisplayskip{3pt}
			\mathbb{P}&\left\{\lambda_{\max}(\sum_k \mathbf{S}_k)\geq t\right\}\leq n \exp(\frac{-t^2}{2\sigma^2+2Bt/3})\\
			&\leq
			\begin{cases}
				n\exp(-3t^2/(8\sigma^2)),&\,\mathrm{for}\,t\leq\sigma^2/B\\
				n\exp(-3t/(8B)),&\,\mathrm{for}\,t\geq\sigma^2/B
			\end{cases}.
		\end{align*}
	\end{lemma}
	\begin{lemma}[\cite{ZhengLaffertyNonCVXFR}, Lemma 7]
		\label{random_graph_lemma}
		The random graph lemma\textnormal{\cite[Lemma 7.1]{KeshavanOptSpace}}, see also \textnormal{\cite[Lemma 36]{RongGeSpuriousLocalMinima}}. Suppose $\Omega$ is the set of edges of a random bipartite graph with $(n,n)$ nodes, where any pair of nodes on different side is connected with probability $p$. If $p\geq \frac{c_r \log n}{\delta_r^2 n}$, then with probability at least $1-\frac{1}{2}n^{-4}$, uniformly for all $\mathbf{x},\mathbf{y}\in \mathbb{R}^n$ and some $\delta_r\in(0,1]$, it holds that
		\begin{equation*}
			\setlength\belowdisplayskip{3pt}
			\setlength\abovedisplayskip{3pt}
			\label{eq_random_graph_lemma}
			p^{-1}\sum_{(i,j)\in\Omega} x_i y_j\leq (1+\delta_r)\Vert\mathbf{x}\Vert_1\Vert\mathbf{y}\Vert_1+c_g \sqrt{\frac{n}{p}}\Vert\mathbf{x}\Vert_2\Vert\mathbf{y}\Vert_2.
		\end{equation*}
	\end{lemma}

\section{Extended Appendix: Numerical Test}
\label{sec_extended_appendix}
\subsection{Simple Numerical Test: EDMC and LRMC}
\label{subsec_EDMC_LRMC_CVXmethod}
Under Bernoulli rule, or the ``uniformly sampling rule", Tasissa and Lai\cite{TasissaEDMCProof} proved that the following algorithm will recover the ground truth EDM in the noiseless case as soon as sample complexity reaches above $p\geq\mathcal{O}(v d\log^2 n/n)$
\begin{align}
	\label{eq_EDM_tr_min_alg}
	\setlength\belowdisplayskip{4pt}
	\setlength\abovedisplayskip{4pt}
	\min_{\mathbf{G}}&\,\mathrm{tr}(\mathbf{G})\nonumber\\
	\mathrm{s.t.}&\,\mathrm{tr}(\mathbf{G}\boldsymbol{\omega}_{\boldsymbol{\alpha}})=\mathrm{tr}(\mathbf{G}^{\star}\boldsymbol{\omega}_{\boldsymbol{\alpha}}),\,\boldsymbol{\alpha}\in\Omega,\,\mathbf{G}\succcurlyeq\mathbf{0}
\end{align}
where $v$ is the incoherence parameter defined w.r.t. $\boldsymbol{\omega}_{\boldsymbol{\alpha}}$ which is different from $\mu$ used in this work up to constant factors\cite[Sec. II-B]{TasissaEDMCProof}\footnote{Again, we note that\cite{TasissaEDMCProof} needs the joint incoherence assumption. Since $d=\mathcal{O}(1)$ in practical EDMC problems and the fact that Gram matrix is PSD, we have $v\leq c\mu^2 d=\mathcal{O}(\mu^2)$.}. They also suggested that since $\boldsymbol{\omega}_{\boldsymbol{\alpha}}$ satisfies $\boldsymbol{\omega}_{\boldsymbol{\alpha}}\mathbf{1}=\mathbf{0}$, the sample constraint naturally enforces $\mathbf{G1}=\mathbf{0}$. \eqref{eq_EDM_tr_min_alg} can be regarded as the so-called trace minimization. We note that their sample complexity is at the same order as\cite{GrossLRMCProof}. In other words, the performance difference between \eqref{eq_EDM_tr_min_alg} and the NNM \eqref{eq_NNM} is caused by the rank, incoherence parameter and other constants. Practically, the phase transition of \eqref{eq_EDM_tr_min_alg} occurs earlier than \eqref{eq_NNM} as one increases the sample rate $p$ while fixing $n$ and $d$.
\begin{equation}
	\label{eq_NNM}
	\setlength\belowdisplayskip{4pt}
	\setlength\abovedisplayskip{4pt}
	\min_{\mathbf{D}},\,\Vert\mathbf{D}\Vert_*,\,\mathrm{s.t.}\,\mathcal{P}_{\Omega}(\mathbf{D})=\mathcal{P}_{\Omega}(\mathbf{D}^{\star}).
\end{equation} 
Take $n=100$ and $d=2$ as an example, we plot their phase transition obtained by 100 independent trails in Fig. \ref{Fig_PhaseTransition_NNM_vs_Trmin}. The constraint $\mathrm{tr}(\mathbf{G}\boldsymbol{\omega}_{\boldsymbol{\alpha}})=\mathrm{tr}(\mathbf{G}^{\star}\boldsymbol{\omega}_{\boldsymbol{\alpha}})$ implicitly imposes a restriction on $\mathbf{D}\in \mathbb{EDM}^n$. Consequently, \eqref{eq_EDM_tr_min_alg} performs better than the original NNM, which does not conflict with the LRMC theory. This phenomenon indicates that it is important to focus on the properties of EDM rather than directly applying methods developed in the LRMC context.

We need to point out that the sample complexity/performance limit of convex algorithms should not be regarded as a ``barrier" for general matrix recovery problems. We use \eqref{eq_EDM_tr_min_alg}, \eqref{eq_NNM} here for illustration purposes only. If one incorporates more prior information about the underlying signal/matrix, then better performance assurance is possible, please see, e.g., \cite{SoltanolBreakSamBarrier}, for the case of quadratic measurement model.
\begin{figure}[!t]
	\centering
	\setlength{\abovecaptionskip}{0cm}
	\setlength{\belowcaptionskip}{0cm}
	\includegraphics[width=4.7cm]{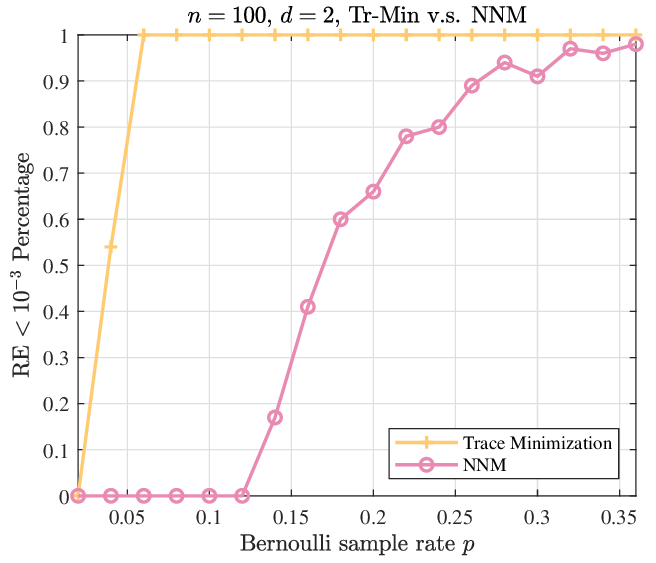}
	\caption{Phase transition of trace minimization and NNM. For each $p$, the result is obtained by 100 independent trails. The point set is generated by standard Gaussian distribution. We claim a success if the EDM recover rate falls below $10^{-3}$.}
	\label{Fig_PhaseTransition_NNM_vs_Trmin}
\end{figure}
\subsection{Phase Transition of SVD-MDS-GD-HZLS}
\label{App_D_Phase_Transition_SVD_MDS_GD}
\begin{figure}[!t]
	\centering
	\setlength{\abovecaptionskip}{0cm}
	\setlength{\belowcaptionskip}{0cm}
	\includegraphics[width=9.2cm]{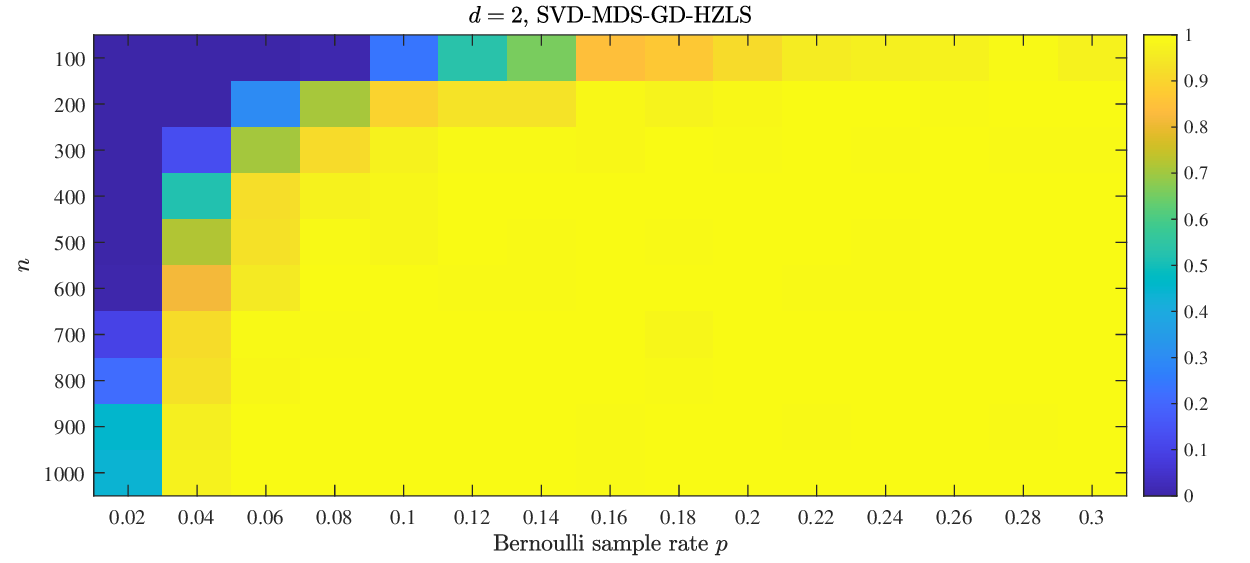}
	\caption{Phase transition of SVD-MDS-GD-HZLS. For each pair $(n,p)$, the result is obtained by 200 independent trails. The point set is generated by standard Gaussian distribution. We claim a success if the EDM recover rate falls below $10^{-3}$.}
	\label{Fig_PhaseTransition_SVD_MDS_HZLS_GD}
\end{figure}
\begin{figure}[!t]
	\centering
	\setlength{\abovecaptionskip}{0cm}
	\setlength{\belowcaptionskip}{0cm}
	\includegraphics[width=9.2cm]{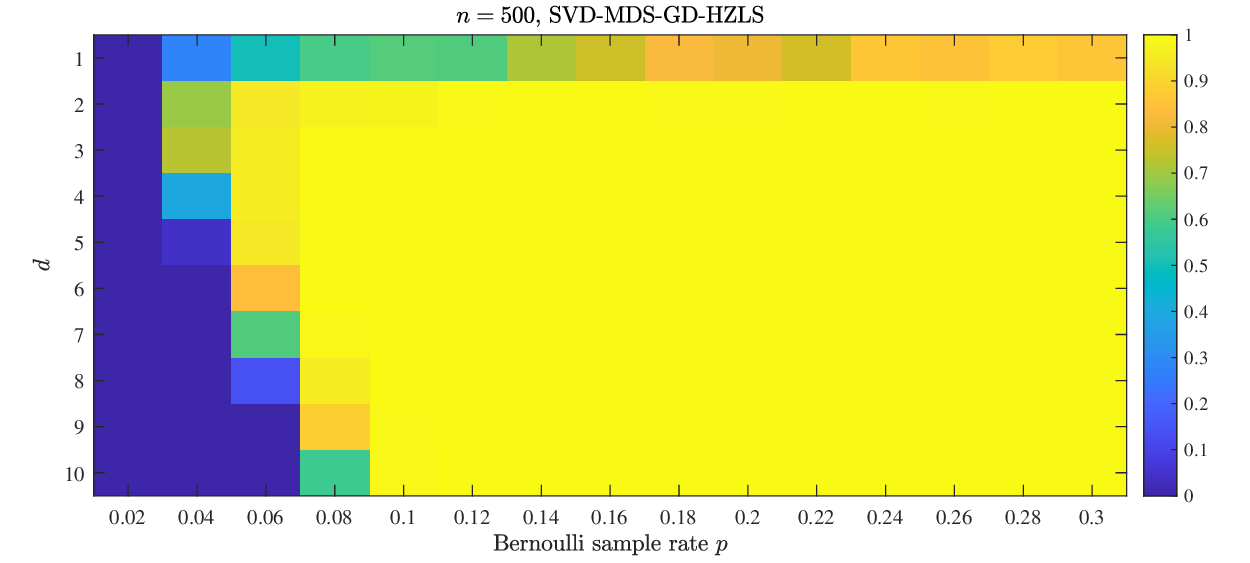}
	\caption{Phase transition of SVD-MDS-GD-HZLS. For each pair $(d,p)$, the result is obtained by 200 independent trails. The point set is generated by standard Gaussian distribution. We claim a success if the EDM recover rate falls below $10^{-3}$.}
	\label{Fig_PhaseTransition_SVD_MDS_HZLS_GD_vary_d}
\end{figure}
We visualize the empirical phase transition of the SVD-MDS initialized GD with HZLS step size when solving the s-stress in two cases: (i) varying $n$ and the Bernoulli sample complexity $p$ (with $d=2$ fixed) in Fig. \ref{Fig_PhaseTransition_SVD_MDS_HZLS_GD}; (ii) varying $d$ and the Bernoulli sample complexity $p$ (with $n=500$ fixed) in Fig. \ref{Fig_PhaseTransition_SVD_MDS_HZLS_GD_vary_d}. The SVD-MDS-GD-HZLS is implemented by replacing the CG step in Algorithm \ref{alg_RCG_on_QuoM} (Line 14) with GD. The algorithm is stopped when either $\Vert \boldsymbol{p}_k \Vert_F \leq 10^{-6}$ or $\alpha_k\Vert \boldsymbol{\xi}_k \Vert_F \leq 10^{-10}$ or $|\bar{f}(\mathbf{Y}_k)-\bar{f}(\mathbf{Y}_{k+1})|\leq 10^{-10}$, and $\mathrm{IMAX}$ is set to $600$. The performance of SVD-MDS-GD-HZLS appears to lie between that of NNM and trace minimization, as indicated by Fig. \ref{Fig_PhaseTransition_SVD_MDS_HZLS_GD}, using $n=100$ as an example. It's worth noting that HZLS is not a monotonic descent line search \cite[Sec. 5.5]{SuttiRMGLSPhDthis}, when the sample rate is low, the SVD-MDS-GD-HZLS allows for a small increase (around $10^{-10}$) in the cost function. Consequently, we opt to terminate SVD-MDS-GD-HZLS if a bulk of ascent events ($10$ times) occur\footnote{The phase transition may be further improved by using a more sophisticated strategy, such as incorporating the root-finding algorithm\cite[Sec. 6.3]{ZhengRQCGtheory} to calculate the initial step size for restarting HZLS.}. The phase transition edge shows linear dependence w.r.t. $d$, as suggested by Fig. \ref{Fig_PhaseTransition_SVD_MDS_HZLS_GD_vary_d}. While the non-convex matrix completion results are acknowledged to be sub-optimal in terms of $\mu, d$\cite{ChenJNonCVX_rec_IOOcm}, it would be intriguing to investigate whether the $l_{2,\infty}$ analysis developed in\cite{ChenIncoOptimalMC}\cite{ChenRobust_CS_EMC} can be paralleled to refine Tasissa's bound\cite{TasissaEDMCProof}.

\subsection{Attractive Basin Under Unit Ball Rule}
\label{App_D_attractiveRegion_untiball}
\begin{figure}[!t]
	\centering
	\setlength{\abovecaptionskip}{0cm}
	\setlength{\belowcaptionskip}{0cm}
	\includegraphics[width=9cm]{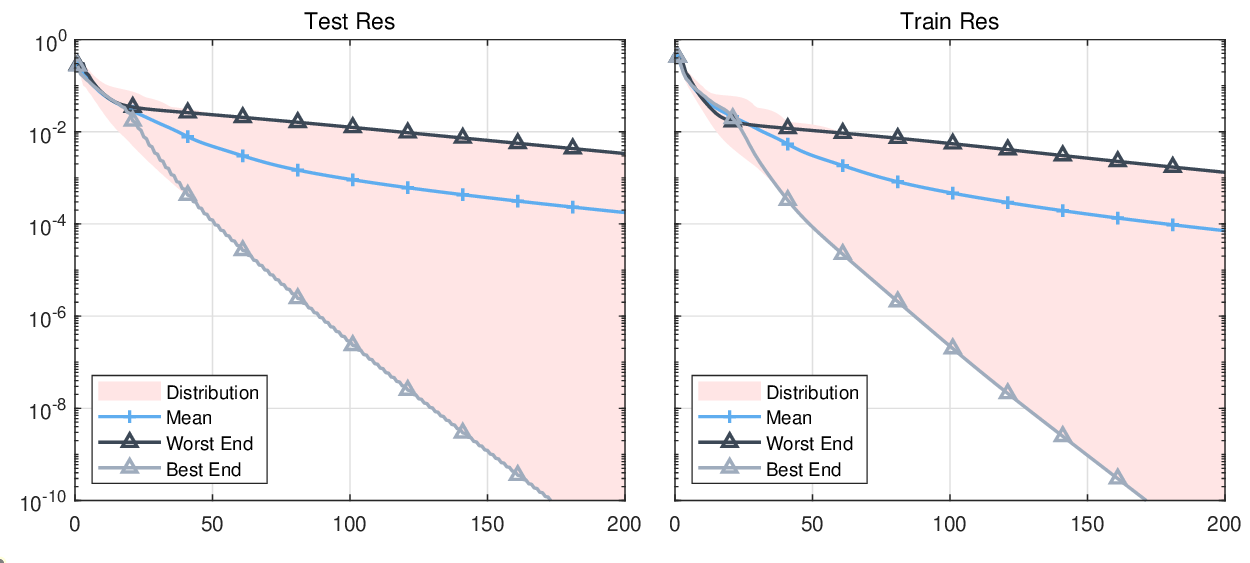}
	\caption{Example of linear convergence of GD when applied to SNL under near-optimal sample rate, but with Hager-Zhang step size. Starting randomly around the ground truth.}
	\label{LinearConvageTestandTrainRes}
\end{figure}
\begin{figure}[!t]
	\centering
	\setlength{\abovecaptionskip}{0cm}
	\setlength{\belowcaptionskip}{0cm}
	\includegraphics[width=9cm]{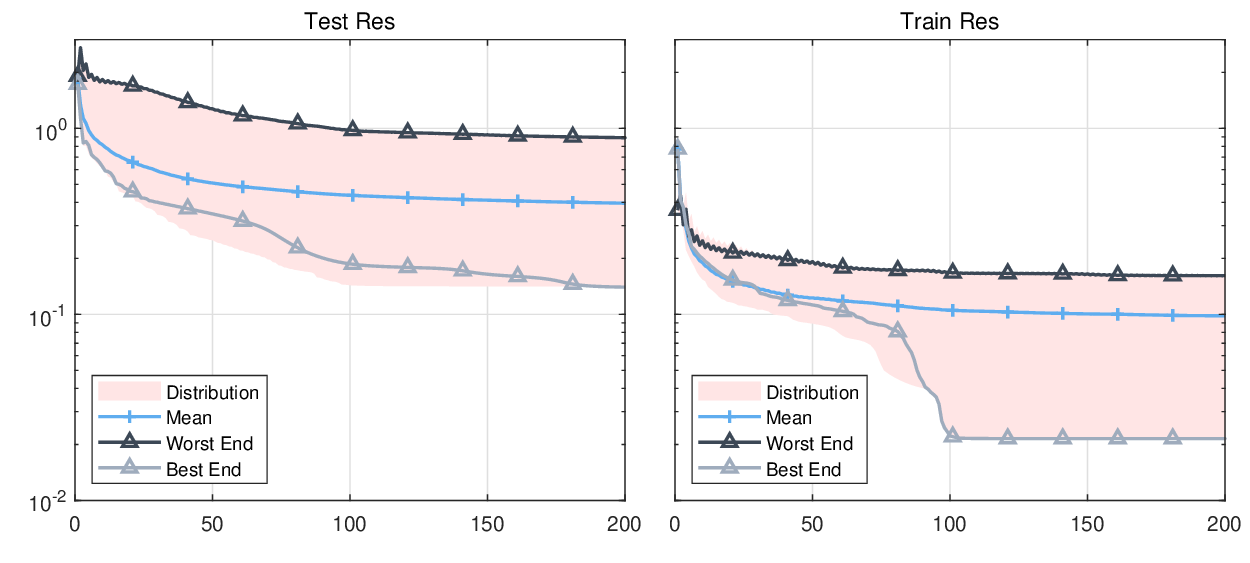}
	\caption{Example of convergence trajectory of GD when applied to SNL under near-optimal sample rate, but with Hager-Zhang step size. Initialized by SVD-MDS.}
	\label{LinearConvageTestandTrainRes_spectral}
\end{figure}
Preliminary numerical experiments suggest that the attractive basin of the s-stress function still appears to exist, even under the unit ball rule. We consider an SNL scenario where 100 nodes are uniformly distributed in $[-0.5,0.5]^2,\,d=2$, and there are no anchor nodes. Using the random graph generated from the unit ball rule and testing the dimension of null space of the stress matrix 1000 times, as described in\cite{SingerTheBound}, we find that when $r>0.35$, the network is globally rigid with high probability. Therefore, we fix $r=0.35$ and generate the sample set. We set the initialization as $\mathbf{Y}_0=\mathbf{J}*(\mathbf{Y}^{\star}+\mathrm{randn}(100,2)*0.6)\in\mathbb{R}^{100\times 2}$, where $\mathbf{J}$ is the geometric centering matrix. The noise level of $0.6$ causes a normalized distance $\Vert\mathbf{Y}_0\mathbf{Y}_0^T-\mathbf{Y}^{\star}\mathbf{Y}^{{\star}T}\Vert_F/\Vert\mathbf{Y}^{\star}\mathbf{Y}^{{\star}T}\Vert_F$ of around $0.8$. Then \eqref{eq_s_stress_gmap_EDMC} is solved by GD with Hager-Zhang step size starting from $\mathbf{Y}_0$. Experiment results on 20 independent trails are shown in Fig. \ref{LinearConvageTestandTrainRes}, where test residual and train residual are defined as follows
\begin{subequations}
	\begin{gather}
		\mathrm{TrainRes}=\frac{0.5\Vert\mathcal{P}_\Omega(g(\hat{\mathbf{Y}}\hat{\mathbf{Y}}^T)-\mathbf{D}^{\star})\Vert_F}{\Vert\mathcal{P}_\Omega(g(\hat{\mathbf{Y}}\hat{\mathbf{Y}}^T))\Vert_F},\\
		\mathrm{TestRes}=\frac{0.5\Vert g(\hat{\mathbf{Y}}\hat{\mathbf{Y}}^T)-\mathbf{D}^{\star}\Vert_F}{\Vert g(\hat{\mathbf{Y}}\hat{\mathbf{Y}}^T)\Vert_F}.
	\end{gather}
\end{subequations}
Clearly, GD shows linear convergence behavior after a few iterations, indicating the existence of a region around the global optimum that exhibits ``restricted strong convexity" even under the unit ball sample model. This observation could be considered as a preliminary conjecture for why gradient refinement after MAP-MDS works. We infer that either the attractive region or the global landscape of s-stress exhibits degeneracy under the unit ball rule. Otherwise, the search for good initial points for s-stress would not have lasted this long\footnote{Compared with phase retrieval under Gaussian ensembles\cite{SunQuWrightPRLandScape}, or other models where the (2r,4r)-strong convexity and smoothness holds\cite{ZhuGlobalGeo}, the benign local/global landscape helps GD to succeed even when initialized randomly\cite{ChenYXPRRandominit}.}. It is also reasonable to observe that SVD-MDS does not certify its ability to enter this region in Fig. \ref{LinearConvageTestandTrainRes_spectral}, since $\Vert\frac{1}{p}\mathcal{P}_{\Omega}\mathbf{D}^{\star}-\mathbf{D}^{\star}\Vert$ is a biased estimator under the unit ball rule.
\subsection{The Shape of The Network}
\label{App_D_Shape_of_Network}
\begin{figure}[!t]
	\centering
	\setlength{\abovecaptionskip}{0cm}
	\setlength{\belowcaptionskip}{0cm}
	\includegraphics[width=9cm]{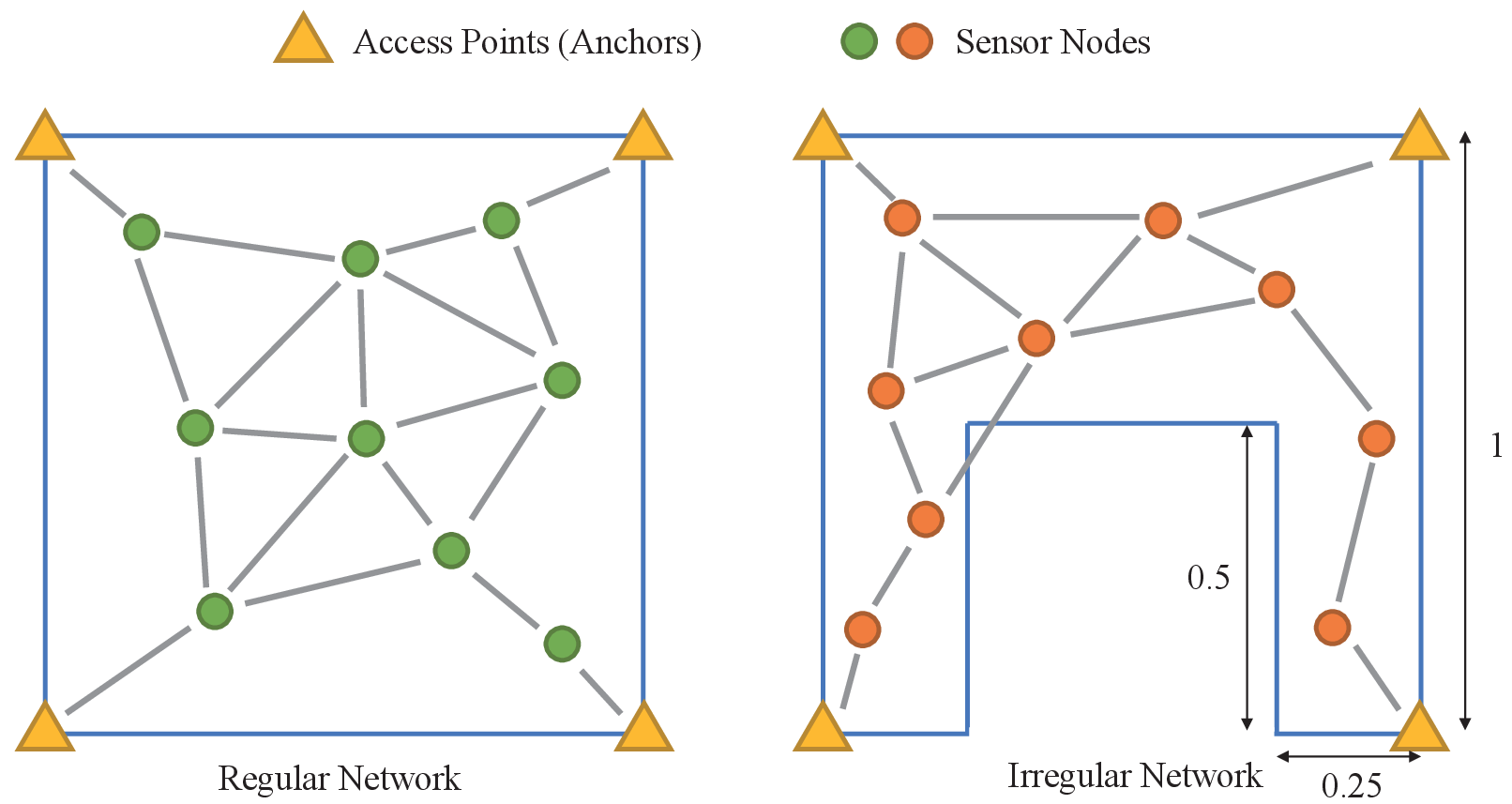}
	\caption{Illustration of regular (left) and irregular (right) network. For irregular network, we randomly drop 100 sensor nodes in the scene described in the right subfigure. We assume anchors always form a cliques structure. For simplicity of simulation, we assume that links that fall outside the scene are still available.}
	\label{RegularirregularNetwork}
\end{figure}
It's important to note that when the network is regularly shaped, for instance, with nodes uniformly distributed in a convex polyhedron, numerical tests demonstrate that the performance of MAP-MDS initialized RHZLS is nearly as effective as MVU-SDP. Numerical results of 1000 independent trails of SVD-MDS-RHZLS, MAP-MDS-RHZLS, r-RHZLS on the scenarios described in Fig. \ref{RegularirregularNetwork} are depicted in Fig. \ref{RegularirregularNetwork_performance}. Loosely speaking, r-RHZLS sacrifices some performance for generalizability, while MAP-MDS is seriously influenced by the network's topology. Information conveyed in Fig. \ref{RegularirregularNetwork_performance} suggests that the MAP-MDS initialization has not completely addressed the ill-posed nature of the SNL problem (the unit ball sample model). Instead of suggesting the s-stress is highly non-convex, it is conjectured that if a proper ``unbiased estimator" is available, then the EDMC problem can be solved effectively. Intuitively, if the EDM is sampled using Bernoulli model, then $\frac{1}{p}\mathcal{P}_{\Omega}$ is a good approximation to the identity operator $\mathcal{I}$\footnote{We note that the incoherence assumption w.r.t. the ground truth EDM is still needed to ensure this ``approximation".}. And this is how simple spectral initialization like SVD-MDS works\cite{ZhangSVD_MDS}. The unit ball sample model causes this simple estimator to be biased. Nonetheless, if nodes are uniformly distributed in a convex polyhedron, then MAP-MDS serves as a good estimator under the unit ball rule\cite{KarbasiOhMAPMDSTri}. Things get complicated when the shape of the point set is irregular, as MAP-MDS becomes biased. The unit ball rule makes the universal estimator $\frac{1}{p}\mathcal{P}_{\Omega}$ biased, and designing new unbiased estimator seems challenging without a prior knowledge of the network topology.
\begin{figure}[!t]
	\centering
	\setlength{\abovecaptionskip}{0cm}
	\setlength{\belowcaptionskip}{0cm}
	\includegraphics[width=9cm]{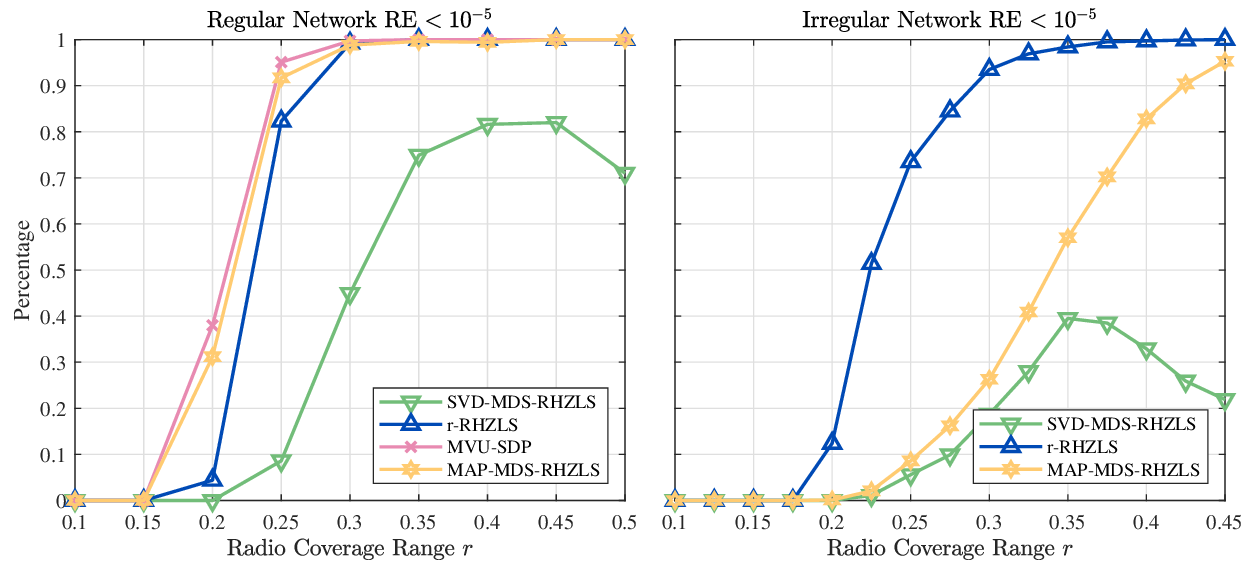}
	\caption{Performance of algorithms on regular and irregular networks. For regular network, the experiment setup is the same as Section \ref{secV_NumExps}.}
	\label{RegularirregularNetwork_performance}
\end{figure}

\bibliographystyle{IEEEtran}
\bibliography{IEEEabrv,reference}

\begin{thebibliography}{100}
\providecommand{\url}[1]{#1}
\csname url@samestyle\endcsname
\providecommand{\newblock}{\relax}
\providecommand{\bibinfo}[2]{#2}
\providecommand{\BIBentrySTDinterwordspacing}{\spaceskip=0pt\relax}
\providecommand{\BIBentryALTinterwordstretchfactor}{4}
\providecommand{\BIBentryALTinterwordspacing}{\spaceskip=\fontdimen2\font plus
\BIBentryALTinterwordstretchfactor\fontdimen3\font minus
  \fontdimen4\font\relax}
\providecommand{\BIBforeignlanguage}[2]{{%
\expandafter\ifx\csname l@#1\endcsname\relax
\typeout{** WARNING: IEEEtran.bst: No hyphenation pattern has been}%
\typeout{** loaded for the language `#1'. Using the pattern for}%
\typeout{** the default language instead.}%
\else
\language=\csname l@#1\endcsname
\fi
#2}}
\providecommand{\BIBdecl}{\relax}
\BIBdecl

\bibitem{MaoSNLSurvey}
G.~Mao, B.~Fidan, and B.~D.~O. Anderson, ``Wireless sensor network localization
  techniques,'' \emph{Comput. Netw.}, vol.~51, no.~10, pp. 2529--2553, 2007.

\bibitem{HanSNLSurvey}
G.~Han, J.~Jiang, C.~Zhang, T.~Q. Duong, M.~Guizani, and G.~K. Karagiannidis,
  ``A survey on mobile anchor node assisted localization in wireless sensor
  networks,'' \emph{IEEE Commun. Surv. Tutorials}, vol.~18, no.~3, pp.
  2220--2243, 2016.

\bibitem{TasissaEDMCProof}
A.~Tasissa and R.~Lai, ``Exact reconstruction of {E}uclidean distance geometry
  problem using low-rank matrix completion,'' \emph{IEEE Trans. Inf. Theory},
  vol.~65, no.~5, pp. 3124--3144, 2019.

\bibitem{ZhangSVD_MDS}
H.~Zhang, Y.~Liu, and H.~Lei, ``Localization from incomplete {E}uclidean
  distance matrix: {P}erformance analysis for the {SVD–MDS} approach,''
  \emph{IEEE Trans. Signal Process.}, vol.~67, no.~8, pp. 2196--2209, 2019.

\bibitem{DingEDMCrestrictCVX}
C.~Ding and H.-D. Qi, ``Convex optimization learning of faithful {E}uclidean
  distance representations in nonlinear dimensionality reduction,'' \emph{Math.
  Program.}, vol. 164, no.~1, pp. 341--381, 2017.

\bibitem{LRMC_Can1}
E.~J. Cand\`es and B.~Recht, ``Exact matrix completion via convex
  optimization,'' \emph{Found. Comput. Math.}, vol.~9, no.~6, pp. 717--772,
  2009.

\bibitem{DavenportLRMRoverview}
M.~A. Davenport and J.~Romberg, ``An overview of low-rank matrix recovery from
  incomplete observations,'' \emph{IEEE J. Sel. Top. Signal Process.}, vol.~10,
  no.~4, pp. 608--622, 2016.

\bibitem{LuoNonCVX_QCQP}
Z.~q. Luo, W.~k. Ma, A.~M.~c. So, Y.~Ye, and S.~Zhang, ``Semidefinite
  relaxation of quadratic optimization problems,'' \emph{IEEE Signal Process
  Mag.}, vol.~27, no.~3, pp. 20--34, 2010.

\bibitem{SoAnthonySNL2006}
A.~M.-C. So and Y.~Ye, ``Theory of semidefinite programming for sensor network
  localization,'' \emph{Math. Program.}, vol. 109, no. 2-3, pp. 367--384, 2006.

\bibitem{BiswasSDRSNLacm}
P.~Biswas, T.-C. Lian, T.-C. Wang, and Y.~Ye, ``Semidefinite programming based
  algorithms for sensor network localization,'' \emph{ACM Trans. Sen. Netw.},
  vol.~2, no.~2, p. 188–220, 2006.

\bibitem{CandTaoLRMC}
E.~J. Candes and T.~Tao, ``The power of convex relaxation: {N}ear-optimal
  matrix completion,'' \emph{IEEE Trans. Inf. Theory}, vol.~56, no.~5, pp.
  2053--2080, 2010.

\bibitem{GrossLRMCProof}
D.~Gross, ``Recovering low-rank matrices from few coefficients in any basis,''
  \emph{IEEE Trans. Inf. Theory}, vol.~57, no.~3, pp. 1548--1566, 2011.

\bibitem{YxChentitQuadraticSampling}
Y.~Chen, Y.~Chi, and A.~J. Goldsmith, ``Exact and stable covariance estimation
  from quadratic sampling via convex programming,'' \emph{IEEE Trans. Inf.
  Theory}, vol.~61, no.~7, pp. 4034--4059, 2015.

\bibitem{LiYMaNonCVXrankOne}
Y.~Li, C.~Ma, Y.~Chen, and Y.~Chi, ``Nonconvex matrix factorization from
  rank-one measurements,'' \emph{IEEE Trans. Inf. Theory}, vol.~67, no.~3, pp.
  1928--1950, 2021.

\bibitem{VershyninHigDimbook}
R.~Vershynin, \emph{High-Dimensional Probability: {A}n Introduction with
  Applications in Data Science}.\hskip 1em plus 0.5em minus 0.4em\relax
  Cambridge University Press, 2018.

\bibitem{BurerMonteiroFR}
S.~Burer and R.~D.~C. Monteiro, ``Local minima and convergence in low-rank
  semidefinite programming,'' \emph{Math. Program.}, vol. 103, no.~3, pp.
  427--444, 2005.

\bibitem{ZhengLaffertyNonCVXFR}
Q.~Zheng and J.~Lafferty, ``Convergence analysis for rectangular matrix
  completion using {B}urer-{M}onteiro factorization and gradient descent,''
  \emph{arXiv preprint arXiv:1605.07051}, 2016.

\bibitem{SL15NonCVXFR}
R.~Sun and Z.~Q. Luo, ``Guaranteed matrix completion via non-convex
  factorization,'' \emph{IEEE Trans. Inf. Theory}, vol.~62, no.~11, pp.
  6535--6579, 2016.

\bibitem{MaImplicitRegularNonCVX_GD}
C.~Ma, K.~Wang, Y.~Chi, and Y.~Chen, ``Implicit regularization in nonconvex
  statistical estimation: {G}radient descent converges linearly for phase
  retrieval, matrix completion, and blind deconvolution,'' \emph{Found. Comput.
  Math.}, vol.~20, no.~3, pp. 451--632, 2020.

\bibitem{RongGeSpuriousLocalMinima}
R.~Ge, C.~Jin, and Y.~Zheng, ``No spurious local minima in nonconvex low rank
  problems: {A} unified geometric analysis,'' in \emph{Proc. 34th Int. Conf.
  Mach. Learn.}, 2017, p. 1233–1242.

\bibitem{LibertiEDMSurvey}
L.~Liberti, C.~Lavor, N.~Maculan, and A.~Mucherino, ``Euclidean distance
  geometry and applications,'' \emph{SIAM Rev.}, vol.~56, no.~1, pp. 3--69,
  2014.

\bibitem{MDSSurvey}
N.~Saeed, H.~Nam, T.~Y. Al-Naffouri, and M.-S. Alouini, ``A state-of-the-art
  survey on multidimensional scaling-based localization techniques,''
  \emph{IEEE Commun. Surv. Tutorials}, vol.~21, no.~4, pp. 3565--3583, 2019.

\bibitem{NetworkLoc1}
J.~Aspnes, T.~Eren \emph{et~al.}, ``A theory of network localization,''
  \emph{IEEE Trans. Mob. Comput.}, vol.~5, no.~12, pp. 1663--1678, 2006.

\bibitem{GortlerGlobalRigidity}
S.~J. Gortler, A.~D. Healy, and D.~P. Thurston, ``Characterizing generic global
  rigidity,'' \emph{Am. J. Math.}, vol. 132, no.~4, pp. 897--939, 2010.

\bibitem{AlfakihuniquenessEDMC}
A.~Y. Alfakih, ``On the uniqueness of {E}uclidean distance matrix
  completions,'' \emph{Linear Algebra Appl.}, vol. 370, pp. 1--14, 2003.

\bibitem{SingerTheBound}
A.~Singer and M.~Cucuringu, ``Uniqueness of low-rank matrix completion by
  rigidity theory,'' \emph{SIAM J. Matrix Anal. Appl.}, vol.~31, no.~4, pp.
  1621--1641, 2010.

\bibitem{DrineasSVD_MDS}
P.~Drineas, A.~Javed, M.~Magdon-Ismail, G.~Pandurangan, R.~Virrankoski, and
  A.~Savvides, ``Distance matrix reconstruction from incomplete distance
  information for sensor network localization,'' in \emph{Ann. IEEE Commun.
  Soc. Ad hoc Commun. Netw. Secon}, vol.~2, 2006, pp. 536--544.

\bibitem{ShangYDistMDS}
Y.~Shang, W.~Ruml, Y.~Zhang, and M.~P.~J. Fromherz, ``Localization from mere
  connectivity,'' in \emph{Proc. ACM Int. Symp. Mobile Ad hoc Netw. Comput.},
  2003, p. 201–212.

\bibitem{BiswasTR_sdp}
P.~Biswas, T.~C. Liang, K.~C. Toh, Y.~Ye, and T.~C. Wang, ``Semidefinite
  programming approaches for sensor network localization with noisy distance
  measurements,'' \emph{IEEE Trans. Autom. Sci. Eng.}, vol.~3, no.~4, pp.
  360--371, 2006.

\bibitem{JavanmardMontanariTrEDMC}
A.~Javanmard and A.~Montanari, ``Localization from incomplete noisy distance
  measurements,'' \emph{Found. Comput. Math.}, vol.~13, no.~3, pp. 297--345,
  2013.

\bibitem{MishraRieEDMC}
B.~Mishra, G.~Meyer, and R.~Sepulchre, ``Low-rank optimization for distance
  matrix completion,'' in \emph{Proc. 50th IEEE Conf. Decision Control Eur.
  Control Conf.}, 2011, pp. 4455--4460.

\bibitem{ParhizkarOptSpace}
R.~Parhizkar, A.~Karbasi, S.~Oh, and M.~Vetterli, ``Calibration using matrix
  completion with application to ultrasound tomography,'' \emph{IEEE Trans.
  Signal Process.}, vol.~61, no.~20, pp. 4923--4933, 2013.

\bibitem{NguyenLRM-CG}
L.~T. Nguyen, J.~Kim, S.~Kim, and B.~Shim, ``Localization of {IoT} networks via
  low-rank matrix completion,'' \emph{IEEE Trans. Commun.}, vol.~67, no.~8, pp.
  5833--5847, 2019.

\bibitem{FangEDMCNewton}
H.-r. Fang and D.~P. O'Leary, ``Euclidean distance matrix completion
  problems,'' \emph{Optim. Methods Software}, vol.~27, no. 4-5, pp. 695--717,
  2012.

\bibitem{DokmanicEDMTheory}
I.~Dokmanic, R.~Parhizkar, J.~Ranieri, and M.~Vetterli, ``Euclidean distance
  matrices: {Essential} theory, algorithms, and applications,'' \emph{IEEE
  Signal Process Mag.}, vol.~32, no.~6, pp. 12--30, 2015.

\bibitem{AgostiniChannelCEDMC}
P.~Agostini, Z.~Utkovski, and S.~Stańczak, ``{Channel Charting}: an
  {Euclidean} distance matrix completion perspective,'' in \emph{Proc. Int.
  Conf. Acoust. Speech Signal Process.}, 2020, pp. 5010--5014.

\bibitem{ChanChartStuder}
C.~Studer, S.~Medjkouh, E.~Gonultaş, T.~Goldstein, and O.~Tirkkonen, ``Channel
  {C}harting: {L}ocating users within the radio environment using channel state
  information,'' \emph{IEEE Access}, vol.~6, pp. 47\,682--47\,698, 2018.

\bibitem{PaulsenCell3DGenomeReconstruction}
J.~Paulsen, O.~Gramstad, and P.~Collas, ``Manifold based optimization for
  single-cell {3D} genome reconstruction,'' \emph{PLoS Comput. Biol.}, vol.~11,
  no.~8, p. e1004396, 2015.

\bibitem{WangY3DGenomeFLAMINGO}
H.~Wang, J.~Yang, Y.~Zhang, J.~Qian, and J.~Wang, ``Reconstruct high-resolution
  {3D} genome structures for diverse cell-types using {FLAMINGO},'' \emph{Nat.
  Commun.}, vol.~13, no.~1, p. 2645, 2022.

\bibitem{MaricRieOptIKRobot}
F.~Mari\'{c}, M.~Giamou, A.~W. Hall, S.~Khoubyarian, I.~Petrovi\'{c}, and
  J.~Kelly, ``Riemannian optimization for distance-geometric inverse
  kinematics,'' \emph{IEEE Trans. Robot.}, vol.~38, no.~3, pp. 1703--1722,
  2022.

\bibitem{LaiPDEmanifolds}
R.~Lai and J.~Li, ``Solving partial differential equations on manifolds from
  incomplete interpoint distance,'' \emph{SIAM J. Sci. Comput.}, vol.~39,
  no.~5, pp. A2231--A2256, 2017.

\bibitem{TabaghiKineticEDM}
P.~Tabaghi, I.~Dokmanić, and M.~Vetterli, ``Kinetic {E}uclidean distance
  matrices,'' \emph{IEEE Trans. Signal Process.}, vol.~68, pp. 452--465, 2020.

\bibitem{KarbasiOhMAPMDSTri}
A.~Karbasi and S.~Oh, ``Robust localization from incomplete local
  information,'' \emph{IEEE/ACM Trans. Networking}, vol.~21, pp. 1131--1144,
  2011.

\bibitem{ChenSpectralMethods}
Y.~Chen, Y.~Chi, J.~Fan, and C.~Ma, ``Spectral methods for data science: {A}
  statistical perspective,'' \emph{Found. Trends Mach. Learn.}, vol.~14, no.~5,
  pp. 566--806, 2021.

\bibitem{AriasCastro3Pertubound}
E.~Arias-Castro, A.~Javanmard, and B.~Pelletier, ``Perturbation bounds for
  procrustes, classical scaling, and trilateration, with applications to
  manifold learning,'' \emph{J. Mach. Learn. Res.}, vol.~21, no.~1, p. Article
  15, 2020.

\bibitem{ImprovedMDSShang}
Y.~Shang and W.~Ruml, ``Improved {MDS}-based localization,'' in \emph{Proc IEEE
  INFOCOM}, vol.~4, 2004, pp. 2640--2651.

\bibitem{WeinbergerMVU}
K.~Q. Weinberger and L.~K. Saul, ``An introduction to nonlinear dimensionality
  reduction by maximum variance unfolding,'' in \emph{Proc. Natl. Conf. Artif.
  Intell.}, vol.~2, 2006, p. 1683–1686.

\bibitem{weightedMCNegahban}
S.~Negahban and M.~J. Wainwright, ``Restricted strong convexity and weighted
  matrix completion: {O}ptimal bounds with noise,'' \emph{J. Mach. Learn.
  Res.}, vol.~13, no.~1, pp. 1665--1697, 2012.

\bibitem{ChenIncoOptimalMC}
Y.~Chen, ``Incoherence-optimal matrix completion,'' \emph{IEEE Trans. Inf.
  Theory}, vol.~61, no.~5, pp. 2909--2923, 2015.

\bibitem{TasissaLRMC_genebasis}
A.~Tasissa and R.~Lai, ``Low-rank matrix completion in a general non-orthogonal
  basis,'' \emph{Linear Algebra Appl.}, vol. 625, pp. 81--112, 2021.

\bibitem{AlfakihSNLSDPSurvey}
A.~Alfakih, M.~Anjos, V.~Piccialli, and H.~Wolkowicz, ``Euclidean distance
  matrices, semidefinite programming and sensor network localization,''
  \emph{Port. Math.}, no.~1, pp. 53--102, 2011.

\bibitem{ChenYXChiLRMCSurvey}
Y.~Chen and Y.~Chi, ``Harnessing structures in big data via guaranteed low-rank
  matrix estimation: {R}ecent theory and fast algorithms via convex and
  nonconvex optimization,'' \emph{IEEE Signal Process Mag.}, vol.~35, no.~4,
  pp. 14--31, 2018.

\bibitem{KearsleySstress}
A.~J. Kearsley, R.~A. Tapia, and M.~W. Trosset, ``The solution of the metric
  {STRESS} and {SSTRESS} problems in multidimensional scaling using {N}ewton's
  method,'' Report, 1994.

\bibitem{CrippenEneemb}
G.~M. Crippen, ``Conformational analysis by energy embedding,'' \emph{J.
  Comput. Chem.}, vol.~3, no.~4, pp. 471--476, 1982.

\bibitem{HavelCompChemEDMC}
T.~F. Havel, ``Distance geometry: {T}heory, algorithms, and chemical
  applications,'' \emph{Encyclopedia of Computational Chemistry}, vol. 120, pp.
  723--742, 1998.

\bibitem{TakaneSStress}
Y.~Takane, F.~W. Young, and J.~de~Leeuw, ``Nonmetric individual differences
  multidimensional scaling: {An} alternating least squares method with optimal
  scaling features,'' \emph{Psychometrika.}, vol.~42, no.~1, pp. 7--67, 1977.

\bibitem{KeshavanOptSpace}
R.~H. Keshavan, A.~Montanari, and S.~Oh, ``Matrix completion from a few
  entries,'' \emph{IEEE Trans. Inf. Theory}, vol.~56, no.~6, pp. 2980--2998,
  2010.

\bibitem{ParhizkarPhDthsis}
R.~Parhizkar, ``Euclidean distance matrices: {P}roperties, algorithms and
  applications,'' Thesis, EPFL, 2013.

\bibitem{TangTohYinReiDimreduce}
T.~Tang, K.-C. Toh, N.~Xiao, and Y.~Ye, ``A {R}iemannian dimention-reduced
  second order method with application in sensor network localization,''
  \emph{arXiv preprint arXiv:2304.10092}, 2023.

\bibitem{LeiYinBlessHighOrderSNL}
M.~Lei, J.~Zhang, and Y.~Ye, ``Blessing of high-order dimensionality: from
  non-convex to convex optimization for sensor network localization,''
  \emph{arXiv preprint arXiv:2308.02278}, 2023.

\bibitem{ZhuGlobalGeo}
Z.~Zhu, Q.~Li, G.~Tang, and M.~B. Wakin, ``The global optimization geometry of
  low-rank matrix optimization,'' \emph{IEEE Trans. Inf. Theory}, vol.~67,
  no.~2, pp. 1308--1331, 2021.

\bibitem{LiSymmSaddleLandscape}
X.~Li, J.~Lu \emph{et~al.}, ``Symmetry, saddle points, and global optimization
  landscape of nonconvex matrix factorization,'' \emph{IEEE Trans. Inf.
  Theory}, vol.~65, no.~6, pp. 3489--3514, 2019.

\bibitem{ChenYXPRRandominit}
Y.~Chen, Y.~Chi, J.~Fan, and C.~Ma, ``Gradient descent with random
  initialization: fast global convergence for nonconvex phase retrieval,''
  \emph{Math. Program.}, vol. 176, no. 1–2, p. 5–37, 2019.

\bibitem{SunQuWrightPRLandScape}
J.~Sun, Q.~Qu, and J.~Wright, ``A geometric analysis of phase retrieval,''
  \emph{Found. Comput. Math.}, vol.~18, no.~5, pp. 1131--1198, 2018.

\bibitem{PhaseretrievalWirtinger}
E.~J. Cand\`es, X.~Li, and M.~Soltanolkotabi, ``Phase retrieval via {Wirtinger}
  flow: {Theory} and algorithms,'' \emph{IEEE Trans. Inf. Theory}, vol.~61,
  no.~4, pp. 1985--2007, 2015.

\bibitem{TuProcrustesFlow}
S.~Tu, R.~Boczar, M.~Simchowitz, M.~Soltanolkotabi, and B.~Recht, ``Low-rank
  solutions of linear matrix equations via procrustes flow,'' in \emph{Proc.
  Int. Conf. Mach. Learn. (ICML)}.\hskip 1em plus 0.5em minus 0.4em\relax
  JMLR.org, 2016, p. 964–973.

\bibitem{BoumalGlobalratesRTRv2}
N.~Boumal, P.-A. Absil, and C.~Cartis, ``Global rates of convergence for
  nonconvex optimization on manifolds,'' \emph{IMA J. Numer. Anal.}, vol.~39,
  no.~1, pp. 1--33, 2018.

\bibitem{ChiYNonCVX_Facor_overview}
Y.~Chi, Y.~M. Lu, and Y.~Chen, ``Nonconvex optimization meets low-rank matrix
  factorization: {A}n overview,'' \emph{IEEE Trans. Signal Process.}, vol.~67,
  pp. 5239--5269, 2018.

\bibitem{ZhangSGDEDMC}
J.~Zhang, H.-M. Chiu, and R.~Y. Zhang, ``Accelerating {SGD} for highly
  ill-conditioned huge-scale online matrix completion,'' \emph{Adv. neural inf.
  proces. syst.}, vol.~35, pp. 37\,549--37\,562, 2022.

\bibitem{MishraRiegeometryMetric}
B.~Mishra, K.~A. Apuroop, and R.~Sepulchre, ``A {R}iemannian geometry for
  low-rank matrix completion,'' \emph{ArXiv}, vol. abs/1211.1550, 2012.

\bibitem{HuangPhaseLift}
W.~Huang, K.~A. Gallivan, and X.~Zhang, ``Solving {PhaseLift} by low-rank
  {Riemannian} optimization methods for complex semidefinite constraints,''
  \emph{SIAM J. Sci. Comput.}, vol.~39, no.~5, pp. B840--B859, 2017.

\bibitem{GaoRRMC}
B.~Gao and P.~A. Absil, ``A {Riemannian} rank-adaptive method for low-rank
  matrix completion,'' \emph{Comput. Optim. Appl.}, vol.~81, no.~1, pp. 67--90,
  2022.

\bibitem{KovnatskyMADMM}
A.~Kovnatsky, K.~Glashoff, and M.~M. Bronstein, ``{MADMM}: {A} generic
  algorithm for non-smooth optimization on manifolds,'' in \emph{Proc. ECCV},
  2016, pp. 680--696.

\bibitem{KrislockEDMandApp}
N.~Krislock and H.~Wolkowicz, \emph{Euclidean Distance Matrices and
  Applications}.\hskip 1em plus 0.5em minus 0.4em\relax Boston, MA: Springer
  US, 2012, pp. 879--914.

\bibitem{BoumalIntrotoMani}
N.~Boumal, \emph{An introduction to optimization on smooth manifolds}.\hskip
  1em plus 0.5em minus 0.4em\relax Cambridge University Press, 2023.

\bibitem{AbsilOptonManifold}
P.-A. Absil, R.~Mahony, and R.~Sepulchre, \emph{Optimization algorithms on
  matrix manifolds}.\hskip 1em plus 0.5em minus 0.4em\relax Princeton
  University Press, 2009.

\bibitem{VandereyckenEmbgeoSnp}
B.~Vandereycken, P.~A. Absil, and S.~Vandewalle, ``Embedded geometry of the set
  of symmetric positive semidefinite matrices of fixed rank,'' in \emph{Proc.
  IEEE/SP Workshop Statist. Signal Process.}, 2009, pp. 389--392.

\bibitem{AbsilRTR}
P.~A. Absil, C.~G. Baker, and K.~A. Gallivan, ``Trust-region methods on
  {Riemannian} manifolds,'' \emph{Found. Comput. Math.}, vol.~7, no.~3, pp.
  303--330, 2007.

\bibitem{MassartQuotientgeometry}
E.~Massart and P.-A. Absil, ``Quotient geometry with simple geodesics for the
  manifold of fixed-rank positive-semidefinite matrices,'' \emph{SIAM J. Matrix
  Anal. Appl.}, vol.~41, no.~1, pp. 171--198, 2020.

\bibitem{ZhengRQCGtheory}
S.~Zheng, W.~Huang, B.~Vandereycken, and X.~Zhang, ``Riemannian optimization
  using three different metrics for {Hermitian} {PSD} fixed-rank constraints:
  {An} extended version,'' \emph{arXiv preprint arXiv:2204.07830}, 2022.

\bibitem{HZ_CG_DESCENT2}
W.~W. Hager and H.~Zhang, ``Algorithm 851: {CG}\_{DESCENT}, a conjugate
  gradient method with guaranteed descent,'' \emph{ACM Trans. Math. Softw.},
  vol.~32, no.~1, p. 113–137, 2006.

\bibitem{HZ_CGdecent1}
W.~W. Hager and H.~Zhang, ``A new conjugate gradient method with guaranteed
  descent and an efficient line search,'' \emph{SIAM J. Optim.}, vol.~16,
  no.~1, pp. 170--192, 2005.

\bibitem{SuttiRieHZLS}
M.~Sutti and B.~Vandereycken, ``Riemannian multigrid line search for low-rank
  problems,'' \emph{SIAM J. Sci. Comput.}, vol.~43, pp. A1803--A1831, 2021.

\bibitem{SuttiRMGLSPhDthis}
M.~Sutti, ``Riemannian algorithms on the stiefel and the fixed-rank manifold,''
  Thesis, Universit\'e de Geneve, 2020.

\bibitem{RingRCGframeWork}
W.~Ring and B.~Wirth, ``Optimization methods on {Riemannian} manifolds and
  their application to shape space,'' \emph{SIAM J. Optim.}, vol.~22, pp.
  596--627, 2012.

\bibitem{VandereyckenFixrankMani_CG}
B.~Vandereycken, ``Low-rank matrix completion by {Riemannian} optimization,''
  \emph{SIAM J. Optim.}, vol.~23, no.~2, pp. 1214--1236, 2013.

\bibitem{Sakai_HZRCG}
H.~Sakai and H.~Iiduka, ``Sufficient descent {Riemannian} conjugate gradient
  methods,'' \emph{J. Optim. Theory Appl.}, vol. 190, no.~1, pp. 130--150,
  2021.

\bibitem{pmlr-v70-jin17a}
C.~Jin, R.~Ge, P.~Netrapalli, S.~M. Kakade, and M.~I. Jordan, ``How to escape
  saddle points efficiently,'' in \emph{Int. Conf. Mach. Learn., ICML},
  vol.~70.\hskip 1em plus 0.5em minus 0.4em\relax PMLR, 2017, pp. 1724--1732.

\bibitem{GluntEmbeddEDM}
W.~Glunt, T.~L. Hayden, and W.-M. Liu, ``The embedding problem for predistance
  matrices,'' \emph{Bull. Math. Biol.}, vol.~53, no.~5, pp. 769--796, 1991.

\bibitem{ZhangPrecondGD}
J.~Zhang, S.~Fattahi, and R.~Y. Zhang, ``Preconditioned gradient descent for
  over-parameterized nonconvex matrix factorization,'' \emph{Adv. neural inf.
  proces. syst.}, vol.~34, pp. 5985--5996, 2021.

\bibitem{XuScaledGD}
X.~Xu, Y.~Shen, Y.~Chi, and C.~Ma, ``The power of preconditioning in
  overparameterized low-rank matrix sensing,'' \emph{arXiv preprint
  arXiv:2302.01186}, 2023.

\bibitem{manopt}
N.~Boumal, B.~Mishra, P.-A. Absil, and R.~Sepulchre, ``{M}anopt, a {M}atlab
  toolbox for optimization on manifolds,'' \emph{J. Mach. Learn. Res.},
  vol.~15, no.~42, pp. 1455--1459, 2014.

\bibitem{MaRPCA_RMCSurvey}
S.~Ma and N.~S. Aybat, ``Efficient optimization algorithms for robust principal
  component analysis and its variants,'' \emph{Proc. IEEE}, vol. 106, no.~8,
  pp. 1411--1426, 2018.

\bibitem{CandesRPCCA}
E.~J. Cand\`es, X.~Li, Y.~Ma, and J.~Wright, ``Robust principal component
  analysis?'' \emph{J. ACM}, vol.~58, no.~3, pp. 1--37, 2011.

\bibitem{ShenLMaFit}
Y.~Shen, Z.~Wen, and Y.~Zhang, ``Augmented lagrangian alternating direction
  method for matrix separation based on low-rank factorization,'' \emph{Optim.
  Methods Software}, vol.~29, no.~2, pp. 239--263, 2014.

\bibitem{pmlr-v54-xu17a}
Z.~Xu, M.~Figueiredo, and T.~Goldstein, ``Adaptive {ADMM} with spectral penalty
  parameter selection,'' in \emph{Proc. Mach. Learn. Res.}, vol.~54, 2016, pp.
  718--727.

\bibitem{HuTNNR}
Y.~Hu, D.~Zhang, J.~Ye, X.~Li, and X.~He, ``Fast and accurate matrix completion
  via truncated nuclear norm regularization,'' \emph{IEEE Trans. Pattern Anal.
  Mach. Intell.}, vol.~35, no.~9, pp. 2117--2130, 2013.

\bibitem{cvx}
M.~Grant and S.~Boyd, ``{CVX}: {M}atlab software for disciplined convex
  programming, version 2.1,'' \url{http://cvxr.com/cvx}, Mar. 2014.

\bibitem{chuHessianSStress}
D.~Chu, H.~Brown, and M.~Chu, ``On least squares {Euclidean} distance matrix
  approximation and completion.'' Department of Mathematics, North Carolina
  State University, Tech. Rep., 2003.

\bibitem{GeNIPS2016_7fb8ceb3}
R.~Ge, J.~D. Lee, and T.~Ma, ``Matrix completion has no spurious local
  minimum,'' in \emph{Adv. neural inf. proces. syst.}, vol.~29, 2016, pp.
  2981--2989.

\bibitem{LichtenbergTasissaEDMC}
S.~Lichtenberg and A.~Tasissa, ``A dual basis approach to multidimensional
  scaling,'' \emph{Linear Algebra Appl.}, vol. 682, pp. 86--95, 2024.

\bibitem{TroppUserFriendMCI}
J.~A. Tropp, ``User-friendly tail bounds for sums of random matrices,''
  \emph{Found. Comput. Math.}, vol.~12, no.~4, pp. 389--434, 2012.

\bibitem{LeeUnify_CS_GenTransDom}
K.~Lee, Y.~Li, K.~H. Jin, and J.~C. Ye, ``Unified theory for recovery of sparse
  signals in a general transform domain,'' \emph{IEEE Trans. Inf. Theory},
  vol.~64, no.~8, pp. 5457--5477, 2018.

\bibitem{ChenJNonCVX_rec_IOOcm}
J.~Chen, D.~Liu, and X.~Li, ``Nonconvex rectangular matrix completion via
  gradient descent without $l_{2,\infty}$ regularization,'' \emph{IEEE Trans.
  Inf. Theory}, vol.~66, no.~9, pp. 5806--5841, 2020.

\bibitem{VuRandomGraph}
V.~A.~N. Vu, ``A simple svd algorithm for finding hidden partitions,''
  \emph{Comb. Probab. Comput.}, vol.~27, no.~1, pp. 124--140, 2018.

\bibitem{ChenLRMC_ErrorsErasures}
Y.~Chen, A.~Jalali, S.~Sanghavi, and C.~Caramanis, ``Low-rank matrix recovery
  from errors and erasures,'' \emph{IEEE Trans. Inf. Theory}, vol.~59, no.~7,
  pp. 4324--4337, 2013.

\bibitem{CaiSPHankelMC_PGD}
J.-F. Cai, T.~Wang, and K.~Wei, ``Spectral compressed sensing via projected
  gradient descent,'' \emph{SIAM J. Optim.}, vol.~28, no.~3, pp. 2625--2653,
  2018.

\bibitem{SoltanolBreakSamBarrier}
M.~Soltanolkotabi, ``Structured signal recovery from quadratic measurements:
  {B}reaking sample complexity barriers via nonconvex optimization,''
  \emph{IEEE Trans. Inf. Theory}, vol.~65, no.~4, pp. 2374--2400, 2019.

\bibitem{ChenRobust_CS_EMC}
Y.~Chen and Y.~Chi, ``Robust spectral compressed sensing via structured matrix
  completion,'' \emph{IEEE Trans. Inf. Theory}, vol.~60, no.~10, pp.
  6576--6601, 2014.

\end{thebibliography}

\vfill

\end{document}